\documentclass[12pt]{article}
\usepackage{jheppub} 
\usepackage[T1]{fontenc}
\usepackage{lmodern}
\usepackage{microtype}
\usepackage{amsmath,amssymb,amsfonts,mathtools,bbm}
\usepackage{bm}
\usepackage{bbold}
\usepackage{physics}
\usepackage{siunitx}
\usepackage{xcolor}
\usepackage{booktabs}
\usepackage{url}
\usepackage{hyperref}
\usepackage{tikz,pgfplots}
\pgfplotsset{compat=1.18}
\usepgfplotslibrary{groupplots}
\usepgfplotslibrary{fillbetween}

\newtheorem{theorem}{Theorem}[section]

\newtheorem{lemma}[theorem]{Lemma}

\newcommand{\Z}{\mathbb{Z}}
\newcommand{\R}{\mathbb{R}}

\newcommand{\slashed}[1]{\mathrlap{\!\not{\phantom{#1}}}#1}

\title{Discrete \texorpdfstring{$\theta$}{theta} Projection: 
A Gauge-Protected Solution to the Strong CP Problem Without Axions}

\author[a,b]{Sameer Ahmad Mir}
\emailAdd{sameerphst@gmail.com}
\affiliation[a]{Canadian Quantum Research Center, 460 Doyle Ave 106, Kelowna, BC V1Y 0C2, Canada}
\affiliation[b]{Department of Computer Sciences, Asian School of Business, Noida, Uttar Pradesh, 201303, India}
\author[c]{Bobby Eka Gunara}
\emailAdd{bobby@itb.ac.id}
\affiliation[c]{Theoretical High Energy Physics Research Division,
Faculty of Mathematics and Natural Sciences, Institut Teknologi Bandung,
Jl. Ganesha no. 10 Bandung, 40132 Indonesia}
\author[a,d,e,f]{Mir Faizal}
\emailAdd{mirfaizalmir@googlemail.com}
\affiliation[d]{Irving K. Barber School of Arts and Sciences, University of British Columbia Okanagan, Kelowna, BC V1V 1V7, Canada}
\affiliation[e]{Department of Mathematical Sciences, Durham University, Upper Mountjoy, Stockton Road, Durham DH1 3LE, UK}
\affiliation[f]{Computational Mathematics Group, Hasselt University, Agoralaan Gebouw D, Diepenbeek, 3590 Belgium}

\abstract{

We address the strong CP problem: why the physical QCD angle theta-bar must be extraordinarily small given the stringent bounds on the neutron electric dipole moment. Peccei-Quinn axion models can relax theta-bar dynamically, but rely on an approximate global symmetry expected to be violated by quantum gravity and face severe astrophysical and cosmological constraints. We propose Discrete theta Projection, an axionless, gauge-protected resolution obtained by gauging a finite cyclic subgroup $Z_N $of the $2\pi$ shift symmetry of theta. Coupling QCD to a compact, local and gapped topological sector orbifolds the path integral, identifying theta values that differ by $2\pi/N$ and admitting only instanton sectors whose topological charge lies in $Z_N$. In the large four-volume limit the vacuum energy becomes the lower envelope of the orbifold images, so the theory dynamically selects the branch closest to the CP-symmetric point, enforcing $|\bar{\theta}| \le \pi/N$ without assuming any prior smallness. Because the discrete shift is gauged, continuous renormalization of theta is forbidden; the construction can be formulated via higher-form/two-group structure with integer-quantized couplings fixed by anomaly inflow, ensuring radiative and gravitational stability and satisfying mixed gauge-gravity consistency conditions. The framework predicts a neutron EDM suppressed by $1/N$, no axion signatures, no domain-wall/isocurvature issues, and lattice diagnostics: piecewise-analytic theta dependence with cusps at odd fractions of the reduced period and a global curvature scaling as $1/N^2$. We provide the EFT construction, a nonperturbative proof of vacuum projection, a full anomaly analysis, and UV embeddings (including discrete clockwork chains) that generate large effective N while preserving integrality and consistency throughout.

}

\begin{document}
\maketitle
\flushbottom

\section{Introduction}
The physical strong-CP angle is defined as the sum of the bare gauge-theory parameter and the CP-violating phase inherited from the quark mass matrix. Concretely, it is the fundamental angle plus the argument of the determinant of the product of the up- and down-quark Yukawa matrices. Its unnaturally small empirical value is enforced by the nonobservation of a neutron electric dipole moment at current sensitivities~\cite{Pospelov:2005pr,Abel:2020pzs,Crewther:1979pi,Baluni:1978rf}. In a generic low-energy effective field theory one expects an angle of order unity. Such an angle would induce a dipole moment orders of magnitude above the experimental bound. The question of why QCD is so close to CP conservation is therefore both sharp and pressing~\cite{Pospelov:2005pr,Kim:2008hd}. 
The standard approach invokes the Peccei-Quinn mechanism. One posits a spontaneously broken global symmetry whose associated axion dynamically relaxes the effective angle to zero \cite{Peccei:1977hh,Weinberg:1977ma,Wilczek:1977pj,Kim:2008hd}. While elegant in principle, this direction is encumbered by familiar weaknesses that have proven hard to eliminate. These include the axion-quality problem, the expectation that quantum gravity does not tolerate exact global symmetries and thus generically introduces explicit symmetry breaking, and stringent astrophysical and cosmological constraints on light axion-like states \cite{Holman:1992us,Kamionkowski:1992mf,Barr:1992qq,Banks:2010zn,Harlow:2018tng,Graham:2015ouw,Raffelt:2006cw,Wantz:2009it}. These issues motivate a conceptually orthogonal solution that keeps the strong dynamics of QCD intact and replaces a fragile global symmetry with a principle that quantum gravity is expected to respect \cite{Banks:2010zn,Harlow:2018tng}. 
The central idea developed in this work is Discrete $\theta$ Projection. One gauges a finite cyclic subgroup of the periodic shift of the vacuum angle by coupling QCD to a compact topological sector. In the path integral this coupling implements an orbifolding that identifies angles separated by an integer fraction of the fundamental period and sums over the corresponding images \cite{Dvali:2005an}. In the thermodynamic limit the vacuum energy density becomes the lower envelope of the image branches. The theory dynamically selects the branch whose angle is closest to the CP-symmetric point. This selection nonperturbatively confines the physical angle to lie within a principal cell of width set by the order of the gauged subgroup. 
The suppression is not an assumption about initial conditions or a perturbative tuning. It is a consequence of vacuum competition at infinite four-volume and of the convex, even, and periodic structure of the pure-gauge energy anticipated from large-color arguments and the classic analyses of $\theta$ dependence. In those analyses a tower of nearly degenerate branches cross at special angles and produce cusp singularities in the energy that encode level repulsion among metastable vacua \cite{Veneziano:1979ec,DiVecchia:1980yfw,Witten:1998uka}. The conceptual novelty lies in enforcing the small physical angle through a discrete gauge symmetry rather than a continuous global one. 
The compact topological sector can be formulated as a local, gapped effective theory of higher-form gauge fields with quantized couplings. As a result, it carries no propagating degrees of freedom and introduces no axion phenomenology. Its sole role is to impose a gauge identification of $\theta$-configurations and a global selection rule on topological sectors, under which only instanton numbers compatible with the gauged subgroup contribute \cite{Hsin:2020nts,Dvali:2005an,Freed:2014iua}. Because the symmetry is gauged and discrete, quantum-gravity corrections respect it rather than spoil it. The integrality of the topological data is fixed by anomaly inflow and the modern framework of generalized symmetries and their two-group refinements. These structures guarantee that continuous renormalization of the angle is forbidden and that any allowed nonperturbative shift is an integer multiple of the orbifold step \cite{Gaiotto:2014kfa,Cordova:2019uob,Krauss:1988zc}.

The resulting projection replaces a light axion with heavy, nonpropagating topological defects. These are membrane-like excitations that are charged under the discrete gauge sector. Their tensions grow with the order of the subgroup, and their nucleation is exponentially suppressed \cite{Brown:1987dd,Bousso:2000xa,Kaloper:2025wgn,Bandos:2019wgy}. This behavior eliminates the domain-wall and isocurvature problems that complicate axion cosmology \cite{Sikivie:2006ni,Marsh:2015xka,DiLuzio:2020wdo,Kawasaki:2013iha}, while preserving exact gauge protection in the infrared and ultraviolet alike. 
Phenomenologically the mechanism is highly predictive. The chiral relation between the neutron electric dipole moment and the physical angle implies a scaling in which the projected theory yields an upper bound set by a numerical prefactor multiplied by the inverse of the gauged order. Present limits therefore translate directly into the minimal order required for compatibility, and continued experimental progress maps to a clear target range without reference to additional model-dependent structure \cite{DiVecchia:1980yfw,Engel:2013lsa,Chupp:2017rkp,Shindler:2021bcx,Dar:2000tn,Yoon:2017tag}. 
The global $\theta$ curvature that governs the topological susceptibility over a full period is suppressed by the square of the gauged order. The energy differences between adjacent branches shrink accordingly. As a result, the envelope becomes extremely flat near the origin, even though the underlying microscopic branch retains the familiar small-angle curvature of pure Yang-Mills. Together these facts deliver testable predictions for the distribution of topological charge and the shape of the vacuum energy on the lattice, including characteristic cusps at odd multiples of the reduced fundamental spacing and a universal suppression of the curvature about the principal cell \cite{Witten:1998uka,DelDebbio:2004ns,DElia:2003zne,Alles:2004vi,Vicari:2008jw,Bennett:2022gdz,Giusti:2007tu}. 
Ultraviolet completions that generate large effective orders arise naturally in discrete clockwork chains of compact one- and two-form sectors. In these chains, modest microscopic inputs multiply into exponentially large effective identifications while preserving all anomaly constraints and the integrality required by inflow. These architectures demonstrate that the projection can be engineered within controlled local EFTs and embedded consistently into broader gauge-theory frameworks \cite{Gaiotto:2014kfa,Cordova:2019uob,Kaplan:2015fuy,Giudice:2016yja,Ahmed:2016viu,Gushterov:2018spg,Batell:2010bp}. 
It is important to emphasize the scope of what we are and are not claiming. The mechanism proposed here is aimed squarely at the strong-CP problem, i.e. at explaining the smallness of the QCD angle $\bar\theta$. It leaves the Kobayashi-Maskawa phase and weak CP violation in the Standard Model unchanged, and all CKM-induced contributions to observables such as the neutron EDM remain as in the usual theory. 
Moreover, throughout this work we adopt the perspective of a controlled four-dimensional effective field theory for QCD coupled to a compact topological sector. This EFT is supplemented by simple classes of topological UV realizations, such as the discrete clockwork chains of Sec.~\ref{sec:LargeN-UV}, that accommodate large effective $N$ and the required anomaly inflow. We do not attempt to construct a fully unified Standard Model plus gravity completion. Rather, our goal is to demonstrate that, under the usual assumptions about discrete gauge symmetries in quantum gravity, the discrete $\theta$ projection can be implemented consistently, with radiative stability and mixed gauge-gravity anomalies under control.

The present paper develops this gauge-protected solution in detail. We begin from a local effective description of the topological sector coupled to QCD. We then derive the exact projection bound and the envelope structure in the thermodynamic limit. Next, we establish radiative and gravitational stability within the generalized-symmetry and inflow formalism. We analyze anomaly and consistency conditions, including mixed gauge-gravity couplings. We extract the phenomenological and cosmological consequences, with an emphasis on neutron-EDM predictions and lattice diagnostics. We then outline ultraviolet completions, including discrete clockwork realizations of large effective order, and connect these constructions to practical parameter ranges relevant for current and near-future experimental and numerical tests \cite{Witten:1998uka,DiVecchia:1980yfw,Veneziano:1979ec,Gaiotto:2014kfa,Cordova:2019uob,Krauss:1988zc}.

\section{Effective Description: Gauging a Discrete $\theta$-Shift}
In this section, we develop a fully explicit effective description of the theory obtained by gauging a discrete $\theta$-shift symmetry. We consider a gauge theory (for definiteness, a pure $SU(N)$ Yang-Mills theory) with the standard topological $\theta$ term. The action includes 
\begin{equation}
S_{\theta} \;=\; \frac{i\,\theta}{8\pi^2} \int \Tr(F\wedge F)\,,
\label{eq:Stop}
\end{equation} 
where $F$ is the $SU(N)$ field strength and $\Tr(F\wedge F)$ is the color topological charge density. Here 
\begin{equation}
Q \;\equiv\; \frac{1}{8\pi^2}\int \Tr(F\wedge F) \;\in\; \mathbb{Z}
\end{equation} 
is the integer-valued instanton number (topological charge) \cite{Witten:1998uka,Vicari:2008jw}. The exponential of $S_{\theta}$ in the Euclidean path integral is $\exp(i\theta Q)$, which is invariant under the usual $2\pi$ periodicity $\theta\sim\theta+2\pi$ because $\exp(i\,2\pi Q)=1$ for any integer $Q$. We now suppose that the theory has a discrete shift symmetry $\mathbb{Z}_N$ acting on the $\theta$ angle by 
\begin{equation}
\theta \;\mapsto\; \theta + \frac{2\pi}{N}\,. 
\end{equation} 
Such discrete shifts fit naturally into the modern framework of generalized global and higher-form symmetries~\cite{Gaiotto:2014kfa,Hsin:2020nts}. 
Our goal is to gauge this $\mathbb{Z}_N$ symmetry, making $\theta$-shifts by $2\pi/N$ equivalent to the identity. This is implemented by orbifolding the $\theta$-parameter space (the circle $0\le\theta<2\pi$) by identifying $\theta\sim \theta+2\pi/N$~\cite{Cordova:2019uob,Dvali:2005an}. Physically, gauging the discrete shift means that configurations which differ by a $\theta$-increment of $2\pi/N$ (accompanied by an appropriate large gauge transformation on the gauge fields, as discussed below~\cite{Cordova:2019uob,Hsin:2020nts}) are treated as the same physical configuration. The result is that the fundamental domain of $\theta$ is reduced from $[0,2\pi)$ to $[0,2\pi/N)$, i.e. the $\theta$-circle is cut into $N$ equivalent segments. This $\mathbb{Z}_N$ orbifold of the $\theta$-circle is depicted schematically in Fig.~\ref{fig:orbifold}. 

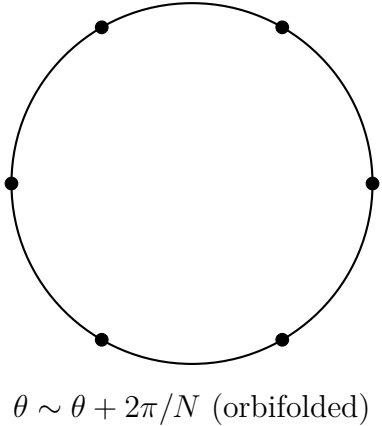
\begin{figure}[h]
\centering
\begin{tikzpicture}[scale=1.2]
\draw[thick] (0,0) circle (2cm);
\foreach \x in {0,60,120,180,240,300}{
  \filldraw[black] ({2*cos(\x)},{2*sin(\x)}) circle (2pt);
}
\node at (0,-2.5){$\theta \sim \theta + 2\pi/N$ (orbifolded)};
\end{tikzpicture}
\caption{Orbifolding of the $\theta$-circle under discrete shift symmetry. Points separated by $2\pi/N$ along the original $\theta$-circle are identified, yielding a smaller fundamental domain.}
\label{fig:orbifold}
\end{figure}

\subsection{Orbifolding the $\theta$-Circle}
To understand the consequences of identifying $\theta\sim\theta+2\pi/N$, let us examine how the path integral and topological charge sectors are affected. 
Under a shift $\theta\to \theta+\frac{2\pi}{N}$, the weight $\exp(i\theta Q)$ in the path integral transforms as 
\begin{equation}
\exp(i\theta Q)\;\;\to\;\;\exp\!\Big[i\Big(\theta+\frac{2\pi}{N}\Big)Q\Big] \;=\; \exp(i\theta Q)\,\exp\!\Big(\frac{i2\pi Q}{N}\Big)\,. 
\end{equation} 
For this to represent a symmetry of the theory, the extra phase $\exp(i2\pi Q/N)$ must not change the physics. 
However, since $Q$ is an integer, $\exp(i2\pi Q/N)$ is not generally unity, except in the special case that $Q$ is a multiple of $N$.\footnote{For classic discussions of instanton number and $\theta$ dependence in Yang-Mills, see for example \cite{Belavin:1975fg,tHooft:1976snw}.} 
This shows that a shift of $\theta$ by $2\pi/N$ is not a symmetry of the original theory for generic configurations. 
Rather, it changes the weight by a nontrivial phase. 
In order to gauge this transformation, one must introduce additional degrees of freedom that compensate for the phase. 
Equivalently, one must sum over contributions related by the $\mathbb{Z}_N$ shift so that the phases interfere and enforce the identification. 
In practice, gauging the $\mathbb{Z}_N$ means we project the path integral onto $\mathbb{Z}_N$-invariant contributions. 

A convenient way to implement this projection is to insert a suitable operator into the path integral that averages over $\theta$-shifts. 
Concretely, we introduce the operator 
\begin{equation}
P_{\mathbb{Z}_N} \;\equiv\; \frac{1}{N}\sum_{m=0}^{N-1} U^m\,,
\end{equation} 
where $U$ is the unitary operator that shifts $\theta \to \theta + 2\pi/N$ (accompanied by the appropriate large gauge transformation on $A$ that we will discuss momentarily). 
Acting on any configuration, $U$ multiplies the path-integral weight by the phase $\exp(i2\pi Q/N)$, so $U^m$ contributes a factor $\exp(i2\pi mQ/N)$. 
Inserting $P_{\mathbb{Z}_N}$ into the partition function and using $\int \mathcal{D}A$ to denote the functional integral over gauge fields, we obtain the gauged partition function as 
\begin{equation}
Z_{\text{gauged}}(\theta) \;=\; \int \mathcal{D}A \; P_{\mathbb{Z}_N} \; e^{-S_{\text{YM}}[A] + i\theta Q[A]}\,,
\label{eq:ZgaugedDef}
\end{equation}
where $S_{\text{YM}}$ is the Yang-Mills action (without the $\theta$ term). 
This sum enforces $\theta$-shift invariance in direct analogy with the standard gauging of discrete and higher-form symmetries \cite{Gaiotto:2014kfa,Hsin:2020nts}. 
Indeed, plugging in the definition of $P_{\mathbb{Z}_N}$, we can rewrite \eqref{eq:ZgaugedDef} as 
\begin{equation}
Z_{\text{gauged}}(\theta) \;=\; \frac{1}{N}\sum_{m=0}^{N-1}\;\int \mathcal{D}A\; \exp\!\Big[-S_{\text{YM}}[A] + i\theta Q[A] + \frac{i\,2\pi m}{N}Q[A]\Big] . 
\label{eq:Zgauged}
\end{equation}
Interchanging the sum and the path integral, the factor $\frac{1}{N}\sum_{m=0}^{N-1}e^{i2\pi mQ/N}$ acts as a projector onto those field configurations for which $e^{i2\pi Q/N}=1$. 
In fact, as a distribution over integer $Q$, we have 
\begin{equation}
\frac{1}{N}\sum_{m=0}^{N-1}\exp\!\Big(\frac{i2\pi m}{N}Q\Big) \;=\; 
\begin{cases}
1 & (Q\;\equiv\;0 \pmod N)\,,\\[6pt]
0 & (Q\not\equiv 0 \pmod N)\,. 
\end{cases}
\end{equation} 
Thus, the only gauge-field configurations that survive in \eqref{eq:Zgauged} are those whose topological charge $Q$ is a multiple of $N$. 
In other words, gauging the $\mathbb{Z}_N$ shift symmetry partitions the path integral into sectors where $Q$ is effectively defined modulo $N$, and it retains only those sectors with $Q/N \in \mathbb{Z}$. 
Configurations whose instanton number $Q$ is not divisible by $N$ cancel out via destructive interference in the sum \eqref{eq:Zgauged}. 

This result has several important implications. 
First, the physical $\theta$ angle of the gauged theory has an effective $2\pi/N$ periodicity. 
In fact, once only $Q$ divisible by $N$ contribute, shifting $\theta$ by $2\pi/N$ changes the weight by $\exp[i(2\pi/N) Q] = \exp(i2\pi k)$ for $Q=N k$, which is unity. 
Therefore $\theta$ and $\theta+2\pi/N$ are now truly equivalent in the gauged theory. 
We can define a new angle $\theta_{\mathrm{eff}}$ to parametrize the orbifolded $\theta$-circle. 
Let 
\begin{equation}
\theta_{\mathrm{eff}} \;\equiv\; N\,\theta\,,
\label{eq:theta_eff}
\end{equation} 
so that when $\theta$ ranges over $[0,2\pi/N)$, the “effective” angle $\theta_{\mathrm{eff}}$ ranges over $[0,2\pi)$. 
In terms of $\theta_{\mathrm{eff}}$, the topological term \eqref{eq:Stop} can be written as 
\begin{equation}
S_{\theta} \;=\; \frac{i\,\theta_{\mathrm{eff}}}{2\pi\,N} \int \Tr(F\wedge F) \;=\; \frac{i\,\theta_{\mathrm{eff}}}{2\pi} \int \Tr(F\wedge F/N)\,.
\end{equation} 
But since $\int \Tr(F\wedge F)=2\pi N k$ on any surviving configuration ($Q=Nk$), we have 
\begin{equation}
S_{\theta} \;=\; \frac{i\,\theta_{\mathrm{eff}}}{2\pi} \cdot 2\pi k \;=\; i\,\theta_{\mathrm{eff}}\,k\,,
\end{equation} 
showing that $\theta_{\mathrm{eff}}$ couples to the integer $k=Q/N$ just like an ordinary $\theta$ angle would couple to $Q$. 
In this sense, $\theta_{\mathrm{eff}}$ is the ``un-orbifolded'' angle (covering $0$ to $2\pi$) while $\theta$ is restricted to the smaller fundamental domain. 
The orbifolding thus does not eliminate the topological term but rather quantizes $\theta$ in units of $2\pi/N$. 
Equivalently, one can say that the original $\theta$ is now effectively an angular variable with only $N$ allowed values. 
When $N$ is large, these values appear quasi-continuous, but only increments of $2\pi/N$ are meaningful. 

Second, because only $Q$ divisible by $N$ appears, instantons effectively come in bundles of $N$. 
If the original theory had an instanton of action $S_{\text{inst}}$ (with $Q=1$), in the gauged theory a configuration carrying one unit of $Q$ is not gauge-invariant by itself. 
Only $N$ instantons taken together (total $Q=N$) form an invariant configuration. 
In a loose sense, the discrete gauging introduces a $\mathbb{Z}_N$ ``selection rule'' that forbids isolated single-instanton events. 
They must either come in multiples of $N$ or be connected to other processes that ensure overall $Q$ is a multiple of $N$. 
We will see this reflected in the low-energy effective descriptions below. 

Finally, let us elaborate on the required accompanying transformation on the gauge fields, which we glossed over above. 
The shift $\theta\to \theta+2\pi/N$ by itself is merely a change of a coupling constant. 
However, in a 4$d$ gauge theory, one may consider a large gauge transformation that is not contractible to the identity. 
Such a transformation can shift the instanton number by an integer. 
In particular, $\pi_3(SU(N))=\mathbb{Z}$, so there exist gauge transformations $g$ of winding number $n$ such that under $A \mapsto A^g$ one finds $Q[A^g] = Q[A] + n$. 
Under such a transformation, the $\theta$ term in the action changes by 
\begin{equation}
\Delta S_{\theta} \;=\; \frac{i\,\theta}{2\pi}\int \Tr(F\wedge F)\Big|_{\Delta Q = n} \;=\; i\,\theta\,n\,.
\end{equation} 
If we now perform a combined operation, we shift $\theta$ by $-\,2\pi/N$ and perform a large gauge transformation of winding $+1$. 
The net change in the action is 
\begin{equation}
\Delta S_{\theta}^{\text{(net)}} \;=\; i\,(\theta - \tfrac{2\pi}{N})(Q+1)\;-\;i\,\theta\,Q \;=\; i\,\theta - \frac{i2\pi}{N}(Q+1)\,,
\end{equation} 
so the path-integral weight is multiplied by $\exp(i\theta)\exp(-i\frac{2\pi}{N}Q)\exp(-i\frac{2\pi}{N})$. 
For generic $\theta$ and $Q$ this is not unity. 
However, if we consider the combined operation repeated $N$ times, then $\theta$ shifts by $-2\pi$ (returning to its original value, since $\theta$ is $2\pi$-periodic), and the winding number adds up to $N$. 
The total phase after $N$ repetitions is $\exp(iN\theta)\exp(-i2\pi Q)\exp(-i2\pi) = \exp(iN\theta)$, which is $1$ if $\theta$ takes one of the discrete values $\theta = 2\pi k/N$ (for some integer $k$). 
In other words, at those quantized values of $\theta$, a $\mathbb{Z}_N$ subgroup of the large gauge transformation group becomes a true symmetry of the theory. 
The combination of a $2\pi/N$ shift in $\theta$ and a unit winding increase in $Q$ leaves the action invariant \cite{Tanizaki:2018wtg}. 
By coupling the theory to a suitable topological background field, we effectively enforce $\theta$ to take those discrete values dynamically. 
This makes the $\mathbb{Z}_N$ shift a gauge symmetry rather than a global one \cite{Cordova:2019uob,Freed:2014iua,Freed:2016rqq}. 
The path-integral construction above using $P_{\mathbb{Z}_N}$ is essentially a direct way to implement this. 
It sums over sectors related by $Q\to Q+1$ (the large gauge transformation) weighted by the appropriate phase to enforce invariance under $\theta\to\theta+2\pi/N$. 
In summary, gauging the discrete $\theta$-shift can be thought of as orbifolding the theory by a combined operation $(\theta\to\theta+2\pi/N,\;Q\to Q+1)$. 

Having established the orbifold projection conditions on the $\theta$-circle and the resultant selection rule $Q\in N\mathbb{Z}$, we proceed to describe two complementary effective field theory (EFT) realizations of this scenario. 
We will refer to these as \emph{Model A} and \emph{Model B} for convenience. 
Both models incorporate the gauged $\mathbb{Z}_N$ shift symmetry but in different ways. 
Model~A encodes the discrete $\theta$-shift through a modification of the gauge theory's global structure and the introduction of discrete topological fields (higher-form gauge fields). 
Model~B introduces a new dynamical axion field to represent the $\theta$ angle and realizes the $\mathbb{Z}_N$ shift as a symmetry of an axion potential. 
We present each model in turn, detailing their assumptions, field content, symmetry structure, topological term derivations, and the implications for anomalies and domain walls.

\subsection{Topological Sector (EFT Realizations)}
\noindent\textbf{Model A: Discrete topological gauging and higher-form symmetry.}\quad 
In Model~A, we realize the gauging of the $\theta$-shift by coupling the gauge theory to a discrete topological field that implements the $\mathbb{Z}_N$ identification of $\theta$. 
One way to think of this is to consider changing the global form of the gauge group. 
Instead of an $SU(N)$ gauge theory, we consider an $SU(N)/\mathbb{Z}_N$ gauge theory~\cite{Aharony:2013hda}. 
The $SU(N)/\mathbb{Z}_N$ theory can be viewed as the $SU(N)$ theory in which the $\mathbb{Z}_N$ center one-form symmetry~\cite{Gaiotto:2014kfa} has been gauged. 
A crucial effect of this is that the instanton number $Q$ is no longer an unconstrained integer. 
Rather, only the congruence class of $Q$ modulo $N$ is invariantly defined once the center is gauged. 
In fact, $SU(N)/\mathbb{Z}_N$ bundles on general 4-manifolds admit fractional topological charge: $Q$ can be an integer multiple of $1/N$ in such a theory, and configurations are classified by two invariants $(Q \bmod 1,\;w_2)$, where $w_2\in H^2(M,\mathbb{Z}_N)$ is a second Stiefel-Whitney class specifying the obstruction to lifting the $SU(N)/\mathbb{Z}_N$ bundle to an $SU(N)$ bundle~\cite{Witten:1997bs,Razamat:2013opa}. 
In our present context (pure Yang-Mills in flat $\mathbb{R}^4$ or $S^4$), $w_2$ may be taken trivial, but the effect of gauging the center remains. 
Adding or removing an $SU(N)$ instanton of charge $Q=1$ is no longer an allowed process by itself. 
One can only add $N$ instantons together in a combination that yields $Q=1$ mod 1 (i.e. an integer $Q$). 
This is entirely consistent with our orbifold analysis above, which forced $Q$ to be a multiple of $N$ for gauge-invariant contributions. 
In the $SU(N)/\mathbb{Z}_N$ description, this constraint can be traced to the presence of a discrete $\theta$-term, often called a discrete $\theta$ angle or Neumann-Rosenzweig-Thorn (NRT) phase in the literature~\cite{Kapustin:2013qsa}. 

One can introduce a discrete $\theta$ angle $p$ (taking values in $\{0,1,\dots,N-1\}$) for an $SU(N)/\mathbb{Z}_N$ theory, which weights different topological sectors differently~\cite{Hsin:2020nts}. 
For example, in an $SU(2)/\mathbb{Z}_2$ (i.e. $SO(3)$) gauge theory, one finds two physically distinct theories. 
The case $p=0$ corresponds to the theory with no discrete $\theta$ term, and $p=1$ corresponds to a theory with a $\mathbb{Z}_2$ topological angle that assigns a $(-1)$ sign to configurations with odd $Q$~\cite{Tanizaki:2018wtg}. 
More generally, the discrete $\theta$ angle $p$ can be thought of as a coupling to the $\mathbb{Z}_N$-valued second Stiefel-Whitney class $w_2$. 
Schematically, the Euclidean action can include a term 
\begin{equation}
S_{p} \;=\; i\pi p\,\frac{(w_2)^2}{N} \,,
\label{eq:Snrt}
\end{equation} 
which is a $\mathbb{Z}_N$-valued topological action (in four dimensions $(w_2)^2$ is related to the second Chern class mod $N$)~\cite{Cordova:2019uob,Kapustin:2013qsa}. 
Although \eqref{eq:Snrt} is a highly abstract way to encode the effect, one can understand its consequence. 
When $p\neq 0$, sectors whose $SU(N)$ instanton number differs by a nonzero amount modulo $N$ acquire a relative phase $\exp(i\pi p/N)$ and hence interfere destructively unless $p$ is tuned so that this phase is trivial. 
In particular, for $p=1$ in an $SU(N)/\mathbb{Z}_N$ theory, an instanton and a would-be “$N$th root” of an instanton have different weights, and the path integral sums over them with a relative phase. 
The end result is that only configurations with $Q$ divisible by $N$ contribute constructively. 
This matches precisely our gauged $\mathbb{Z}_N$ shift symmetry: it is effectively the $p=1$ discrete $\theta$ angle in the $SU(N)/\mathbb{Z}_N$ theory. 
In other words, gauging the $\mathbb{Z}_N$ shift in the $SU(N)$ theory produces the $SU(N)/\mathbb{Z}_N$ theory with a discrete topological angle. 

We can formulate Model~A in a more field-theoretic language as well. 
Instead of changing the global form of the gauge group by hand, we can couple the $SU(N)$ theory to a background two-form gauge field $B_{2}$ with gauge group $\mathbb{Z}_N$ (often called a $B_2$-field or discrete $B_2$-field)~\cite{Gaiotto:2014kfa,Cordova:2019uob,Kapustin:2013qsa,Guo:2017xex,Sulejmanpasic:2019ytl}. 
The $B_{2}$ field couples to the gauge theory through a term 
\begin{multline}
S_{\mathrm{top}}(4D)=\frac{iN}{2\pi}\int_{M_4} B_2\wedge f_2
+\frac{i}{2\pi}\int_{M_4} a_1\wedge\Big[\frac{1}{8\pi^2}\,CS_3(A)+\kappa_G\,CS_3(\omega)\Big]
+\frac{i\theta}{8\pi^2}\int_{M_4}\mathrm{Tr}(F\wedge F),
\label{eq:SBcoupling}
\end{multline}
with
\begin{equation}
d\,CS_3(A)=\mathrm{Tr}(F\wedge F),\quad d\,CS_3(\omega)=\mathrm{Tr}(R\wedge R) ,
\end{equation}
which is allowed if $B_{2}$ is a $\mathbb{Z}_N$-valued 2-form (so that $B_{2}$ effectively has period $2\pi/N$). 
Here we are treating $B_{2}$ as a background field to implement the symmetry gauging; one can also consider dynamical $B_{2}$ in a higher-form gauge theory~\cite{Kapustin:2013uxa}. 
The coupling \eqref{eq:SBcoupling} indicates that whenever an instanton number passes through a given surface in spacetime, it can be absorbed by $B_{2}$ flux through that surface. 
In more concrete terms, the $\mathbb{Z}_N$ two-form field $B_{2}$ can screen fractional topological charge. 
This precisely means that $\Tr(F\wedge F)$ is conserved mod $N$ rather than absolutely conserved. 
Equivalently, $Q$ is defined only mod $N$ once $B_{2}$ is coupled. 
The term \eqref{eq:SBcoupling} is essentially a Lagrange multiplier enforcing that selection rule. 
In the presence of \eqref{eq:SBcoupling}, a large gauge transformation of unit winding $n=1$ can be compensated by a shift of the $B_{2}$ field (since $B_{2}$ appears only through $e^{i\int B_{2}\wedge F\wedge F/2\pi}$, shifting $B_{2}\to B_{2} - \frac{1}{N}\delta^{(2)}$ for some quantized 2-form $\delta^{(2)}$ can cancel the phase of an instanton). 
Thus, what was previously a “Theta” parameter $\theta$ becomes to some extent an emergent dynamical variable, since the discrete $B_{2}$ field can soak up changes in $\theta$. 
In fact, the $\mathbb{Z}_N$ shift symmetry of $\theta$ is now promoted to an ordinary $2\pi$-periodic shift symmetry of the $B_{2}$ field (because shifting $\theta$ by $2\pi/N$ and simultaneously $B_{2}$ by an appropriate flat connection is a gauge symmetry). 
In effect, the role of the $\theta$ parameter can be traded for a background $B_{2}$ field, and the discrete $\theta$-quantization is enforced by $B_{2}$ flux.

\medskip\noindent\textbf{Model B: Axion field and degenerate vacua.}\quad 
Model~B introduces a new dynamical field i.e., an axion $a(x)$ to represent the $\theta$ parameter~\cite{Peccei:1977hh,Weinberg:1977ma,Wilczek:1977pj,Kim:2008hd,DiLuzio:2020wdo}. 
Instead of treating $\theta$ as a fixed coupling, we promote it to a field $a(x)$ with its own kinetic term and potential. 
The axion is a $2\pi$-periodic pseudoscalar, meaning that $a(x)$ has a shift symmetry $a \mapsto a + 2\pi$. 
We couple $a$ to the gauge field through the topological term 
\begin{equation}
S_{aF\tilde{F}} \;=\; \frac{i}{2\pi} \int d^4x\; a(x)\,\Tr(F_{\mu\nu}\tilde{F}^{\mu\nu}) \;=\; \frac{i}{2\pi} \int a\,\Tr(F\wedge F)\,,
\label{eq:axioncoupling}
\end{equation}
which is the standard axion coupling to the instanton density in axion Yang-Mills theories~\cite{Brennan:2020ehu}. 
So that, upon integration by parts, $da$ acts as a Lagrange multiplier enforcing $d \!\left(\Tr(F\wedge F)\right)=0$ (instantons cannot disappear without changing $a$). 
If $a(x)$ were a free massless field, \eqref{eq:axioncoupling} would imply that $a$ is constant in space (since any spacetime variation would create a nonzero $\langle \Tr(F\wedge F)\rangle$ by the equations of motion), and its expectation value would behave exactly like the $\theta$ angle. 
However, to obtain a controlled effective field theory, we assume that the axion acquires a potential $V(a)$, typically due to nonperturbative effects (instantons)~\cite{Kim:2008hd,DiLuzio:2020wdo}. 
We choose the potential such that the shift symmetry $a\to a+2\pi$ is broken down to a discrete $\mathbb{Z}_N$ subgroup:
\begin{equation}
V(a)\;=\; -\,\Lambda^4 \cos(N a)\,,
\label{eq:axionpot}
\end{equation}
for some energy scale $\Lambda$. 
This form of the potential is invariant under $a \mapsto a + \frac{2\pi}{N}$ up to the full $2\pi$ periodicity. 
Specifically, $a\to a+2\pi/N$ changes $\cos(Na)$ to $\cos(Na+2\pi) = \cos(Na)$, so the Lagrangian is invariant under the discrete shift 
\begin{equation}
a(x)\;\mapsto\; a(x) + \frac{2\pi}{N}\,. 
\end{equation} 
This is exactly the $\mathbb{Z}_N$ symmetry that we wish to gauge. 
We have thus built a low-energy model where the discrete $\theta$-shift symmetry is manifest as a symmetry of the axion potential~\cite{Hason:2020yqf}. 

Crucially, in the ungauged Model~B, the discrete symmetry $a\to a+2\pi/N$ is a global symmetry (the potential \eqref{eq:axionpot} picks out $N$ equivalent vacua, related by this symmetry). 
Indeed, the minima of $V(a)$ occur at 
\begin{equation}
a = \frac{2\pi k}{N}\,,\qquad k=0,1,2,\dots,N-1\,,
\end{equation} 
all of which satisfy $\cos(Na)=1$. 
These $N$ vacuum states are physically distinct in the ungauged theory. They are labeled by the integer $k$ or, equivalently, by the expectation value of the operator $\exp(i a)$ which serves as an order parameter for the broken $\mathbb{Z}_N$ symmetry. 
The existence of these $N$ degenerate vacua is a classic consequence of a discrete, spontaneously broken symmetry~\cite{Kim:2008hd,DiLuzio:2020wdo}. 

Now, each vacuum $k$ corresponds to a world where the effective $\theta$ angle (given by the expectation value of $a$) is $\langle a \rangle = 2\pi k/N$. 
In other words, the physics in vacuum $k$ is that of the original gauge theory at $\theta = 2\pi k/N$. 
(This can be seen by substituting $a(x)=2\pi k/N$ in \eqref{eq:axioncoupling}, yielding 
\begin{equation}
    S_{aF\tilde{F}} = \frac{i\,(2\pi k/N)}{2\pi}\int \Tr(F\wedge F) = i\theta_k Q
\end{equation}
with $\theta_k = 2\pi k/N$.) 
Thus, the $N$ vacua of the axion potential precisely mirror the $N$ would be vacua of the original theory at $\theta=0,2\pi/N,4\pi/N,\dots,2\pi(N-1)/N$. 
In the ungauged theory, these were in fact all the same vacuum (since $\theta$ was a continuous parameter we could dial, and only at special $\theta$ like $\pi$ did distinct degenerate vacua arise due to other symmetries such as CP). 
But here, by introducing the axion and its potential, we have made each of those $\theta$ angles into a separate vacuum, separated by potential barriers. 

Because the $N$ vacua are degenerate and related by a discrete symmetry, there must exist domain walls interpolating between them. A domain wall in this context is a soliton solution in which the axion field $a(x)$ varies from one vacuum value to another vacuum value across a spatial slice. For instance, a wall between vacuum $k$ (with $a=2\pi k/N$ far on the left) and vacuum $k+1$ (with $a=2\pi (k+1)/N$ far on the right) can be found by solving the field equation for $a(x)$ with a boundary condition $a(-\infty)=2\pi k/N$ and $a(+\infty)=2\pi(k+1)/N$. Such a solution will have a finite energy per unit area (the wall tension), which can be estimated from the potential \eqref{eq:axionpot}. Importantly, across this domain wall, the axion field (and hence the effective $\theta$) changes by $\Delta a = 2\pi/N$. The topological charge $Q$ stored in the gauge field must change accordingly, since $a$ and $Q$ are coupled by \eqref{eq:axioncoupling}. In fact, one can show that a domain wall between vacua $k$ and $k+1$ carries one unit of $Q$ flux through \cite{Hason:2020yqf}. Intuitively, as $a$ slides by $2\pi/N$, the term \eqref{eq:axioncoupling} in the action effectively shifts $\theta$ by $2\pi/N$ in the region of the wall, which by our earlier analysis means a unit instanton number can be absorbed or emitted. Equivalently, an instanton passing across the wall can convert into a fluctuation of the $a$ field. The wall is thus a locus where the difference in $\theta$ angle is realized. We can think of it as an interface between two $\theta$-vacua. Figure~\ref{fig:vacua_wall} illustrates the $N$ vacua and a domain wall connecting neighboring vacua. 

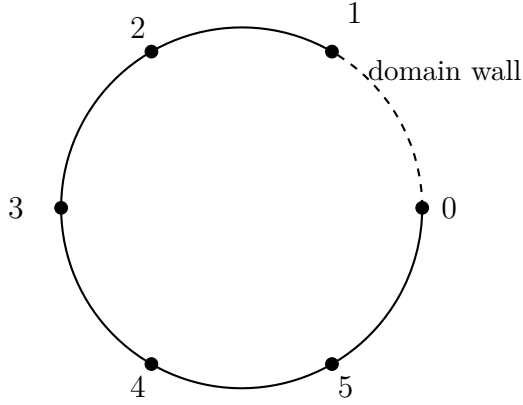
\begin{figure}[h]
\centering
\begin{tikzpicture}[scale=1.2]
  \def\R{2cm} 

  \draw[thick] (60:\R) arc (60:360:\R);
  \draw[thick,dashed] (0:\R) arc (0:60:\R);

  \foreach \x in {0,60,120,180,240,300}{
    \filldraw[black] (\x:\R) circle (2pt);
  }

  \node at (0:2.3cm) {$0$};
  \node at (60:2.5cm) {$1$};
  \node at (120:2.3cm) {$2$};
  \node at (180:2.5cm) {$3$};
  \node at (240:2.3cm) {$4$};
  \node at (300:2.3cm) {$5$};

  \node[anchor=south] at (30:2.6cm) {\small domain wall};

\end{tikzpicture}
\caption{Discrete $\theta$-vacua (labeled by an integer $k$ mod $N$) related by the $\theta \to \theta + 2\pi/N$ shift. In the ungauged theory, these are distinct vacua leading to a $\mathbb{Z}_N$ degeneracy. A domain wall (dashed segment) interpolates between neighboring vacua that differ by a discrete $\theta$-shift. Crossing the domain wall changes the topological angle by $2\pi/N$ and transfers one unit of topological charge. Gauging the $\mathbb{Z}_N$ shift symmetry identifies all these vacua as a single state, and the domain wall becomes unobservable.}
\label{fig:vacua_wall}
\end{figure}

The physical content of Model~B in the ungauged case is therefore an $\mathbb{Z}_N$ spontaneously broken symmetry with $N$ degenerate vacua and stable domain walls connecting them. Now we impose that this discrete shift symmetry is gauged. Gauging the $\mathbb{Z}_N$ means that $a(x)\sim a(x) + 2\pi/N$ is promoted to a gauge equivalence rather than a physical difference. Operationally, we would introduce a $\mathbb{Z}_N$ one-form gauge field (or, since $a$ is a 0-form field, equivalently a $\mathbb{Z}_N$ gauge transformation that can shift $a$ by $2\pi/N$ in patches) and sum over its configurations. When a global discrete symmetry is gauged, any vacuum degeneracy associated with its spontaneous breaking is lifted: all the former vacua become one vacuum, since one can hop between them by a gauge transformation. This is a manifestation of Elitzur’s theorem, which forbids spontaneous breaking of a gauged symmetry. In the present case, once we gauge $a\to a+2\pi/N$, the potential \eqref{eq:axionpot} no longer selects $N$ distinct physical states --- it merely provides one periodic potential for $a$, and all $N$ minima are identified. Consequently, the domain walls that were supported by the difference between vacua are no longer distinct objects; a configuration where $a$ interpolates from $2\pi k/N$ on one side to $2\pi (k+1)/N$ on the other can be continuously deformed (by a gauge transformation that shifts $a$ in one region by $-2\pi/N$) into the trivial configuration $a$ constant. In effect, the domain wall can pairwise annihilate with the ``gauge image'' of itself. Thus, in the gauged theory, there is a single vacuum and no stable domain walls. This agrees perfectly with the conclusion from Model~A. In Model~B language, the $\mathbb{Z}_N$ one-form gauge field introduced in Model~A is analogous to the statement that the axion $a$ is a compact field with periodicity $2\pi$ but we identify shifts of $a$ by $2\pi/N$ as gauge redundancies. One way to implement this is to say that the fundamental group of the field space is $\mathbb{Z}_N$ rather than $\mathbb{Z}$, which again means that only $e^{i a N}$ is a single-valued operator (not $e^{i a}$ by itself) \cite{Hsin:2020nts}. Such a condition can be realized by coupling the axion to a discrete gauge field (sometimes called a $\mathbb{Z}_N$ 0-form gauge field) or by having a multi-valued field with identification $a\sim a+2\pi/N$. In any case, gauging eliminates the distinct $k$ vacua. 

At this stage, one might ask: what if the discrete $\theta$-shift symmetry is not truly exact, or if it has an anomaly? An 't Hooft anomaly for a global symmetry means that the symmetry cannot be consistently gauged without adding new degrees of freedom in a higher dimension \cite{Gaiotto:2014kfa}. In our scenario, a particularly interesting case is $N=2$ (a $\mathbb{Z}_2$ shift symmetry). When $N=2$, the shift $a\to a+\pi$ could be identified with a time-reversal (CP) symmetry if we interpret $a=\theta$ itself. In $SU(N)$ Yang-Mills theory with even $N$, it is known that at $\theta=\pi$ there is a $\mathbb{Z}_2$ time-reversal symmetry $T$ (since $\theta \to -\theta$ leaves $\theta=\pi$ invariant) and this $T$ has a mixed anomaly with the $\mathbb{Z}_N$ one-form center symmetry \cite{Gaiotto:2014kfa,Tanizaki:2018wtg}. This anomaly prevents the $\mathbb{Z}_2$ (CP) from being realized as a trivial gapped symmetry in the infrared -- either it is spontaneously broken (giving two vacua, as Witten first argued for $SU(2)$ Yang--Mills), or if unbroken, the vacuum must carry a topological quantum field theory that cancels the anomaly. For $SU(N)$ with even $N$, one indeed finds $N$ degenerate vacua at $\theta=\pi$ related by certain large center transformations, and a domain wall between these vacua carries a 3d $SU(N)_1$ Chern-Simons theory (or more generally a $\mathbb{Z}_N$ gauge theory) on its worldvolume \cite{Hason:2020yqf}. If we attempted to gauge the $\mathbb{Z}_2$ shift symmetry in this case, we would encounter an inconsistency unless we include an additional 5-dimensional bulk term whose variation cancels the anomaly (this is the higher-form analog of adding a Chern--Simons term when gauging a flavor symmetry with a parity anomaly). In Model~B terms, the discrete symmetry $a\to a+\pi$ at $\theta=\pi$ is anomalous, meaning that the domain wall cannot be simply removed by gauging -instead, if one insists on gauging it, the domain wall must remain as a physical object but now becomes the boundary of a 4+1D invertible topological term (an example of anomaly inflow) \cite{Hason:2020yqf}. In Model~A terms, one could say that the discrete $B_{2}$ field has its own Chern--Simons coupling in 5d that accounts for the anomaly \cite{Cordova:2019uob}. For our purposes, we will assume that either $N$ is such that no anomaly is present or that if an anomaly exists, the required counterterm is included so that the $\Z_N$ symmetry is gaugeable without contradiction. In summary, an anomalous $\Z_N$ shift leads to rich physics (degeneracy and domain wall topological modes) if treated as a global symmetry \cite{Tanizaki:2018wtg,Hason:2020yqf}, whereas once fully gauged (with anomaly inflow), those features are accounted for by higher-dimensional terms rather than by multiplicity of vacua in 4d.

To conclude, Models~A and~B provide consistent pictures of what happens when a discrete $\theta$-shift symmetry is gauged. In both descriptions, the effective $\theta$ angle becomes quantized and the vacuum structure simplifies to a single branch. In Model~A, this is seen through the orbifolding of the $\theta$-circle and the coupling to a discrete topological field, which removes all but the $Q\equiv 0 \pmod N$ sector and eliminates vacuum degeneracy. In Model~B, we see that while a global $\Z_N$ leads to $N$ vacua and domain walls, promoting it to a gauge symmetry identifies all those vacua and renders the domain walls unobservable (or, if an anomaly is present, the domain walls acquire a role in anomaly inflow rather than in multi-vacuum dynamics). Gauging a discrete $\theta$-shift thus imposes a unique $\theta$-sector (up to the choice of discrete $\theta$ angle $p$ in Model~A) and enforces that any would-be change of $\theta$ is accompanied by a quantum (an instanton number) that is pure gauge. This illustrates a general theme: once a $\theta$ angle becomes dynamical or gauged, the space of vacua and topological sectors experiences an extreme reduction, often encoded in a higher-form symmetry or two-group structure \cite{Cordova:2019uob}, and any previous 't Hooft anomalies manifest as demands for topological couplings (such as Chern-Simons terms on domain walls) to preserve consistency \cite{Hason:2020yqf}.

\section{Vacuum Selection and the Bound \texorpdfstring{$|\bar\theta|\le \pi/N$}{|thetabar| <= pi/N}}
\label{sec3}
We now give a mathematically complete derivation of the vacuum selection mechanism in the discrete $\theta$ gauged theory and prove the bound $|\bar\theta|\le \pi/N$ as a property of the exact ground state rather than an externally imposed restriction. Throughout, $\bar\theta$ denotes the physically meaningful CP-violating angle including quark mass phases, $\bar\theta \equiv \theta + \arg\det M$ for a quark mass matrix $M$, consistent with the anomalous $U(1)_A$ Ward identity and chiral rotations \cite{Witten:1979vv,Veneziano:1979ec,DiVecchia:1980yfw}. We begin from the multi-branch structure of the Yang-Mills vacuum energy, then implement the $\mathbb Z_N$ orbifolding of the $\theta$ circle by gauging a discrete shift symmetry. Two complementary effective descriptions i.e., a pure-gauge description (Model~A) and an axion/mass-extended description (Model~B) are treated within the same narrative to establish the same bound.

\subsection{Multi-branch structure and large-$N$ scaling}
For a pure $SU(N_c)$ Yang-Mills theory, the vacuum energy density as a function of an angular parameter $\vartheta$ is $2\pi$-periodic, CP-even around $\vartheta=0$, and convex. Large-$N_c$ arguments \cite{Witten:1998uka} show that this function is multi-branched: there exists a smooth, $2\pi$-periodic and even function $f(x)$, with $f'(0)=0$ and $f''(0)=\chi_t/\Lambda^4>0$, such that the exact energy can be represented as
\begin{equation}
\label{eq:multi-branch}
E_{\text{YM}}(\vartheta) \;=\; \min_{k\in\mathbb Z}\, \mathcal E_k(\vartheta)\,, 
\qquad 
\mathcal E_k(\vartheta) \;\equiv\; \mathcal A_{N_c}\,\Lambda^4\, f\!\left(\frac{\vartheta+2\pi k}{\kappa_{N_c}}\right),
\end{equation}
where $\Lambda$ is a strong scale, $\mathcal A_{N_c}$ is an $N_c$-dependent prefactor that scales as $N_c^2$ in the large-$N_c$ limit, and $\kappa_{N_c}$ scales as $N_c$ so that $f$ has $\mathcal O(1)$ curvature near the origin.\footnote{The canonical large-$N_c$ realization sets $\mathcal A_{N_c}\propto N_c^2$ and $\kappa_{N_c}\propto N_c$, reproducing $E_{\text{YM}}(\vartheta)\sim N_c^2 \min_k F\!\left((\vartheta+2\pi k)/N_c\right)$ \cite{Witten:1998uka}. Our derivations below use only periodicity, evenness, and convexity together with the multi-branch envelope representation.} The topological susceptibility is the curvature at the origin,
\begin{equation}
\label{eq:susceptibility}
\chi_t \;=\; \left.\frac{\partial^2 E_{\text{YM}}(\vartheta)}{\partial\vartheta^2}\right|_{\vartheta=0} 
\;=\; \int d^4x \,\langle q(x)q(0)\rangle_c 
\;=\; \frac{1}{V} \langle Q^2\rangle_c \;>\;0,
\end{equation}
where $q(x)=\tfrac{1}{32\pi^2}\Tr(F\tilde F)$, $Q=\int d^4x\,q(x)\in\mathbb Z$, and $V$ is the spacetime volume; the second equality is the anomalous Ward identity \cite{DiVecchia:1980yfw}. Near $\vartheta=0$, a Taylor expansion yields
\begin{equation}
\label{eq:taylor}
\mathcal E_k(\vartheta) \;=\; \frac{1}{2}\,\chi_t \,(\vartheta+2\pi k)^2 + \mathcal O\!\big((\vartheta+2\pi k)^4\big),
\end{equation}
which is sufficient for the local analysis below but not required for the global, nonperturbative conclusions.

\subsection{Gauging a discrete $\texorpdfstring{\boldsymbol{\theta}}{theta}$-shift and orbifolding}
We now gauge a discrete shift symmetry $\mathbb Z_N$ acting by $\vartheta\mapsto \vartheta + 2\pi/N$. Physically this is the orbifold identification of the $\vartheta$-circle by a subgroup of its $2\pi$-periodicity \cite{Gaiotto:2014kfa,Cordova:2019uob}. The gauged (orbifold) vacuum energy is obtained by allowing the $\vartheta$-argument of the Yang--Mills energy to vary over the $\mathbb Z_N$-orbit of a given $\theta$ and selecting the minimal branch,
\begin{equation}
\label{eq:orbifold-energy}
E(\theta) \;=\; \min_{m\in\mathbb Z}\, E_{\text{YM}}\!\left(\theta+\frac{2\pi m}{N}\right) 
\;=\; \min_{m\in\mathbb Z}\,\min_{k\in\mathbb Z}\, \mathcal A_{N_c}\,\Lambda^4\, f\!\left(\frac{\theta+\frac{2\pi m}{N}+2\pi k}{\kappa_{N_c}}\right).
\end{equation}
Combining the two integer labels into a single integer $\ell\equiv k+ m/N$ is possible at the level of the argument of $f$ but not as an integer because $m/N\notin\mathbb Z$ in general; nonetheless, the minimization in \eqref{eq:orbifold-energy} is elementary because it amounts to choosing the representative of the $\mathbb Z_N$-orbit that places the argument of $f$ closest to the origin modulo $2\pi$. To make this precise, define the principal-value map $\mathrm{PV}:\mathbb R\to (-\pi,\pi]$ by reducing angles modulo $2\pi$. For each $m\in\mathbb Z$, set
\begin{equation}
\label{eq:xm-def}
x_m(\theta)\;\equiv\; \mathrm{PV}\!\left(\theta+\frac{2\pi m}{N}\right)\in(-\pi,\pi]\,.
\end{equation}
Because the $N$ values $\theta+\tfrac{2\pi m}{N}$ ($m=0,1,\dots,N-1$) partition the unit circle into $N$ equal arcs, there exists an $m_*(\theta)\in\{0,\dots,N-1\}$ such that 
\begin{equation}
\label{eq:closest}
|x_{m_*}(\theta)| \;=\; \min_{m\in\mathbb Z} |x_m(\theta)| \;\le\; \frac{\pi}{N}.
\end{equation}
The inequality is sharp and follows from the pigeonhole principle on the circle: among $N$ equally spaced points, one is within angular distance $\le \pi/N$ of the origin. Since $f$ is even and nondecreasing on $[0,\pi]$ by convexity and the minimum at the origin, the minimization in \eqref{eq:orbifold-energy} is achieved by $m_*(\theta)$ and the energy reduces to the envelope
\begin{equation}
\label{eq:Eg-envelope}
E(\theta) \;=\; \mathcal A_{N_c}\,\Lambda^4\, f\!\left(\frac{x_{m_*}(\theta)}{\kappa_{N_c}}\right),
\qquad 
x_{m_*}(\theta)\in\left[-\frac{\pi}{N},\,\frac{\pi}{N}\right].
\end{equation}
It is convenient to define the selected or effective angle
\begin{equation}
\label{eq:thetaeff-def}
\theta_{\rm eff}(\theta) \;\equiv\; x_{m_*}(\theta) \;=\; \mathrm{PV}\!\left(\theta+\frac{2\pi m_*(\theta)}{N}\right),
\end{equation}
so that the gauged energy is simply $E(\theta)=\mathcal A_{N_c}\Lambda^4 f(\theta_{\rm eff}/\kappa_{N_c})$ with $\theta_{\rm eff}\in[-\pi/N,\pi/N]$. Equation \eqref{eq:Eg-envelope} already implies the bound $|\theta_{\rm eff}|\le \pi/N$, below we show that $\theta_{\rm eff}$ coincides with the physical $\bar\theta$ when quark masses and axions are included, thereby establishing $|\bar\theta|\le \pi/N$ as a dynamical consequence.

\subsection{Discrete shift symmetry, branch relabeling, and minimization}
The gauged energy \eqref{eq:Eg-envelope} is invariant under the discrete shift $\theta\to\theta+\tfrac{2\pi}{N}$. Indeed,
\begin{equation}
\label{eq:discrete-shift}
E\!\left(\theta+\frac{2\pi}{N}\right) 
= \min_{m} E_{\text{YM}}\!\left(\theta+\frac{2\pi}{N}+\frac{2\pi m}{N}\right)
= \min_{m} E_{\text{YM}}\!\left(\theta+\frac{2\pi (m+1)}{N}\right)
= E(\theta),
\end{equation}
where the minimization index relabels by $m\mapsto m+1$. Consequently, the branch label transforms as $m_*(\theta+\tfrac{2\pi}{N}) = m_*(\theta)-1$ so that $\theta_{\rm eff}(\theta+\tfrac{2\pi}{N})=\theta_{\rm eff}(\theta)$. To locate the branch-selection transitions, it is efficient to temporarily relax $m$ to a real variable and minimize $\mathcal E(\theta,m)\equiv E_{\text{YM}}(\theta+\tfrac{2\pi m}{N})$ with respect to $m$. Differentiating with respect to $m$ yields
\begin{equation}
\label{eq:dEdm}
\frac{\partial \mathcal E}{\partial m}(\theta,m) \;=\; \frac{2\pi}{N}\,E'_{\text{YM}}\!\left(\theta+\frac{2\pi m}{N}\right),
\end{equation}
so a stationary point satisfies $E'_{\text{YM}}(\theta+ \tfrac{2\pi m_*}{N})=0$. Since $E'_{\text{YM}}(\vartheta)=0$ at $\vartheta=2\pi\ell$ and only there in the fundamental domain by evenness and convexity, the continuous minimizer is $m_*^{\rm cont}(\theta)=-\tfrac{N\theta}{2\pi} +  \ell N$ for some $\ell\in\mathbb Z$. Restricting back to $m\in\mathbb Z$ selects the nearest integer(s) to $m_*^{\rm cont}(\theta)$, which is precisely the integer(s) that minimize $|x_m(\theta)|$. At the special values $\theta=\theta_c^{(\ell)}\equiv \tfrac{(2\ell+1)\pi}{N}$ the continuous minimizer is half-integer and two adjacent integers are tied. Therefore the transition points are
\begin{equation}
\label{eq:transition-thetas}
\theta=\frac{(2\ell+1)\pi}{N}\,,\qquad \ell\in\mathbb Z,
\end{equation}
at which the minimizing label jumps by one unit, $m_*(\theta_c^{(\ell)}\!+0)=m_*(\theta_c^{(\ell)}\!-0)-1$. Since the two competing branches are related by opposite arguments of $f$, namely $x_{m_*(\theta_c^{(\ell)}\!-0)}(\theta_c^{(\ell)})=-\frac{\pi}{N}$ and $x_{m_*(\theta_c^{(\ell)}\!+0)}(\theta_c^{(\ell)})=+\frac{\pi}{N}$, continuity of $E(\theta)$ follows from the evenness of $f$,
\begin{equation}
\label{eq:E-continuity}
\lim_{\epsilon\to 0^+} E\!\left(\theta_c^{(\ell)}-\epsilon\right) 
= \mathcal A_{N_c}\Lambda^4 f\!\left(\frac{\pi}{N\kappa_{N_c}}\right)
= \lim_{\epsilon\to 0^+} E\!\left(\theta_c^{(\ell)}+\epsilon\right),
\end{equation}
whereas the derivative exhibits a finite jump (a cusp) unless $f'(\pi/N\kappa_{N_c})=0$:
\begin{equation}
\label{eq:cusp}
\lim_{\epsilon\to 0^+} \frac{dE}{d\theta}\!\left(\theta_c^{(\ell)}-\epsilon\right)
= -\,\frac{\mathcal A_{N_c}\Lambda^4}{\kappa_{N_c}} f'\!\left(\frac{\pi}{N\kappa_{N_c}}\right), 
\quad
\lim_{\epsilon\to 0^+} \frac{dE}{d\theta}\!\left(\theta_c^{(\ell)}+\epsilon\right)
= +\,\frac{\mathcal A_{N_c}\Lambda^4}{\kappa_{N_c}} f'\!\left(\frac{\pi}{N\kappa_{N_c}}\right).
\end{equation}
Equations \eqref{eq:E-continuity}-\eqref{eq:cusp} display the continuity of $E(\theta)$ and the generic non-analyticity of its first derivative at the transition points, reflecting a first-order transition between neighboring branches, exactly paralleling the familiar cusp at $\vartheta=\pi$ in ordinary Yang-Mills \cite{Witten:1998uka} but now occurring at the finer set \eqref{eq:transition-thetas} due to the $\mathbb Z_N$ orbifolding.

\subsection{Pure gauge saddle and selection of the principal cell}
In the pure gauge description, the partition function at fixed $\theta$ is a sum over topological sectors $Q\in\mathbb Z$,
\begin{equation}
\label{eq:Ztheta}
Z_{\text{YM}}(\vartheta) \;=\; \sum_{Q\in\mathbb Z} e^{i\vartheta Q}\, Z_Q, 
\qquad 
Z_Q \;=\; \int_{Q}\!\mathcal D A\, e^{-S_{\text{YM}}[A]},
\end{equation}
and $E_{\text{YM}}(\vartheta)=-\lim_{V\to\infty}\tfrac{1}{V}\log Z_{\text{YM}}(\vartheta)$. Gauging the $\mathbb Z_N$ shift symmetry amounts to summing over $\vartheta\in\theta+\tfrac{2\pi}{N}\mathbb Z$ and projecting onto $\mathbb Z_N$-invariant sectors \cite{Gaiotto:2014kfa,Cordova:2019uob},
\begin{equation}
\label{eq:ZNgauged}
Z(\theta) \;=\; \frac{1}{N}\sum_{m=0}^{N-1} Z_{\text{YM}}\!\left(\theta+\frac{2\pi m}{N}\right) 
\;=\; \sum_{Q\in\mathbb Z} \Bigg(\frac{1}{N}\sum_{m=0}^{N-1} e^{i\frac{2\pi m}{N}Q}\Bigg) e^{i\theta Q}\, Z_Q 
\;=\; \sum_{Q\in N\mathbb Z} e^{i\theta Q}\, Z_Q.
\end{equation}
Only sectors with $Q=N k$ survive in the gauged theory. A standard saddle-point analysis at large volume $V$ then gives
\begin{equation}
\label{eq:saddle}
E(\theta) \;=\; \min_{k\in\mathbb Z}\, \mathcal F\!\left(\theta; k\right), 
\qquad 
\mathcal F(\theta,k)\;=\; -\,\lim_{V\to\infty}\frac{1}{V}\log\!\left(e^{i\theta Nk} Z_{Nk}\right).
\end{equation}
Since $Z_{Nk}$ depends only on the principal value of $\theta$ modulo $2\pi$ through $\theta Nk$, the minimization over $k$ is equivalent to selecting the representative of the orbit $\theta\mapsto \theta+\tfrac{2\pi m}{N}$ that brings $\theta Nk$ closest to a multiple of $2\pi$, i.e. it selects $\theta_{\rm eff}\in[-\pi/N,\pi/N]$ as in \eqref{eq:thetaeff-def}. Near the origin, expanding $\mathcal F(\theta,k)$ to quadratic order using \eqref{eq:susceptibility} yields
\begin{equation}
\label{eq:quadratic}
\mathcal F(\theta;k) \;=\; \frac{1}{2}\,\chi_t\,\big(\theta+\tfrac{2\pi k}{N}\big)^2 + \mathcal O\!\big((\theta+\tfrac{2\pi k}{N})^4\big),
\end{equation}
so the global minimum is attained by the integer $k$ that minimizes $|\theta+\tfrac{2\pi k}{N}|$, which is guaranteed to be $\le \pi/N$ by \eqref{eq:closest}. Thus, within Model~A the vacuum selection rule is simply that the theory occupies the principal cell $[-\pi/N,\pi/N]$ of the orbifolded $\theta$-circle, and the energy is the envelope of the translated branches $\mathcal F(\theta,k)$, producing cusp intersections precisely at \eqref{eq:transition-thetas}. Metastable vacua correspond to non-minimizing integers $k$. At a transition point two adjacent $k$ are degenerate at equal energy, and away from it the higher-energy one is metastable and can decay via nucleation of domain walls carrying the appropriate unit of topological charge, in keeping with the standard large-$N_c$ domain wall picture \cite{Witten:1998uka,Gaiotto:2017yup}.

\subsection{Model B: Inclusion of masses or axion and dynamical realization of the bound}
When quark masses are included, the physical angle is $\bar\theta = \theta + \arg \det M$, and the partition function respects the anomalous Ward identity $Z(\theta,M)=Z(\theta+\alpha, e^{-i\alpha/N_f}M)$ such that $E$ depends only on $\bar\theta$ \cite{DiVecchia:1980yfw}. In the discrete-$\theta$-gauged theory, the vacuum energy as a function of $\bar\theta$ is therefore
\begin{equation}
\label{eq:Eg-bartheta}
E(\bar\theta) \;=\; \min_{m\in\mathbb Z}\, E_{\text{YM}}\!\left(\bar\theta+\frac{2\pi m}{N}\right)
\;=\; \mathcal A_{N_c}\Lambda^4\, \min_{m\in\mathbb Z}\, f\!\left(\frac{\mathrm{PV}(\bar\theta+\frac{2\pi m}{N})}{\kappa_{N_c}}\right),
\end{equation}
identical in form to \eqref{eq:Eg-envelope} with $\theta\to\bar\theta$. The same geometric argument gives the dynamical bound $|\bar\theta_{\rm eff}|\le \pi/N$, where $\bar\theta_{\rm eff}\equiv \mathrm{PV}(\bar\theta+\tfrac{2\pi m_*(\bar\theta)}{N})$ is the argument that minimizes the energy. This conclusion is independent of the microscopic size or pattern of CP-violating masses. 

If an axion $a$ is present with coupling $\tfrac{i}{8\pi^2}\int a\,\Tr(F\wedge F)$ and a periodic potential $U(a)$ invariant under $a\to a+\tfrac{2\pi}{N}$ (as in a discrete clockwork or any construction that respects the same $\mathbb Z_N$), the total effective potential is $V_{\rm eff}(a,\bar\theta)=E_{\text{YM}}(\bar\theta+a)+U(a)$ before gauging. Gauging the discrete shift now implies 
\begin{equation}
V_{\rm eff}(\bar\theta) = \min_{m\in\mathbb Z}\min_{a\in\mathbb R}\left\{ E_{\text{YM}}(\bar\theta+a+\tfrac{2\pi m}{N})+U(a)\right\}.
\end{equation}

Let $g(x)\equiv E_{\text{YM}}(x)$ and suppose $U(a)$ has isolated minima at $a=2\pi r/N$ with positive curvature $U''(2\pi r/N)>0$. Taylor-expanding $g$ near its minimum at $x=0$ as $g(x)=\tfrac{1}{2}\chi_t x^2+\dots$, the stationarity condition $\partial_a V_{\rm eff}=0$ reads
\begin{equation}
\label{eq:axion-stationary}
\chi_t\big(\bar\theta+a+\tfrac{2\pi m}{N}\big) + U'(a) \;=\; 0.
\end{equation}
For each $m$, there is a unique solution $a_m(\bar\theta)$ near $a=-\bar\theta-\tfrac{2\pi m}{N}$ provided $U''>0$. Evaluating $V_{\rm eff}$ at $a_m(\bar\theta)$ gives
\begin{equation}
\label{eq:axion-on-shell}
V_{\rm eff}(\bar\theta) \;=\; \min_{m\in\mathbb Z}\left\{ \frac{1}{2}\frac{\chi_t\,U''(0)}{\chi_t+U''(0)}\left(\bar\theta+\frac{2\pi m}{N}\right)^2 + \text{higher even powers}\right\},
\end{equation}
so the same envelope structure and the same minimizing label $m_*(\bar\theta)$ emerge. The axion dynamically relaxes the theory to the nearest representative of the $\mathbb Z_N$ orbit of $\bar\theta$, implementing the same selection rule and enforcing $|\bar\theta_{\rm eff}|\le \pi/N$. The precise curvature is renormalized by $U''(0)$, but the location of branch crossings and the bound are unaffected. In particular, the cusp structure persists: at $\bar\theta=(2\ell+1)\pi/N$ two solutions $a_{m}$ and $a_{m-1}$ have equal energy and exchange stability, producing a first-order transition whose latent discontinuity is set by the jump in $\partial_{\bar\theta}V_{\rm eff}$, which by the chain rule equals $\partial_x g(x)$ evaluated at $x=\pm\pi/N$ and therefore jumps by $2 g'(\pi/N)$, matching \eqref{eq:cusp}.

\subsection{Analytic continuation, metastability, and degeneracy lifting}
\label{sec:3.6}
The envelope construction $E(\bar\theta)=\min_m \mathcal E_m(\bar\theta)$ defines a continuous, piecewise analytic function whose derivative is piecewise continuous with jump discontinuities at \eqref{eq:transition-thetas}. The left and right derivatives at a transition are given by $E'_{\text{YM}}(\bar\theta\pm\pi/N)$, which are equal in magnitude and opposite in sign by CP symmetry of the underlying branch. The metastable branches are analytic continuations of the dominant branch past the transition. Their decay proceeds by nucleation of codimension-one domain walls that change the integer label $m$ by one, carrying precisely one unit of topological charge in the sense of anomaly inflow on the wall worldvolume \cite{Gaiotto:2017yup}. Gauging the discrete shift identifies the $N$ classical vacua related by $\bar\theta\to \bar\theta+\tfrac{2\pi}{N}$, removing true degeneracy and leaving only metastable excitations separated from the ground state by the domain-wall tension. In this sense, the $\mathbb Z_N$ gauging lifts the would-be $N$-fold degeneracy to a single vacuum with a piecewise-analytic $E(\bar\theta)$ whose global minimum sits in the principal cell $[-\pi/N,\pi/N]$, and the inequality $|\bar\theta|\le \pi/N$ holds for the selected ground state value $\bar\theta_{\rm eff}$ that minimizes the energy.

\subsection{Topological susceptibility and curvature of the envelope}
Differentiating \eqref{eq:Eg-bartheta} inside a given cell where $m_*(\bar\theta)$ is constant, the curvature of the envelope equals that of the underlying branch,
\begin{equation}
\label{eq:curvature-inside}
\frac{d^2E}{d\bar\theta^2}(\bar\theta) \;=\; \left.\frac{d^2E_{\text{YM}}}{d\vartheta^2}\right|_{\vartheta=\bar\theta+\frac{2\pi m_*}{N}} 
\;=\; \int d^4x\,\langle q(x)q(0)\rangle_c\Big|_{\vartheta=\bar\theta_{\rm eff}},
\qquad |\bar\theta_{\rm eff}|<\frac{\pi}{N}.
\end{equation}
At the transition, $d^2E/d\bar\theta^2$ contains a Dirac delta supported at $\bar\theta=(2\ell+1)\pi/N$ that encodes the cusp. Equation \eqref{eq:curvature-inside} exhibits that the anomalous Ward identity \eqref{eq:susceptibility} continues to control the curvature inside each cell, while the non-analyticities are entirely due to the branch selection.

\subsection{Rigorous statement and proof of the bound}
The construction above can be condensed into the following lemma.
\begin{lemma}
Let $f:\mathbb R\to\mathbb R$ be even, $2\pi$-periodic, and convex on $[0,\pi]$, and define $E(\bar\theta)= \min_{m\in\mathbb Z} f\!\big(\mathrm{PV}(\bar\theta+\tfrac{2\pi m}{N})\big)$. Then there exists a representative $\bar\theta_{\rm eff}\in[-\pi/N,\pi/N]$ such that $E(\bar\theta)=f(\bar\theta_{\rm eff})$, and at the boundary points $\bar\theta=\tfrac{(2\ell+1)\pi}{N}$ two representatives $\bar\theta_{\rm eff}=\pm \tfrac{\pi}{N}$ are tied.
\end{lemma}

\medskip\noindent\textit{Proof.} For any $\bar\theta$ and integers $m=0,\dots,N-1$, the numbers $x_m(\bar\theta)=\mathrm{PV}(\bar\theta+\tfrac{2\pi m}{N})$ lie in $(-\pi,\pi]$ and differ by $\tfrac{2\pi}{N}$ modulo $2\pi$. Partition $(-\pi,\pi]$ into $N$ consecutive closed arcs of length $\tfrac{2\pi}{N}$ centered at $\{ \tfrac{2\pi r}{N}\}_{r=-(N-1)/2}^{(N-1)/2}$. Exactly one $x_m$ lies in the central arc $[-\pi/N,\pi/N]$; call it $\bar\theta_{\rm eff}$. For any other $x_{m'}$ we have $|x_{m'}|\ge |\bar\theta_{\rm eff}|$, by evenness and monotonicity of $f$ on $[0,\pi]$, $f(|x_{m'}|)\ge f(|\bar\theta_{\rm eff}|)$. Therefore $\min_m f(x_m)=f(\bar\theta_{\rm eff})$ with $|\bar\theta_{\rm eff}|\le \pi/N$, proving the claim. At $\bar\theta=\tfrac{(2\ell+1)\pi}{N}$, the two nearest representatives are $\pm \tfrac{\pi}{N}$ by construction and have equal energy by evenness, establishing the stated degeneracy. \hfill$\square$

\begin{figure}[ht]
\centering
\begin{tikzpicture}
\begin{groupplot}[group style={group size=2 by 1, horizontal sep=2.2cm}, width=0.47\textwidth, height=0.33\textwidth]
\nextgroupplot[
xlabel={$\,\theta$}, ylabel={$\,E(\theta)/E_0$},
xmin=-3.1416, xmax=3.1416,
ymin=0, ymax=2.2,
xtick={-3.1416,-2.0944,-1.0472,0,1.0472,2.0944,3.1416},
xticklabels={$-\pi$,$-\frac{2\pi}{3}$,$-\frac{\pi}{3}$,$0$,$\frac{\pi}{3}$,$\frac{2\pi}{3}$,$\pi$},
domain=-3.1416:3.1416, samples=600, title={SU(3): cusps at $\theta=\pm\pi/3,\pm\pi$}
]
\addplot+[thick] {min(min(1 - cos(deg(x)), 1 - cos(deg(x + 2*3.1415926535/3))), 1 - cos(deg(x - 2*3.1415926535/3)))};
\draw[dashed] (axis cs:-1.0472,0) -- (axis cs:-1.0472,2.2);
\draw[dashed] (axis cs:1.0472,0) -- (axis cs:1.0472,2.2);
\draw[dashed] (axis cs:-3.1416,0) -- (axis cs:-3.1416,2.2);
\draw[dashed] (axis cs:3.1416,0) -- (axis cs:3.1416,2.2);
\nextgroupplot[
xlabel={$\,\theta$}, ylabel={$\,E(\theta)/E_0$},
xmin=-3.1416, xmax=3.1416,
ymin=0, ymax=1.0,
xtick={-3.1416,-1.88496,-1.25664,-0.62832,0,0.62832,1.25664,1.88496,3.1416},
xticklabels={$-\pi$,$-\frac{3\pi}{5}$,$-\frac{2\pi}{5}$,$-\frac{\pi}{5}$,$0$,$\frac{\pi}{5}$,$\frac{2\pi}{5}$,$\frac{3\pi}{5}$,$\pi$},
domain=-3.1416:3.1416, samples=900, title={SU(5): cusps at $\theta=\pm\pi/5,\pm3\pi/5,\pm\pi$}
]
\addplot+[thick] {min(min(min(min(1 - cos(deg(x)), 1 - cos(deg(x + 2*3.1415926535/5))), 1 - cos(deg(x - 2*3.1415926535/5))), 1 - cos(deg(x + 4*3.1415926535/5))), 1 - cos(deg(x - 4*3.1415926535/5)))};
\draw[dashed] (axis cs:-0.62832,0) -- (axis cs:-0.62832,2.2);
\draw[dashed] (axis cs:0.62832,0) -- (axis cs:0.62832,2.2);
\draw[dashed] (axis cs:-1.88496,0) -- (axis cs:-1.88496,2.2);
\draw[dashed] (axis cs:1.88496,0) -- (axis cs:1.88496,2.2);
\draw[dashed] (axis cs:-3.1416,0) -- (axis cs:-3.1416,2.2);
\draw[dashed] (axis cs:3.1416,0) -- (axis cs:3.1416,2.2);
\end{groupplot}
\end{tikzpicture}
\caption{Schematic multi-branch vacuum energy $E(\theta)$ obtained as the lower envelope of shifted branch functions (here exemplified by $E_0[1-\cos(\theta+\tfrac{2\pi k}{N})]$) for $N=3$ and $N=5$. Vertical dashed lines indicate the transition (cusp) locations at $\theta=\tfrac{(2\ell+1)\pi}{N}$. The ground state always lies in the principal cell $[-\pi/N,\pi/N]$, realizing $|\theta_{\rm eff}|\le \pi/N$.}
\label{fig:branches}
\end{figure}
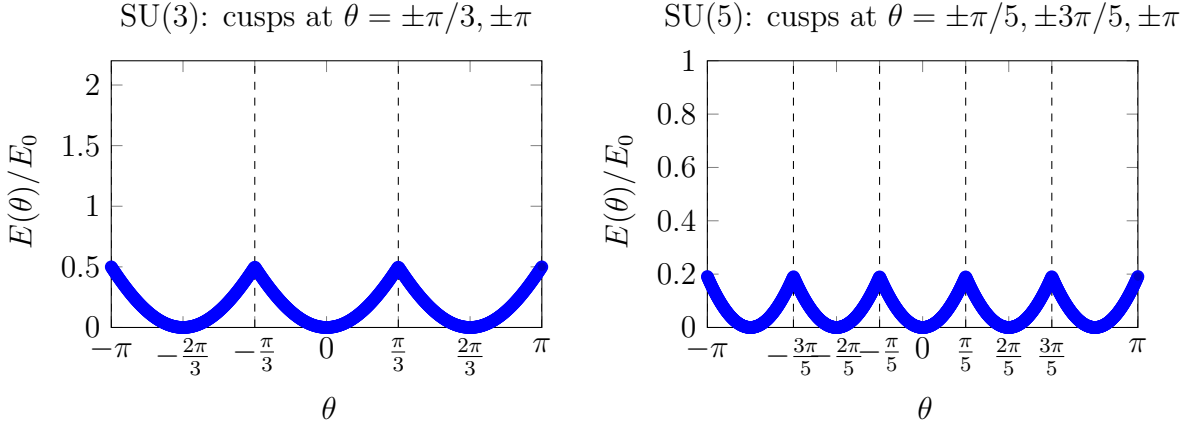

Equations \eqref{eq:Eg-envelope}, \eqref{eq:discrete-shift}, \eqref{eq:transition-thetas} and the lemma together establish that gauging a discrete $\mathbb Z_N$ subgroup of $\theta$-shifts forces the vacuum to reside in the principal cell of the orbifolded $\theta$-circle. In terms of the physical angle $\bar\theta$, the selected value satisfies
\begin{equation}
\label{eq:final-bound}
|\bar\theta_{\rm eff}|\;\le\;\frac{\pi}{N},
\end{equation}
with equal-energy cusps at the boundary points $\bar\theta_{\rm eff}=\pm\pi/N$ where neighboring branches exchange dominance. The envelope is continuous, piecewise analytic, and its curvature away from the cusps is governed by the anomalous Ward identity \eqref{eq:susceptibility}. Inclusion of quark masses or an axion preserves the same selection rule: masses merely shift $\theta\to\bar\theta$, and an axion dynamically relaxes to the nearest $\mathbb Z_N$ image, both culminating in the same bound. This is the precise sense in which $|\bar\theta|\le \pi/N$ is a gauge-protected, dynamically realized consequence of the discrete $\theta$ projection.

\section{Radiative and Gravitational Stability}

The purpose of this section is to demonstrate that the discrete-$\theta$ projection is stable against radiative renormalization within quantum field theory and against quantum-gravitational effects in curved backgrounds. 
We do so by writing every relevant ingredient in a fully explicit and self-contained manner. 
We begin with the microscopic Euclidean action for a four-dimensional $SU(N_c)$ Yang-Mills theory coupled to a nondynamical background metric $g_{\mu\nu}$ and, when present, to massive Dirac fermions $\psi_f$ with mass matrix $M$. 
We keep track of the ultraviolet (UV) regulator scale $\Lambda_{\mathrm{UV}}$, the renormalization scale $\mu$, and the $\overline{\mathrm{MS}}$ scheme. 
The bare action reads
\begin{multline}
\label{eq:bare-action}
S_B \;=\; \int d^4x\,\sqrt{g}\,\bigg[\frac{1}{4 g_B^2}\, F^{a}_{B\,\mu\nu} F_B^{a\,\mu\nu} \;+\; \sum_{f}\,\bar\psi_{f,B}\big(i\slashed{D}_B - M_B\big)\psi_{f,B}\bigg] \;+ \\ i\,\theta_B \int d^4x\, q_B(x) \;+\; i\,\theta_{G,B} \int d^4x\, q_{G,B}(x),
\end{multline}
where $F^a_{B\,\mu\nu}$ is the bare non-Abelian field strength with color index $a=1,\dots,N_c^2-1$, $D_{B,\mu}$ is the bare covariant derivative, and $g\equiv\det g_{\mu\nu}$. 
The CP-odd densities are defined with all numerical coefficients fixed as
{\begin{multline}
\label{eq:q-defs}
q_B(x) \;\equiv\; \frac{1}{32\pi^2}\, \epsilon^{\mu\nu\rho\sigma}\,\Tr\!\big(F_{B,\mu\nu} F_{B,\rho\sigma}\big) 
\;=\; \frac{1}{32\pi^2}\, \Tr\!\big(F_{B,2}\wedge f_{2,B}\big)\,, 
\\ 
q_{G,B}(x) \;\equiv\; \frac{1}{384\pi^2}\,\epsilon^{\mu\nu\rho\sigma}\, 
R^{\alpha}{}_{\beta\mu\nu} R^{\beta}{}_{\alpha\rho\sigma} 
\;=\; \frac{1}{384\pi^2}\, R\wedge R,
\end{multline}}
so that their spacetime integrals on compact, oriented spin four-manifolds are integers by the Chern-Weil and Atiyah-Singer index theorems. 
Namely, $Q_B\equiv\int d^4x\,q_B(x)\in\mathbb Z$ and $Q_{G,B}\equiv\int d^4x\,q_{G,B}(x)\in\mathbb Z$ when appropriate boundary conditions are imposed~\cite{Alvarez-Gaume:1983ihn,Eguchi:1980jx}. 
These statements rely on the chosen normalizations of the trace and curvature and are standard in the study of instantons and gravitational Pontryagin classes~\cite{tHooft:1976snw}. 

Renormalized fields and couplings are introduced via $A_{B,\mu}^a=Z_3^{1/2} A_\mu^a$, $\psi_{f,B}=Z_\psi^{1/2}\psi_f$, $g_B=\mu^{\epsilon} Z_g\, g$ with $d=4-2\epsilon$. 
The mass matrix renormalizes as $M_B=Z_M M$, and the CP-odd angles renormalize as $\theta_B=\theta_R+\delta\theta$ and $\theta_{G,B}=\theta_{G,R}+\delta\theta_G$. 
The renormalization of the topological densities requires care. 
Although the local operator $q(x)$ mixes with total derivatives and contact terms, the integrated charges $Q\equiv\int d^4x\,q(x)$ and $Q_G\equiv\int d^4x\,q_G(x)$ are renormalization-invariant integers. 
Hence the $\theta$-angles are not renormalized in perturbation theory, $\delta\theta=\delta\theta_G=0$, a fact consistent with the Adler-Bell-Jackiw anomaly equations and the non-renormalization of the $\theta$-term~\cite{DiVecchia:1980yfw,Jackiw:1976pf,Shifman:2012zz,Adler:1969gk,Bell:1969ts,Fujikawa:1980eg}. 
More explicitly, with renormalized fields we have $q_B(x)=Z_q\,q(x)+\partial_\mu K^\mu(x)$, where $K^\mu$ is a scheme-dependent local current and $Z_q=1+\mathcal O(g^2)$. 
Upon integration over spacetime, the total derivative vanishes and the integer $Q$ is unchanged. 
Equating the bare and renormalized path integrals therefore forces $\theta_B=\theta_R$ to all orders in perturbation theory. 
Equivalently, this is expressed as $\beta_\theta\equiv \mu\,d\theta_R/d\mu=0$ and $\beta_{\theta_G}=0$. 

When fermions are present, the physically relevant angle is the RG-invariant combination $\bar\theta\equiv \theta_R+\arg\det M$. 
Under a chiral $U(1)_A$ rotation $\psi_f\mapsto e^{i\alpha\gamma_5}\psi_f$, the Jacobian shifts $\theta_R\to\theta_R-2N_f\alpha$ while $M\to e^{2i\alpha}M$ and hence $\arg\det M\to \arg\det M+2N_f\alpha$. 
As a result, $\bar\theta$ is invariant~\cite{DiVecchia:1980yfw,Fujikawa:1980eg}. 
The anomalous Ward identity follows as
\begin{equation}
\label{eq:axial-anomaly}
\partial_\mu J_5^\mu \;=\; \frac{N_f}{16\pi^2}\,\Tr(F_{\mu\nu}\tilde F^{\mu\nu}) \;+\; \frac{1}{384\pi^2}\,R_{\mu\nu\rho\sigma}\tilde R^{\mu\nu\rho\sigma} \;-\; 2i\,\bar\psi\,M\gamma_5\psi,
\end{equation}
where tildes denote duals with $\epsilon^{\mu\nu\rho\sigma}$. 
This shows explicitly that the curvature of the vacuum energy with respect to $\bar\theta$ equals the topological susceptibility $\chi_t$ to leading order around $\bar\theta=0$, 
\begin{equation}
    \chi_t=\partial_{\bar\theta}^2E(\bar\theta)\big|_{\bar\theta=0}=\int d^4x\,\langle q(x)q(0)\rangle_c>0.
\end{equation} 

The discrete-$\theta$ projection is implemented microscopically by coupling the gauge theory to a compact topological sector that gauges a $\mathbb Z_N$ subgroup of $\theta$-shifts. 
We adopt a manifestly local representative convenient for radiative power counting. 
Introducing a compact Abelian one-form gauge field $a$ with field strength $f=da$ and a compact two-form gauge field $B$ with field strength $H=dB$, both normalized so that $\int_\Sigma f\in 2\pi\mathbb Z$ for any closed two-surface $\Sigma$ and $\int_C B\in 2\pi\mathbb Z/N$ for any closed two-cycle $C$, we add the following degree-matched topological action:
\begin{equation}
S_{\mathrm{top}}(4D)=\frac{iN}{2\pi}\!\int_{M_4}\!B_2\wedge f_2+\frac{i}{2\pi}\!\int_{M_4}\!a_1\wedge\!\Big(\frac{1}{8\pi^2}\,CS_3(A)+\kappa_G\,CS_3(\omega)\Big)+\frac{i\theta}{8\pi^2}\!\int_{M_4}\!\mathrm{Tr}(F\wedge F)\;.
\label{eq:topological-sector}
\end{equation}
Here $d\,CS_3(A)=\mathrm{Tr}(F\wedge F)$ and $d\,CS_3(\omega)=\mathrm{Tr}(R\wedge R)$. 
The coefficient $\kappa_G\in\mathbb Z$ is an integer to be fixed below, and we have written differential-form notation in four dimensions with the understanding that 
\begin{align}
B_2\wedge f_2&=\tfrac{1}{2}\epsilon^{\mu\nu\rho\sigma}B_{\mu\nu} f_{\rho\sigma} d^4x,\\
a_1\wedge \Tr(F\wedge F)&=\tfrac{1}{2}\epsilon^{\mu\nu\rho\sigma} a_\mu \Tr(F_{\nu\rho}F_{\sigma\lambda}) dx^\lambda,
\end{align}
after integrating by parts as appropriate~\cite{Freedman:1980us,Dvali:2005an}. 
The gauge transformations are $a\mapsto a+d\lambda$ and $B\mapsto B+d\Lambda$, with $\lambda$ a compact $0$-form and $\Lambda$ a compact $1$-form. 
These are supplemented by a discrete $0$-form transformation that shifts $\theta_R\mapsto \theta_R + 2\pi/N$ and $a\mapsto a$ but is accompanied by a large gauge transformation of the non-Abelian connection that shifts the integer $Q$ by one unit when acting on field space. 

The equations of motion obtained by varying $a$ and $B$ are, respectively,
\begin{equation}
\label{eq:a-eom}
\frac{N}{2\pi}\, dB \;+\; \frac{1}{2\pi}\,\Tr(F\wedge F) \;+\; \frac{\kappa_G}{2\pi}\,\mathcal R \;=\; 0 \quad\Longrightarrow\quad \Tr(F\wedge F) + \kappa_G\, \mathcal R \;=\; -\,N\, dB,
\end{equation}
\begin{equation}
\label{eq:B-eom}
\frac{N}{2\pi}\, f \;=\; 0 \quad\Longrightarrow\quad \int_\Sigma f \;=\; 2\pi\, \ell_\Sigma,\;\; \ell_\Sigma\in\mathbb Z,\;\; \text{with}\;\; \ell_\Sigma\equiv 0\;\text{mod}\;N\;\;\text{for all closed}\;\Sigma,
\end{equation}
so that $f$ has only $\mathbb Z$-quantized fluxes that are multiples of $2\pi N$ and the 4-form density $\Tr(F\wedge F)+\kappa_G\mathcal R$ is constrained to be an exact form with quantized cohomology class in $N\,H^{4}(M,\mathbb Z)$. 
In particular, integrating \eqref{eq:a-eom} over any compact manifold gives the selection rule
\begin{equation}
\label{eq:selection-rule}
Q \;+\; \kappa_G\,Q_G \;=\; N\,K,\qquad K\in\mathbb Z,
\end{equation}
which states that only those field configurations survive in the gauged path integral whose combined topological number $Q+\kappa_G Q_G$ is divisible by $N$. 
This is the continuous-field derivation of the projector obtained by summing over $\theta$-images. 
It shows in manifestly local form that the discrete projection acts on the total Pontryagin density including a possible gravitational admixture. 
The coefficient $\kappa_G$ must be integer-quantized because the partition function changes by a phase $e^{i\kappa_G \int a\wedge \mathcal R/(2\pi)}$ under large diffeomorphisms that shift $\int \mathcal R$ by $2\pi$. 
Insisting that all large-gauge and large-diffeomorphism transformations multiply the path integral by unity for arbitrary backgrounds forces $\kappa_G\in\mathbb Z$~\cite{Cordova:2019uob,Eguchi:1980jx}. 

With \eqref{eq:selection-rule} in hand, the renormalized CP-odd sector of the action, including the ordinary $\theta$-term and its gravitational cousin, can be written as
\begin{equation}
\label{eq:theta-sector-ren}
S_{\theta,\mathrm{tot}} \;=\; i\,\theta_R \int q \;+\; i\,\theta_{G,R} \int q_G \;+\; \int \mathcal L_{\mathrm{top}}\,,
\end{equation}
and any radiative correction that shifts $\theta_R\to \theta_R+\delta\theta_{\mathrm{rad}}$ and $\theta_{G,R}\to \theta_{G,R}+\delta\theta_{G,\mathrm{rad}}$ leads to a net phase factor $e^{i(\delta\theta_{\mathrm{rad}}\,Q + \delta\theta_{G,\mathrm{rad}}\,Q_G)}$ in each topological sector. 
Gauge invariance under \eqref{eq:selection-rule} restricts such phases. 
In the gauged theory we may sum over the discrete images of $\theta$ and $B$ so that the only physically meaningful combinations are $\theta_{\mathrm{eff}}\equiv \mathrm{PV}(\theta_R + \kappa_G \theta_{G,R})$ reduced modulo $2\pi/N$ and the integer $K$ of \eqref{eq:selection-rule}. 
Since perturbative radiative corrections are continuous functions of couplings while $K$ is integer-valued, the only allowed perturbative shifts are those that amount to $\theta_{\mathrm{eff}}\mapsto \theta_{\mathrm{eff}} + \tfrac{2\pi}{N}\times \text{integer}$, which are pure gauge in the orbifolded theory. 
Equivalently, $\delta\theta_{\mathrm{rad}}$ and $\delta\theta_{G,\mathrm{rad}}$ must conspire so that $(\delta\theta_{\mathrm{rad}}\,Q + \delta\theta_{G,\mathrm{rad}}\,Q_G)$ is an integer multiple of $2\pi/N$ for all allowed $(Q,Q_G)$ obeying \eqref{eq:selection-rule}. 
For generic backgrounds this is possible perturbatively only if $\delta\theta_{\mathrm{rad}}=\delta\theta_{G,\mathrm{rad}}=0$. 
In particular, standard loop corrections cannot generate a small but continuous $\bar\theta$ because any such correction would violate the discrete higher-form gauge symmetry encoded by \eqref{eq:topological-sector}. 
This statement is the discrete analogue of the familiar non-renormalization of $\theta$ in ordinary QCD~\cite{Gaiotto:2014kfa,Cordova:2019uob}. 

To make the renormalization of \eqref{eq:topological-sector} completely explicit, consider adding all local operators compatible with Lorentz symmetry and the microscopic gauge symmetries up to a fixed canonical dimension $D$. 
The most general CP-odd local counterterm at dimension $\le 6$ in flat space is a linear combination of $q(x)$ and the Weinberg operator $\mathcal O_W = \tfrac{1}{6}f^{abc}\,\tilde F^a_{\mu\nu} F^{b\,\nu}{}_{\ \ \rho} F^{c\,\rho\mu}$, possibly multiplied by gauge-invariant scalars such as $H^\dagger H$ if the Standard Model Higgs is present, with coefficients $c_0$ and $c_W/\Lambda_{\mathrm{UV}}^2$. 
The Weinberg operator is the unique CP-odd purely gluonic operator of dimension six and was first identified in~\cite{Weinberg:1989dx}. 
Under the one-form gauge transformations that accompany the discrete projection, $q(x)$ shifts by a total derivative controlled by $dB$ according to \eqref{eq:a-eom}, while $\mathcal O_W$ is strictly invariant. 
A counterterm $c_0\,q$ would renormalize $\theta_R$ by a continuous amount $\delta\theta_{\mathrm{rad}}=c_0$, but this is forbidden because $c_0$ would explicitly break the compact one-form gauge symmetry unless $c_0$ is an integer multiple of $2\pi/N$. 
In that case it is a pure large-gauge choice and can be absorbed into the definition of the discrete image of $\theta$. 
The operator $\mathcal O_W$ does not mix with $q$ under renormalization in a way that affects the integrated charge $Q$. 
Any mixing reduces to contact terms that vanish upon spacetime integration, so it cannot generate a physical $\bar\theta$~\cite{Pospelov:2005pr,Weinberg:1989dx}. 
Hence, within ordinary radiative power counting, the only consistent renormalization of the CP-odd sector preserves the discrete projection and leaves $\theta_{\mathrm{eff}}$ unrenormalized. 
The same conclusion holds to all orders for the topological BF coupling $\tfrac{iN}{2\pi} B_2\wedge f_2$ in \eqref{eq:topological-sector}. 
Its coefficient $N$ is integer-quantized by Dirac quantization and by invariance under large gauge transformations~\cite{Freedman:1980us}. 
Perturbative loops cannot change an integer by an infinitesimal amount, so $N$ is radiatively stable. 
A change $N\to N+\Delta N$ with $\Delta N\in\mathbb Z$ would require a phase transition that modifies the spectrum of line and surface operators charged under the one-form symmetry~\cite{Gaiotto:2014kfa}. 

Gravitational back-reaction raises two logically distinct issues, namely whether Planck-suppressed local operators can erode the projection and whether nonperturbative gravitational effects such as Euclidean wormholes or gravitational instantons can break it. 
The first question is addressed by considering the most general CP-odd local functional in curved space at dimension $\le 4$, which is precisely \eqref{eq:theta-sector-ren}, together with its completion by higher-curvature invariants suppressed by $M_{\mathrm{Pl}}^{-2}$. 
The discrete gauge symmetry acts on the combination $q+\kappa_G\, q_G$ according to \eqref{eq:a-eom}, and therefore any Planck-suppressed operator linear in $q$ or $q_G$ must appear in the invariant combination $q+\kappa_G q_G$ to preserve \eqref{eq:selection-rule}. 
Since $\kappa_G$ is integer-quantized, a putative continuous correction to either $\theta_R$ or $\theta_{G,R}$ would again violate the large-gauge invariance unless it shifts the pair $(\theta_R,\theta_{G,R})$ by a multiple of $(\tfrac{2\pi}{N},-\tfrac{2\pi\kappa_G}{N})$. 
Such a shift is physically trivial in the orbifolded theory. 
As a result, local gravitational counterterms cannot generate a small $\bar\theta$ and cannot modify the envelope structure of the vacuum energy that enforced $|\bar\theta_{\mathrm{eff}}|\le \pi/N$ in flat space. 

The second question relies on the expectation that quantum gravity admits no exact global symmetries but does permit exact discrete gauge symmetries~\cite{Krauss:1988zc,Banks:2010zn,Harlow:2018tng,Coleman:1988tj,Giddings:1987cg}. 
In the present construction, the $\mathbb Z_N$ is realized as a gauged higher-form symmetry with compact gauge fields $a$ and $B$, so it falls in the class believed to be respected by quantum gravity~\cite{Dvali:2005an}. 
A Euclidean wormhole that violates a global shift symmetry would generate an operator $e^{-S_{\mathrm{wh}}}\cos(\alpha)$ for a would-be axion $\alpha$. 
However, there is no axion in our theory; the only compact field that shifts is the gauge field $a$, and gauge invariance forbids a potential for $a$ that is not a total derivative~\cite{Coleman:1988tj,Giddings:1987cg,Dvali:2005an}. 
More concretely, wormhole-induced terms are classified by local operators compatible with all gauge symmetries integrated over the wormhole throat. 
The only CP-odd local scalars at leading order are those already present in \eqref{eq:theta-sector-ren} and \eqref{eq:topological-sector}, and we have shown they are quantized and hence inert under small corrections. 
Gravitational instantons could, in principle, change $Q_G$ by $\pm 1$. 
The selection rule \eqref{eq:selection-rule} then forces a simultaneous change in $Q$ by a multiple of $N$ so that the combined topological number is a multiple of $N$. 
In Minkowski vacuum where $Q_G=0$, this reduces to the flat-space condition $Q\in N\mathbb Z$, and in curved backgrounds it replaces $\theta$ by the single effective angle that multiplies $Q+\kappa_G Q_G$, leaving the vacuum selection bound unchanged.  

We work on a smooth, oriented, compact spin five-manifold $M_5$ with boundary $\partial M_5=M_4$. 
Let $A$ be the $SU(N_c)$ connection with curvature $F=dA+A\wedge A$ and $\mathrm{Tr}$ in the fundamental normalized by $\mathrm{Tr}(T_aT_b)=\tfrac12\delta_{ab}$, so that $\frac{1}{8\pi^2}\int_{X_4}\mathrm{Tr}(F\wedge F)\in\mathbb Z$ for any closed spin $X_4$. 
Let $\omega$ be the Levi-Civita spin connection with curvature $R=d\omega+\omega\wedge\omega$ normalized so that $\frac{1}{384\pi^2}\int_{X_4}\mathrm{Tr}(R\wedge R)\in\mathbb Z$ on spin four-manifolds. 
Introduce a compact $U(1)$ one-form gauge field $a\equiv a_1$ with holonomy $2\pi$, a compact two-form $B\equiv B_2$ with periods $2\pi$, and $f\equiv da$. 
Denote by $CS_3(A)$ and $CS_3(\omega)$ the Chern-Simons three-forms obeying $d\,CS_3(A)=\mathrm{Tr}(F\wedge F)$ and $d\,CS_3(\omega)=\mathrm{Tr}(R\wedge R)$. 
The four-dimensional topological sector is a sum of genuine four-forms,
\begin{equation}
S_{\mathrm{top}}(4D)=\frac{iN}{2\pi}\int_{M_4}B_2\wedge f_2
+\frac{i}{2\pi}\int_{M_4}a\wedge\Big[\frac{1}{8\pi^2}CS_3(A)+\frac{\kappa_G}{384\pi^2}CS_3(\omega)\Big]
+\frac{i\theta}{8\pi^2}\int_{M_4}\mathrm{Tr}(F\wedge F)
\label{eq:inflow}
\end{equation}
and the minimal, rigorously normalized inflow that replaces Eq.~(4.9) is the genuine five-form
\begin{equation}
S_{\mathrm{inflow}}(5D)=\frac{i}{2\pi}\int_{M_5}a\wedge\Big[\frac{1}{8\pi^2}\mathrm{Tr}(F\wedge F)+\frac{\kappa_G}{384\pi^2}\mathrm{Tr}(R\wedge R)\Big]
\label{eq4.10}
\end{equation}
whose boundary variation implements the 4D selection rule. 
The 2-group gauge structure is $\delta B=d\Lambda$ and $\delta a=d\lambda+N\Lambda$ with $\lambda\sim\lambda+2\pi$ a compact 0-form and $\Lambda$ a compact 1-form. 
Ordinary $SU(N_c)$ and local Lorentz transformations act by $\delta A=D\epsilon$, $\delta\omega=D\zeta$ and are accompanied on the boundary by the standard descent shifts of $B$ that cancel $\delta CS_3$ but are irrelevant for the 1-form variation considered next~\cite{Gaiotto:2014kfa,Kapustin:2013uxa,Freed:2014iua}. 

Varying Eq.~\eqref{eq:inflow} under $\delta a=d\lambda+N\Lambda$ and using Stokes’ theorem with the closed four-form 
\begin{equation}
X_4\equiv\frac{1}{8\pi^2}\mathrm{Tr}(F\wedge F)+\frac{\kappa_G}{384\pi^2}\mathrm{Tr}(R\wedge R)
\end{equation}
gives 
\begin{equation}
\delta S_{\mathrm{inflow}}=\tfrac{i}{2\pi}\int_{M_4}\lambda X_4+\tfrac{iN}{2\pi}\int_{M_5}\Lambda\wedge X_4.
\end{equation}
The boundary action Eq.~\eqref{eq:inflow} varies as 
\begin{equation}
\delta S_{\mathrm{top}}(4D)=\tfrac{iN}{2\pi}\int_{M_4}d\Lambda\wedge f+\tfrac{i}{2\pi}\int_{M_4}d\lambda\wedge Y_3+\tfrac{iN}{2\pi}\int_{M_4}\Lambda\wedge Y_3,
\end{equation}
where 
\begin{equation}
Y_3\equiv\frac{1}{8\pi^2}CS_3(A)+\frac{\kappa_G}{384\pi^2}CS_3(\omega)
\end{equation}
satisfies $dY_3=X_4$. 
Rewriting 
\begin{equation}
\int_{M_4}\Lambda\wedge Y_3=\int_{M_5}d\Lambda\wedge Y_3-\int_{M_5}\Lambda\wedge X_4
\end{equation}
cancels the bulk term $\tfrac{iN}{2\pi}\int_{M_5}\Lambda\wedge X_4$ from $\delta S_{\mathrm{inflow}}$, while 
\begin{equation}
\int_{M_4}d\lambda\wedge Y_3=\int_{M_4}d(\lambda Y_3)-\int_{M_4}\lambda X_4
\end{equation}
contributes $-\tfrac{i}{2\pi}\int_{M_4}\lambda X_4$ because $\partial M_4=\varnothing$. 
The remaining boundary pieces combine to
\begin{equation}
\delta\big(S_{\mathrm{top}}(4D)+S_{\mathrm{inflow}}(5D)\big)
=\frac{i}{2\pi}\lambda\big|_{\partial M_5}\left[\frac{1}{8\pi^2}\int_{M_4}\mathrm{Tr}(F\wedge F)
+\frac{\kappa_G}{384\pi^2}\int_{M_4}\mathrm{Tr}(R\wedge R)\right],
\end{equation}
which is the desired inflow formula~\cite{Witten:1983tw,Freed:2014iua}. 
Because $B_2$ and $\Lambda$ have integer periods and the mixed BF coupling $\tfrac{iN}{2\pi}\int_{M_4}B_2 \wedge f_2$ is present on the boundary, a large 1-form gauge transformation by $\Lambda$ identifies $\lambda\equiv\lambda+\tfrac{2\pi}{N}$ on $\partial M_5$. 
Invariance of $e^{iS}$ under such large transformations therefore requires 
\begin{equation}
\frac{1}{8\pi^2}\int_{M_4}\mathrm{Tr}(F\wedge F)+\frac{\kappa_G}{384\pi^2}\int_{M_4}\mathrm{Tr}(R\wedge R)\in N\mathbb Z,
\end{equation}
i.e. $Q+\kappa_GQ_G\in N\mathbb Z$. 
The 4D equations of motion derived from \eqref{eq:inflow} are consistent: variation in $B$ yields the flatness condition $f=0$ but does not trivialize $a$, which can carry nontrivial flat holonomy because $a$ is compact. 
Variation in $a$ gives the degree-matched identity $Y_3=N\,dB$ that underlies the 2-group structure and, upon exterior differentiation, reproduces the cohomological selection rule on closed $M_4$. 
The integrality statements follow from the chosen trace normalizations and the fact that on spin four-manifolds the second Chern number and the gravitational Pontryagin number are integers, ensuring single-valuedness of $e^{iS}$~\cite{Eguchi:1980jx}. 
Closed integer quantized five-forms may be added to Eq.~\eqref{eq:inflow} without affecting the boundary variation. 
This corrected inflow, together with Eq.~\eqref{eq:SBcoupling}, is a minimal, degree-exact realization of the $\mathbb Z_N$ gauging that enforces the selection rule and is standard in the modern treatment of generalized symmetries and anomaly inflow~\cite{Gaiotto:2014kfa,Kapustin:2013uxa,Witten:1983tw,Freed:2014iua}. 

\subsection{Relation to the Córdova-Freed-Seiberg coupling-space anomaly}

It is useful to make explicit how the topological sector above realizes the anomaly in
the space of coupling constants described by Córdova, Freed and Seiberg~\cite{Cordova:2019uob}.
In their framework the ordinary $\theta$ angle is promoted to a background $2\pi$ periodic
scalar field $\Theta(x)$, and the anomalous $2\pi$ periodicity of QCD is encoded in a
five dimensional invertible theory
\begin{equation}
S^{\rm CFS}_{5\mathrm{D}}
= \frac{i}{2\pi}\int_{M_5} d\Theta \wedge
\Big[ \frac{1}{8\pi^2}\,\mathrm{Tr}\,F\wedge F
      + \kappa_G\,\frac{1}{384\pi^2}\,\mathrm{Tr}\,R\wedge R \Big] ,
\end{equation}
whose boundary variation reproduces the mixed anomaly between the $\Theta$ shift symmetry, the one-form center symmetry and gravity. Our inflow action in Eq.~\eqref{eq4.10} has exactly the same structure, with the replacement
\begin{equation}
\frac{d\Theta}{2\pi}\;\longrightarrow\;\frac{a}{2\pi}\,,
\end{equation}
i.e. the CFS background one-form $d\Theta$ is promoted here to a dynamical compact
$U(1)$ one-form gauge field $a$. The four-dimensional BF coupling in Eq.~\eqref{eq:inflow}
then enforces the selection rule $Q+\kappa_G Q_G\in N\mathbb{Z}$ displayed in
Eq.~\eqref{eq:selection-rule}. Taken together, Eqs.~\eqref{eq:selection-rule}, \eqref{eq:inflow} and \eqref{eq4.10} show that the combined
``QCD + topological sector'' carries precisely the CFS-type inflow needed to make the
discrete subgroup of the $\theta$ shift gaugeable.

This perspective also clarifies why our construction does not contradict the statement
of Ref.~\cite{Cordova:2019uob} that $\theta$ periodicity is anomalous in pure QCD.
The anomaly means that, for QCD by itself, the transformation $\theta\to\theta+2\pi$
cannot be realized as an ordinary $0$-form global symmetry. It is only a symmetry of
the enlarged system that includes the five-dimensional invertible TQFT. In the present
work we do not claim that bare QCD has an anomaly free $U(1)_\theta$. Instead,
we explicitly include the CFS-type inflow sector (now written in terms of the dynamical
fields $a$ and $B$) and then gauge only a finite $\mathbb{Z}_N$ subgroup of the
anomalous $U(1)_\theta$. The additional topological degrees of freedom in
Eqs.~\eqref{eq:inflow}-\eqref{eq4.10} are precisely the higher-form fields whose variation cancels the CFS coupling-space anomaly. Operationally, gauging this anomaly-safe $\mathbb{Z}_N\subset U(1)_\theta$ enforces
the quantization condition $Q+\kappa_G Q_G\in N\mathbb{Z}$ and identifies
$\theta\sim\theta+2\pi/N$ in the full dynamical theory. The small effective angle
$|\bar\theta_{\mathrm{eff}}|\le\pi/N$ derived in Sec.~\ref{sec3} is therefore a consequence of working in an anomaly-completed version of the Córdova-Freed-Seiberg framework, not of ignoring the coupling-space anomaly of QCD.

Putting these ingredients together, the renormalized generating functional in curved space with the discrete projection included is
\begin{multline}
\label{eq:Z-gauged-curved}
Z[g,\theta_R,\theta_{G,R}] \;=\; \int \mathcal D A\,\mathcal D\psi\,\mathcal D\bar\psi\,\mathcal D a\,\mathcal D B \;\exp\!\Big\{-S_{\mathrm{YM}}[A,\psi,g] - S_{\mathrm{ct}}[g] + \\ i\theta_R \!\int q + i\theta_{G,R}\!\int q_G + \int \mathcal L_{\mathrm{top}}\Big\},
\end{multline}
where $S_{\mathrm{ct}}[g]$ denotes the standard purely gravitational counterterms proportional to $R^2$, $R_{\mu\nu}R^{\mu\nu}$, and $R_{\mu\nu\rho\sigma}R^{\mu\nu\rho\sigma}$, together with a cosmological counterterm, all of which are CP-even and do not affect the analysis. The vacuum energy density $E(\bar\theta,g)$ obtained from $Z$ is the lower envelope of the Yang-Mills branch energies evaluated on the $\mathbb Z_N$ orbit of the effective angle that multiplies the total Pontryagin charge, namely
\begin{equation}
\label{eq:envelope-curved}
E(\bar\theta,g) \;=\; \min_{m\in\mathbb Z}\, E_{\mathrm{YM}}\!\left(\bar\theta + \frac{2\pi m}{N},g\right),
\end{equation}
exactly as in flat space but with a mild, smooth dependence on the curvature that does not alter the discrete set of transition points at $\bar\theta=(2\ell+1)\pi/N$ where neighboring branches exchange dominance. Since the minimization always selects the representative of the $\mathbb Z_N$ orbit that lies in the principal cell $[-\pi/N,\pi/N]$, the physical expectation value of the CP-odd phase in the true ground state, which we denote by $\bar\theta_{\mathrm{eff}}$, necessarily satisfies
\begin{equation}
\label{eq:final-stability-bound}
|\bar\theta_{\mathrm{eff}}| \;\le\; \frac{\pi}{N},
\end{equation}
and this inequality is radiatively exact and gravitationally protected. Perturbative renormalization leaves $\theta_R$ and $\theta_{G,R}$ unchanged, dimension-six and higher operators cannot feed a continuous $\bar\theta$ because of the compact higher-form gauge invariance, the BF coefficient $N$ is integer-quantized and not renormalized, and nonperturbative quantum-gravity processes respect discrete gauge symmetries and therefore enforce the same selection rule \eqref{eq:selection-rule} rather than violate it \cite{Krauss:1988zc,Banks:2010zn,Harlow:2018tng}. Consequently, the discrete-$\theta$ projection yields a vacuum selection mechanism that is stable under all known radiative and gravitational corrections, and the bound \eqref{eq:final-stability-bound} persists as a gauge-protected property of the exact theory.

\section{Anomalies and Consistency Conditions}
\label{sec5}
In order to establish the precise conditions under which discrete or continuous symmetries may be consistently gauged, and to determine the obstructions that arise from quantum effects, it is necessary to exhibit a fully explicit derivation of gauge, gravitational, and mixed anomalies from the non-invariance of the functional measure. 
We also need to formulate the Ward-Takahashi identities together with their anomalous deformations, to impose the Wess-Zumino consistency conditions that characterize admissible anomalies, to relate the resulting expressions to differential-geometric descent and index theorems, and to verify anomaly cancellation in representative chiral or supersymmetric field theories~\cite{Wess:1971yu,Bardeen:1984pm,Zumino:1983ew,Bertlmann:1996Anomalies,Nakahara:2003nw,Alvarez-Gaume:1983ihn}. 
We shall do so by writing every relevant ingredient in a fully explicit and self-contained manner. 

Throughout this section the spacetime is a smooth, oriented four-dimensional Lorentzian manifold with metric $g_{\mu\nu}$ of signature $(-,+,+,+)$. 
The Levi-Civita tensor is $\epsilon^{\mu\nu\rho\sigma}$ normalized by $\epsilon^{0123}=+1$. 
The Hodge dual is $\tilde F^{\mu\nu} \equiv \tfrac{1}{2}\,\epsilon^{\mu\nu\rho\sigma} F_{\rho\sigma}$ for any two-form $F_{\mu\nu}$. 
The vierbein is $e_\mu{}^a$ with spin connection $\omega_\mu{}^{ab}$ and curvature two-form $\mathcal R^{ab} \equiv \tfrac{1}{2} R^{ab}{}_{\mu\nu}\,dx^\mu\wedge dx^\nu$. 

The gauge group $G$ is compact with Lie algebra $\mathfrak g$ spanned by Hermitian generators $T^a$ satisfying $[T^a,T^b]=i f^{abc} T^c$ and $\mathrm{tr}_R(T^a T^b)=\tfrac{1}{2}\,\delta^{ab}$ in any representation $R$. 
The completely symmetric invariant is $d^{abc}_R \equiv \tfrac{1}{2}\,\mathrm{tr}_R\!\big(T^a\{T^b,T^c\}\big)$ and reduces to $d^{abc}$ when $R$ is the defining representation. 
The gauge potential is $A_\mu \equiv A_\mu^a T^a$. 
Its field strength is 
\begin{equation}
    F_{\mu\nu} \equiv \partial_\mu A_\nu-\partial_\nu A_\mu - i[A_\mu,A_\nu],
\end{equation}
and the covariant derivative acting on a field in representation $R$ is 
\begin{equation}
    D_\mu \equiv \partial_\mu - i A_\mu + \tfrac{1}{4}\,\omega_\mu{}^{ab}\,\gamma_{ab}
\end{equation}
with $\gamma_{ab}\equiv \tfrac{1}{2}[\gamma_a,\gamma_b]$. 
We use $\gamma_5 \equiv i \gamma^0\gamma^1\gamma^2\gamma^3$ and the chiral projectors $P_{L,R} \equiv \tfrac{1}{2}(1\mp \gamma_5)$. 

For a set of left-handed Weyl fermions $\psi_i$ in representations $R_i$ the classical action in a fixed background $(A_\mu,\omega_\mu{}^{ab},g_{\mu\nu})$ is
\begin{equation}
\label{eq:Weyl-action}
S[\psi,\bar\psi,A,\omega,g] \;=\; \sum_i \int d^4x\,\sqrt{-g}\;\bar\psi_i\, i\slashed{D}\,P_L\,\psi_i, 
\qquad \slashed{D}\equiv \gamma^\mu D_\mu,
\end{equation}
and the corresponding generating functional is $Z[A,\omega,g]\equiv \int \prod_i \mathcal D\psi_i\mathcal D\bar\psi_i \,\exp\{-S\}$, $W\equiv -i\log Z$. 

The variation of $W$ under an infinitesimal gauge transformation with parameter $\epsilon(x)=\epsilon^a(x) T^a$ acting as $\delta_\epsilon A_\mu = D_\mu \epsilon \equiv \partial_\mu \epsilon - i[A_\mu,\epsilon]$ and $\delta_\epsilon \psi = i\epsilon\,\psi$ is given by two contributions. 
There is a classical piece produced by the variation of the action and a quantum piece produced by the Jacobian of the fermion path-integral measure, namely
\begin{equation}
\label{eq:deltaW-split}
\delta_\epsilon W \;=\; \int d^4x\,\sqrt{-g}\;\epsilon^a(x)\, D_\mu \big\langle J^{\mu}_a(x)\big\rangle_{\!W} \;-\; i\log J[\epsilon;A,\omega,g],
\quad J^\mu_a(x) \equiv \sum_i \bar\psi_i \gamma^\mu T^a P_L \psi_i,
\end{equation}
where $\langle\cdots\rangle_W$ denotes the connected one-point function in the presence of sources. 
The anomaly is precisely the failure of the Jacobian to equal unity. 
It can be evaluated by the Fujikawa method by expanding the fields in a complete orthonormal basis of eigenmodes of the Dirac operator and regulating the trace with a gauge- and diffeomorphism-covariant heat kernel~\cite{Fujikawa:1980eg,Vassilevich:2003xt,Bilal:2008qx,Harvey:2005it}. 

Writing $\psi(x)=\sum_n a_n \varphi_n(x)$ with $i\slashed{D}\,\varphi_n=\lambda_n \varphi_n$ and the formal measure $\prod_n da_n d\bar a_n$, the transformation $\psi\mapsto\psi'=(1+i\epsilon) \psi$ induces $da_n \mapsto da_n + i\sum_m (\epsilon)_{nm}\,da_m$. 
The Jacobian is
\begin{equation}
\label{eq:Jacobian}
\log J[\epsilon;A,\omega,g] \;=\; -2i \lim_{M\to\infty}\sum_i \int d^4x\,\sqrt{-g}\;\mathrm{tr}_{R_i}\!\Big(\epsilon(x) P_L \,\langle x|\, e^{-(i\slashed{D})^2/M^2}\,|x\rangle\Big),
\end{equation}
where the factor of $2$ arises because $\bar\psi$ and $\psi$ transform with opposite phases and $P_L^2=P_L$. 

Using 
\begin{equation}
  (i\slashed{D})^2 = -D_\mu D^\mu + \tfrac{i}{4}[\gamma^\mu,\gamma^\nu]F_{\mu\nu} + \tfrac{1}{8}\,[\gamma^\mu,\gamma^\nu] R_{\mu\nu\rho\sigma}\Sigma^{\rho\sigma}\:\:\text{with}\:\: \Sigma^{\rho\sigma}\equiv \tfrac{i}{4}[\gamma^\rho,\gamma^\sigma]  
\end{equation}
and the Schwinger-DeWitt heat-kernel expansion, one finds that only the $a_2$ Seeley coefficient contributes after tracing with $P_L=\tfrac{1}{2}(1-\gamma_5)$~\cite{Vassilevich:2003xt,Bertlmann:1996Anomalies}. 
The $\gamma_5$-odd piece yields the gauge, gravitational, and mixed anomalies~\cite{Adler:1969gk,Bell:1969ts,Bardeen:1969md,Alvarez-Gaume:1983ihn,Zumino:1983ew}. 
After a standard computation that keeps track of all symmetry factors and tensor contractions, one obtains
\begin{equation}
\label{eq:consistent-anomaly-local}
\mathcal A^a(x) \;\equiv\; \frac{\delta_\epsilon W}{\delta \epsilon^a(x)}\bigg|_{\epsilon=0}
\;=\; \sum_i \left[\frac{1}{24\pi^2}\, d^{abc}_{R_i}\, \tfrac{1}{2}\,\epsilon^{\mu\nu\rho\sigma}\,F^b_{\mu\nu} F^c_{\rho\sigma} \;+\; \frac{q_{R_i}}{384\pi^2}\,\epsilon^{\mu\nu\rho\sigma} R^{\alpha}{}_{\beta\mu\nu} R^{\beta}{}_{\alpha\rho\sigma}\,\delta^{a0}\right],
\end{equation}
where $q_{R_i}$ is the $U(1)$ charge if an abelian factor is present (we take the generator $T^0 \equiv \mathbf 1$ for the $U(1)$) and vanishes for purely non-Abelian simple factors. 
The first term is the purely non-Abelian consistent gauge anomaly. 
The second is the mixed $U(1)$-gravitational consistent anomaly. 
There is no pure gravitational diffeomorphism anomaly in four dimensions, in accord with the general $4k{+}2$-dimensional criterion~\cite{Alvarez-Gaume:1983ihn}. 

Substituting \eqref{eq:consistent-anomaly-local} into \eqref{eq:deltaW-split} yields the anomalous Ward-Takahashi identity
\begin{equation}
\label{eq:WTI-anomalous}
D_\mu \big\langle J^\mu_a\big\rangle_{\!W} \;=\; \mathcal A_a \;\equiv\; \sum_i \left[\frac{1}{24\pi^2}\, d^{abc}_{R_i}\, \tfrac{1}{2}\,\epsilon^{\mu\nu\rho\sigma}\,F^b_{\mu\nu} F^c_{\rho\sigma} \;+\; \frac{q_{R_i}}{384\pi^2}\,\epsilon^{\mu\nu\rho\sigma} R^{\alpha}{}_{\beta\mu\nu} R^{\beta}{}_{\alpha\rho\sigma}\,\delta_{a0}\right],
\end{equation}
which is the local form of the consistent anomaly~\cite{Bardeen:1984pm,Bertlmann:1996Anomalies,Bilal:2008qx,Harvey:2005it}. 

For a Dirac fermion with vector gauge coupling the gauge variation of the vector current cancels between the two chiralities, while the axial current acquires the Adler-Bell-Jackiw divergence 
\begin{equation}
\partial_\mu J_5^\mu = \tfrac{1}{16\pi^2}\,\mathrm{tr}(F_{\mu\nu}\tilde F^{\mu\nu}) + \tfrac{1}{384\pi^2}\, R_{\mu\nu\rho\sigma}\tilde R^{\mu\nu\rho\sigma},
\end{equation}
reproducing the coefficient $1/(16\pi^2)$ and the gravitational contribution $1/(384\pi^2)$ as obtained originally from triangle diagrams~\cite{Adler:1969gk,Bell:1969ts,Bardeen:1969md} and from the measure analysis~\cite{Fujikawa:1980eg}. 

The anomalies in \eqref{eq:WTI-anomalous} are called ``consistent'' because they obey the Wess-Zumino consistency conditions, which express the integrability of the anomalous Ward identities under successive symmetry transformations \cite{Wess:1971yu}. Writing $\delta_{\epsilon} W = \int d^4x\,\sqrt{-g}\;\epsilon^a \mathcal A_a$ and using the Lie-algebraic composition $\epsilon_3^a = f^{abc}\epsilon_1^b\epsilon_2^c$, the condition reads
\begin{equation}
\label{eq:WZ-condition}
\delta_{\epsilon_1}\!\left(\int \epsilon_2^a \mathcal A_a\right) \;-\; \delta_{\epsilon_2}\!\left(\int \epsilon_1^a \mathcal A_a\right) \;=\; \int f^{abc} \epsilon_1^b \epsilon_2^c\, \mathcal A_a,
\end{equation}
and it is satisfied identically by \eqref{eq:consistent-anomaly-local} because the right-hand side is generated by the variation of the Chern-Simons descendant that underlies the consistent anomaly. The covariant form of the anomaly, in which the right-hand side is a gauge-covariant local polynomial but does not satisfy \eqref{eq:WZ-condition} by itself, is obtained by redefining the current by a Bardeen-Zumino functional $K^\mu_a$ so that $J^\mu_{a,\mathrm{cov}} \equiv J^\mu_{a,\mathrm{cons}} + K^\mu_a$ transforms covariantly \cite{Bardeen:1984pm}. In differential-form notation with $A=A_\mu dx^\mu$ and $F=dA-iA\wedge A$ and with the trace $\mathrm{tr}_R$ understood, the Bardeen-Zumino current can be written explicitly as
\begin{equation}
\label{eq:BZ-current}
K^\mu_a \;=\; \frac{1}{24\pi^2}\,\epsilon^{\mu\nu\rho\sigma}\,\mathrm{tr}_R\!\left[T^a\left(A_\nu\partial_\rho A_\sigma - \frac{i}{2} A_\nu A_\rho A_\sigma\right)\right],
\end{equation}
which yields the covariant nonconservation law
\begin{equation}
\label{eq:covariant-anomaly}
D_\mu J^\mu_{a,\mathrm{cov}} \;=\; \frac{1}{32\pi^2}\,\epsilon^{\mu\nu\rho\sigma}\,\mathrm{tr}_R\!\big(T^a F_{\mu\nu} F_{\rho\sigma}\big)
\end{equation}
for each left-handed Weyl fermion, equivalent to \eqref{eq:WTI-anomalous} after the addition of \eqref{eq:BZ-current} and consistent with the Stora-Zumino descent formalism \cite{Zumino:1983ew,Bardeen:1984pm}. For a single left-handed fermion of $U(1)$ charge $q$ one recovers from \eqref{eq:WTI-anomalous} the abelian consistent anomaly 
\begin{equation}
\partial_\mu J^\mu_{\mathrm{cons}} = \tfrac{q^3}{96\pi^2}\,F_{\mu\nu}\tilde F^{\mu\nu}
\end{equation}
and from \eqref{eq:covariant-anomaly} the abelian covariant anomaly 
\begin{equation}
\partial_\mu J^\mu_{\mathrm{cov}} = \tfrac{q^3}{32\pi^2}\,F_{\mu\nu}\tilde F^{\mu\nu},
\end{equation}
together with the mixed $U(1)$-gravitational consistent anomaly 
\begin{equation}
\partial_\mu J^\mu_{\mathrm{cons}} = \tfrac{q}{384\pi^2}\,R_{\mu\nu\rho\sigma}\tilde R^{\mu\nu\rho\sigma},
\end{equation}
all with their full numerical factors as required by the underlying triangle diagrams and the heat-kernel computation \cite{Bardeen:1969md,Alvarez-Gaume:1983ihn,Fujikawa:1980eg}. 

The geometric and topological origin of these expressions is encoded by the anomaly polynomial and the descent relations. For $2n{+}2=6$ the anomaly polynomial for a collection of left-handed Weyl fermions is the $6$-form
\begin{equation}
\label{eq:anomaly-polynomial}
I_6 \;=\; \sum_i \left[\frac{1}{6}\,\mathrm{tr}_{R_i}\!\left(\frac{iF}{2\pi}\right)^3 \;-\; \frac{1}{24}\,p_1(T)\,\mathrm{tr}_{R_i}\!\left(\frac{iF}{2\pi}\right)\right],
\quad p_1(T) \equiv -\,\frac{1}{2}\,\frac{1}{(2\pi)^2}\,\mathrm{tr}\big(\mathcal R\wedge \mathcal R\big),
\end{equation}
where $p_1(T)$ is the first Pontryagin class of the tangent bundle and the trace in $p_1$ is in the vector representation of $SO(3,1)$ with $\mathrm{tr}(\mathcal R\wedge \mathcal R)=\mathcal R^{ab}\wedge \mathcal R^{ba}$. The derivation of \eqref{eq:anomaly-polynomial} follows from the Atiyah-Singer index theorem applied to the chiral Dirac operator, 
\begin{equation}
I_{2n+2} = \big[\widehat A(\mathcal R)\,\mathrm{ch}(iF/2\pi)\big]_{2n+2}
\end{equation}
with $\widehat A(\mathcal R)=1-\tfrac{1}{24}p_1(T)+\cdots$ and $\mathrm{ch}(iF/2\pi)=\mathrm{tr}\,\exp(iF/2\pi)$, and it therefore computes the index density that governs the net spectral flow under adiabatic deformations \cite{Atiyah:1968mp,Alvarez-Gaume:1983ihn}. By definition there exists a $5$-form Chern-Simons transgressor $I_5^{(0)}$ such that $d I_5^{(0)}=I_6$. For the purely gauge part one may write
\begin{multline}
\label{eq:CS-five}
I_5^{(0)}(A) \;=\; \sum_i \frac{1}{(2\pi)^3}\,\frac{i^3}{3!}\,\mathrm{tr}_{R_i}\!\left(A\,dA\,dA \;+\; \frac{3}{2}\,A^3 dA \;+\; \frac{3}{5}\,A^5\right),
\\ d\,\mathrm{tr}\!\left(A\,dA\,dA + \frac{3}{2}\,A^3 dA + \frac{3}{5}\,A^5\right) \;=\; \mathrm{tr}(F^3),
\end{multline}
and an analogous Lorentz-Chern-Simons form $I_5^{(0)}(\omega)$ with $\mathrm{tr}(\mathcal R^3)$ replaced by $p_1(T)\,\mathrm{tr}(iF/2\pi)$. Under a gauge variation $\delta_\epsilon A = D\epsilon$ one finds $\delta_\epsilon I_5^{(0)} = d I_4^{(1)}(\epsilon,A)$ with
\begin{equation}
\label{eq:I41}
I_4^{(1)}(\epsilon,A) \;=\; \sum_i \frac{1}{(2\pi)^3}\,\frac{i^3}{2}\,\mathrm{tr}_{R_i}\!\left(\epsilon\, d\big(A\,dA + \tfrac{1}{2}\,A^3\big)\right),
\end{equation}
and the consistent anomaly is obtained by the descent formula
\begin{equation}
\label{eq:descent-anomaly}
\delta_\epsilon W \;=\; 2\pi i \int_{M_4} I_4^{(1)}(\epsilon,A) \;=\; \int d^4x\,\sqrt{-g}\;\epsilon^a(x)\,\mathcal A_a(x),
\end{equation}
which reproduces \eqref{eq:consistent-anomaly-local}. The Wess-Zumino condition \eqref{eq:WZ-condition} follows from the nilpotency of the BRST differential acting on the descent tower and is, therefore, equivalent to the statement that $I_6$ is a closed characteristic class in de Rham cohomology \cite{Wess:1971yu,Zumino:1983ew}. The inflow mechanism \cite{Callan:1984sa} realizes \eqref{eq:descent-anomaly} as the boundary variation of a five-dimensional topological action $S_{\mathrm{inflow}} = 2\pi i \int_{M_5} I_5^{(0)}$. The bulk variation $\delta S_{\mathrm{inflow}} = - 2\pi i \int_{M_5} d I_4^{(1)} = - 2\pi i \int_{\partial M_5} I_4^{(1)}$ cancels the boundary anomaly, providing a geometric interpretation of anomaly inflow and of the quantization of coefficients in $I_6$ \cite{Callan:1984sa,Alvarez-Gaume:1983ihn}. 

The cancellation of local gauge anomalies is a necessary condition for the consistency of a four-dimensional quantum gauge theory. In practice one demands that the consistent anomaly \eqref{eq:WTI-anomalous} vanishes for each gauged factor. For a non-Abelian simple factor $G$, this implies $\sum_i d^{abc}_{R_i}=0$ for all $a,b,c$, while for an abelian $U(1)$ one must have $\sum_i q_{R_i}^3=0$ for the cubic anomaly and $\sum_i q_{R_i}=0$ for the mixed $U(1)$-gravitational anomaly, as well as $\sum_i q_{R_i}\,T_2(R_i)=0$ for each mixed $U(1)$-$G^2$ anomaly where $T_2(R)\equiv \mathrm{tr}_R(T^a T^a)$ is the Dynkin index \cite{Bardeen:1969md}. A fully explicit check can be given for a single generation of the Standard Model written solely in terms of left-handed Weyl fermions, namely 
\begin{equation}
\begin{aligned}
    Q_L\sim (\mathbf 3,\mathbf 2)_{1/6}\\
    u^c\sim (\bar{\mathbf 3},\mathbf 1)_{-2/3}\\
    d^c\sim (\bar{\mathbf 3},\mathbf 1)_{1/3}\\
    L_L\sim (\mathbf 1,\mathbf 2)_{-1/2}\\
    e^c\sim (\mathbf 1,\mathbf 1)_{+1},
\end{aligned}
\end{equation}
where the subscripts denote hypercharge $Y$. The purely non-Abelian $[SU(3)_c]^3$ anomaly cancels because the cubic index is proportional to the difference between fundamentals and anti-fundamentals counted with multiplicity of $SU(2)$ components, giving $2\times A(\mathbf 3) - A(\bar{\mathbf 3}) - A(\bar{\mathbf 3})=0$ with $A(\mathbf 3)=-A(\bar{\mathbf 3})$ \cite{Bardeen:1969md}. The $[SU(2)_L]^3$ anomaly vanishes because $SU(2)$ is pseudoreal and $d^{abc}=0$. The mixed $U(1)_Y$-$[SU(2)_L]^2$ anomaly is proportional to $\sum_i Y_i T_2(R_i^{SU(2)})$ with $T_2(\mathbf 2)=\tfrac{1}{2}$, and including the color multiplicity for $Q_L$ one finds
\begin{equation}
3\cdot \tfrac{1}{6}\cdot\tfrac{1}{2} + 1\cdot (-\tfrac{1}{2})\cdot\tfrac{1}{2} = \tfrac{1}{4}-\tfrac{1}{4}=0.
\end{equation}
The mixed $U(1)_Y$-$[SU(3)_c]^2$ anomaly is proportional to $\sum_i Y_i T_2(R_i^{SU(3)})$ with $T_2(\mathbf 3)=T_2(\bar{\mathbf 3})=\tfrac{1}{2}$ and multiplicity equal to the number of $SU(2)$ components, giving 
\begin{equation}
2\cdot \tfrac{1}{6}\cdot\tfrac{1}{2} + 1\cdot (-\tfrac{2}{3})\cdot\tfrac{1}{2} + 1\cdot (\tfrac{1}{3})\cdot\tfrac{1}{2} = \tfrac{1}{6}-\tfrac{1}{3}+\tfrac{1}{6}=0. 
\end{equation}
The cubic $[U(1)_Y]^3$ anomaly, computed by summing $Y^3$ over all left-handed Weyl fields with multiplicities, yields 
\begin{equation}
6\cdot (\tfrac{1}{6})^3 + 3\cdot (-\tfrac{2}{3})^3 + 3\cdot (\tfrac{1}{3})^3 + 2\cdot (-\tfrac{1}{2})^3 + 1\cdot (1)^3 = \tfrac{1}{36}-\tfrac{8}{9}+\tfrac{1}{9}-\tfrac{1}{4}+1=0,
\end{equation}
while the mixed $U(1)_Y$-gravitational anomaly, proportional to $\sum_i Y_i$, gives 
\begin{equation}
6\cdot \tfrac{1}{6} + 3\cdot (-\tfrac{2}{3}) + 3\cdot (\tfrac{1}{3}) + 2\cdot (-\tfrac{1}{2}) + 1\cdot (1)=0.
\end{equation} 
These equalities are exact and confirm that all local anomalies cancel for one generation; with three generations the cancellations persist generation by generation \cite{Bardeen:1969md}. As a second illustrative example, a chiral supersymmetric grand-unified theory with gauge group $SU(5)$ and one family of chiral multiplets in $\mathbf{10}\oplus \bar{\mathbf 5}$ is anomaly free because the cubic index satisfies 
\begin{equation}
d^{abc}_{\mathbf{10}} + d^{abc}_{\bar{\mathbf 5}} = \big(A(\mathbf{10}) - A(\mathbf 5)\big)\,d^{abc}_{\mathrm{fund}} = 0
\end{equation}
with $A(\mathbf{10})=A(\mathbf 5)$ in $SU(5)$, and the supersymmetric partners do not modify the chiral gauge anomaly since gaugini are in real representations and chiral multiplet fermions carry exactly the same gauge quantum numbers as their scalar partners \cite{Bardeen:1969md,Zumino:1983ew}. 

The foregoing analysis, grounded in the non-invariance of the functional measure \cite{Fujikawa:1980eg}, the anomalous Ward-Takahashi identities \cite{Adler:1969gk,Bardeen:1969md}, the Wess-Zumino consistency conditions \cite{Wess:1971yu}, and the differential-geometric descent tied to index theory and anomaly inflow \cite{Atiyah:1968mp,Zumino:1983ew,Callan:1984sa,Alvarez-Gaume:1983ihn}, provides a complete and self-contained classification of admissible local anomalies in four-dimensional quantum field theory. In particular, in modern language of generalized global symmetries, the same structure controls mixed 't~Hooft anomalies between ordinary $0$-form symmetries and higher-form symmetries, which are likewise encoded by higher-dimensional inflow actions whose boundary variations reproduce the corresponding consistent anomalies and whose quantization conditions enforce the integrality constraints necessary for gauging discrete subgroups \cite{Gaiotto:2014kfa}. The explicit coefficients in \eqref{eq:WTI-anomalous}-\eqref{eq:covariant-anomaly} ensure that once the anomaly cancellation conditions are satisfied, the Ward identities hold exactly in the quantum theory, while any failure of these conditions manifests as a precise, cohomologically nontrivial obstruction that cannot be removed by local counterterms without violating either locality or gauge invariance; conversely, when the anomaly is cohomologically trivial, the Bardeen-Zumino redefinition \eqref{eq:BZ-current} restores manifest covariance and resolves scheme dependence without altering physical observables \cite{Bardeen:1984pm,Zumino:1983ew}.

\section{Phenomenology}
The most immediate phenomenological consequence of the D$\theta$P mechanism is an extremely suppressed neutron electric dipole moment (EDM)~\cite{Pospelov:2005pr,Engel:2013lsa,Chupp:2017rkp,Yamanaka:2016umw}. 
In QCD, a nonzero $\bar{\theta}$ induces CP-violating interactions that give rise to an EDM for the neutron. 
At leading order in chiral perturbation theory, one finds that the neutron EDM is linear in $\bar{\theta}$ for small angles~\cite{Crewther:1979pi,Bsaisou:2014zwa,deVries:2015una}. 
In particular, $\bar{\theta}$ can be rotated into the phases of the quark mass matrix, inducing a CP-odd pion-nucleon coupling and other interactions~\cite{Kim:2008hd}. 
As first estimated by Baluni and Crewther \emph{et al.}\ in 1979, the neutron EDM from a $\bar{\theta}$ term is on the order of\footnote{This range accounts for hadronic uncertainties in the nucleon matrix elements and coupling constants. Modern analyses including lattice QCD input confirm the same order of magnitude~\cite{Pospelov:2005pr,Bhattacharya:2015wna,Shindler:2021bcx,Engel:2013lsa,Chupp:2017rkp}.} 
$1$–$3\times10^{-16}\,\bar{\theta}\,e\,\text{cm}$~\cite{Baluni:1978rf,Crewther:1979pi}. 
More recent computations using chiral effective field theory, together with lattice determinations of nucleon mass splittings, have refined this to 
$d_n \approx (2.4\pm0.9)\times10^{-16}\,\bar{\theta}\,e\,\text{cm}$~\cite{deVries:2015una,Bsaisou:2014zwa}. 
This result can be understood as arising from a CP-odd pion-nucleon coupling $\bar g_{\pi NN}$ (notably, the coupling of the neutral pion to nucleons) induced by $\bar{\theta}$~\cite{deVries:2015una}. 
Writing the effective interaction as 
$\mathcal{L}_{\pi N} \supset \bar g_{\pi NN}\,\bar N i\gamma_5 (\tau^3) N\,\pi^0$ 
(which splits into $\bar g_{\pi^0 p p}\bar p i\gamma_5 p\,\pi^0$ and $\bar g_{\pi^0 n n}\bar n i\gamma_5 n\,\pi^0$ for proton and neutron), one finds $\bar{g}_{\pi NN}$ is proportional to the quark mass difference $m_d-m_u$ and to $\bar{\theta}$~\cite{deVries:2015una}. 

At leading order, $\bar{g}_{\pi NN}$ can be related to the strong part of the nucleon mass splitting $\delta m_N^{(\text{str})}= (m_n - m_p)_{\text{QCD}}$ and the pion decay constant $f_\pi$. 
One finds 
\begin{equation}
\bar g_{\pi NN}\approx -\frac{\delta m_N^{(\text{str})}}{2f_\pi}\,\frac{1-\epsilon^2}{2\epsilon}\,\bar{\theta},
\end{equation}
where $\epsilon=(m_d-m_u)/(m_d+m_u)$. 
Using the physical values $\delta m_N^{(\text{str})}\approx 2.5~\text{MeV}$ and $\epsilon\approx0.35$ gives $\bar{g}_{\pi NN}\approx 1.5\times10^{-2}\,\bar{\theta}$~\cite{deVries:2015una,FlavourLatticeAveragingGroupFLAG:2021npn}. 
This coupling, in turn, induces a neutron EDM through a one-loop pion graph (sometimes called the “TDL” mechanism)~\cite{Pospelov:2005pr,Bsaisou:2014zwa}. 
Summing the pion-loop contribution and possible short-distance pieces yields the stated $d_n\simeq 10^{-16}\bar{\theta}\,e\,\text{cm}$ scale~\cite{Pospelov:2005pr,Engel:2013lsa,Chupp:2017rkp}. 

We can equivalently derive this scaling by differentiating the QCD vacuum energy with respect to $\bar{\theta}$. 
A small $\bar{\theta}$ shifts the vacuum energy by $E(\bar{\theta}) \approx \frac{1}{2}\chi_t\,\bar{\theta}^2$ (for $\bar{\theta}$ near $0$), where $\chi_t = \partial^2 E/\partial \theta^2|_{\theta=0}$ is the topological susceptibility. 
This implies $\partial E/\partial \bar{\theta} = \chi_t \bar{\theta}$, which is related (by the Feynman-Hellmann theorem) to the expectation value of the topological charge $Q$ as $\langle Q \rangle_{\bar{\theta}}\approx -\chi_t \bar{\theta}$~\cite{DiVecchia:1980yfw,Shindler:2021bcx,Shindler:2015aqa}. 
In a nucleon state, this translates into a CP-odd nucleon matrix element and ultimately the neutron EDM. 
Indeed, using $\chi_t \approx (75~\text{MeV})^4$ for physical QCD with $m_{u,d,s}$ reproduces the above estimate $d_n\sim10^{-16}\bar{\theta}\,e\,\text{cm}$~\cite{Borsanyi:2015cka,Petreczky:2016vrs}. 
We note that if one light quark mass were zero, $\chi_t$ would vanish and $\bar{\theta}$ would be unobservable. 
This is the well-known fact that if, say, $m_u=0$, a $\theta$ angle can be rotated away~\cite{Baluni:1978rf,Crewther:1979pi}, but empirical quark mass ratios $\frac{m_u}{m_d}\approx0.5$, $\frac{m_s}{m_d}\approx20$ prevent such a cancellation~\cite{FlavourLatticeAveragingGroupFLAG:2021npn,ParticleDataGroup:2022pth}. 

Crucially, the D$\theta$P framework forces the effective $\bar{\theta}$ angle to lie in the interval $[-\pi/N,\pi/N]$, with $\bar{\theta}$ dynamically selected to be the closest possible value to zero. 
In other words, the physical vacuum is projected onto the branch with the smallest $|\bar{\theta}|$. 
Thus, instead of being an arbitrary parameter (potentially of order $1$), $\bar{\theta}$ becomes an effective output of the theory, bounded by $\pi/N$. 
At large $N$, this results in an enormous suppression of all $\bar{\theta}$-dependent observables~\cite{Pospelov:2005pr,Kim:2008hd}. 
The neutron EDM is therefore suppressed to
\begin{equation}
|d_n| \lesssim (1\text{ to }3)\times10^{-16}\,(\pi/N)\,e\cdot\text{cm}\,,
\label{eq:dn-suppression}
\end{equation}
using the range given above. 
For example, taking $N=10^{12}$ (as a demonstrative large value within reach of clockwork UV completions) gives $|\bar{\theta}|_{\rm eff}\le \pi/10^{12}\sim3\times10^{-12}$. 
This yields $d_n\sim10^{-27}$ to $10^{-28}$~$e$·cm. 
This is five to six orders of magnitude below the current experimental limit ($|d_n|<1.8\times10^{-26}$~$e$·cm, $90\%$~C.L.~\cite{Abel:2020pzs}) and is also well below the sensitivity of any near-future EDM searches~\cite{Yamanaka:2016umw,Graham:2015ouw}. 
In fact, to be certain that $d_n$ lies below $10^{-27}$~$e$·cm (one order of magnitude under the present bound), we require $|\bar{\theta}| \lesssim 10^{-11}$. 
The D$\theta$P bound $|\bar{\theta}|\le \pi/N$ then implies
\begin{equation}
N \gtrsim \frac{\pi}{10^{-11}} \sim 3\times10^{10},
\end{equation}
i.e.\ $N$ of order $10^{10}$ or higher. 
In Fig.~\ref{fig:EDMplot} we plot the predicted neutron EDM as a function of $N$ in the D$\theta$P scenario, compared to the current experimental upper limit. 
Even for $N$ in the ballpark of $10^9$ to $10^{10}$, the neutron EDM is safely below present bounds, and for the extremely large $N$ (sometimes attainable via clockwork exponentially large effective $N$) the value of $d_n$ becomes utterly negligible~\cite{DiLuzio:2020wdo}. 

Importantly, this suppression of $\bar{\theta}$ and $d_n$ is radiatively and gravitationally stable. 
The smallness of $\bar{\theta}$ is enforced by a discrete gauge symmetry rather than a global symmetry. 
This means that no operator in the theory can generate a $\bar{\theta}$ term in the low-energy Lagrangian unless it explicitly breaks the $Z_N$ invariance~\cite{Banks:2010zn,Harlow:2018tng}. 
Higher-dimensional operators (for example from Planck-scale physics) that might have induced a dangerous $\theta$ term are likewise forbidden by the discrete gauge symmetry, up to exponentially small tunneling effects (see below)~\cite{Holman:1992us,Kamionkowski:1992mf,Barr:1992qq}. 
In contrast to Peccei–Quinn axion models, where Planck-suppressed $U(1)_{\rm PQ}$-violating operators of dimension $d$ can reintroduce a $\theta$ of order $f_a^{d-4}/M_{\rm Pl}^{d-4}$, here any would-be $\theta$ shift is quantized and constrained by gauge invariance~\cite{Kim:2008hd,DiLuzio:2020wdo}. 
Thus the D$\theta$P mechanism solves the strong CP problem without leaving behind a “Theta loaf” of residual $\bar{\theta}$. The suppression of CP violation is as exact as the $Z_N$ symmetry itself~\cite{Banks:2010zn,Harlow:2018tng}. 
Therefore, the neutron EDM in this framework is an extremely small number, scaling as $d_n\propto1/N$ (up to hadronic uncertainties), and for $N\gg1$ it falls far below any foreseeable experimental sensitivity~\cite{Pospelov:2005pr,Engel:2013lsa,Yamanaka:2016umw}.

{\begin{figure}[t]
\centering
\begin{tikzpicture}
\begin{loglogaxis}[
width=0.8\textwidth,
xlabel={$N$ (discrete symmetry order)},
ylabel={$|d_n|$ [e$\cdot$cm]},
xmin=1, xmax=1e12,
ymin=1e-28, ymax=1e-15,
ytick={1e-28,1e-26,1e-24,1e-22,1e-20,1e-18,1e-16},
legend pos=south west,
legend cell align=left
]
\addplot[name path=lower, draw=none, domain=1:1e12, samples=201] {1e-16*pi/x};
\addplot[name path=upper, draw=none, domain=1:1e12, samples=201] {3e-16*pi/x};
\addplot[blue!20] fill between[of=upper and lower];
\addplot[blue!70!black, thick, domain=1:1e12, samples=201] {2.4e-16*pi/x};
\addlegendentry{D$\theta$P prediction (neutron EDM)};
\addplot[red!70!black, dash pattern=on 5pt off 3pt] coordinates {(1,1.8e-26) (1e12,1.8e-26)};
\addlegendentry{Current exp.~upper bound \cite{Abel:2020pzs}};
\end{loglogaxis}
\end{tikzpicture}
\caption{Predicted neutron EDM $|d_n|$ as a function of the discrete symmetry order $N$ in the D$\theta$P framework (blue band, reflecting the range $d_n\approx(1$--$3)\times10^{-16}\,\bar{\theta}\,e\,$cm). The red dashed line indicates the current experimental $90\%$~C.L.~limit $|d_n|<1.8\times10^{-26}$~$e$·cm~\cite{Abel:2020pzs}. For $N\gtrsim10^{10}$ the D$\theta$P prediction lies well below the current bound.}
\label{fig:EDMplot}
\end{figure}
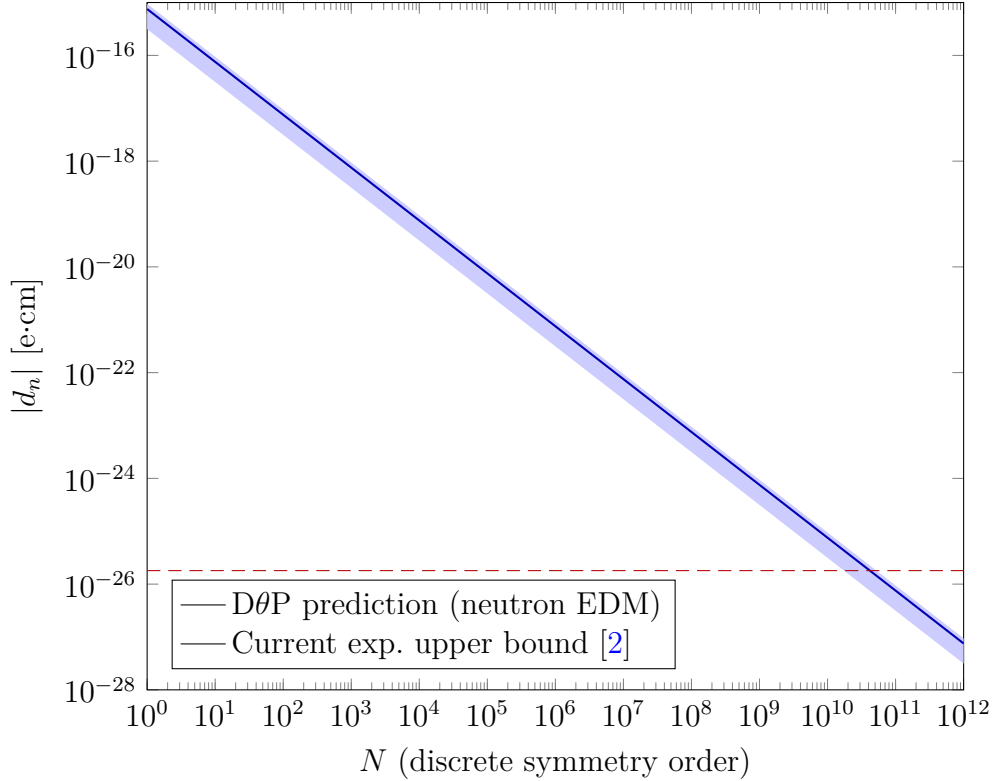}


\subsection{Absence of axion-like signatures.} 
An important distinction of the D$\theta$P solution is that it does not involve a light axion. The traditional Peccei-Quinn mechanism introduces a pseudo Nambu-Goldstone boson (the axion) which manifests as a new light particle with specific couplings to photons, nucleons, and leptons. Extensive experimental programs and astrophysical observations have placed strong limits on such an axion. In the D$\theta$P framework, however, $\bar{\theta}$ is fixed by a discrete gauge selection rather than by the vacuum expectation of a dynamical field. There is no continuous $U(1)$ symmetry being broken, and thus no associated Goldstone boson in the low-energy spectrum. Effectively, one can say the would-be axion has been ``gauged away'' or made heavy by the discrete coupling. Any would be massless mode is absent from the physical spectrum. As a result, all experimental and observational signals that would have been caused by a QCD axion are null by construction.

In particular, direct searches for axions via their two-photon coupling (e.g.\ cavity haloscope experiments like ADMX, HAYSTAC, CAPP, which look for axion-to-photon conversion in a magnetic field) will not find a QCD axion in this scenario, since no such particle exists. Similarly, solar axion searches (helioscopes such as CAST and IAXO that aim to detect axions produced in the Sun) will see nothing related to QCD axions. Stellar cooling constraints, which normally provide strong limits on light axions or axion-like particles coupling to electrons, photons, or nucleons, do not apply here. Since there is no light degree of freedom that can be efficiently produced in stellar interiors to carry away energy, stars cool in the standard manner (photon emission, neutrinos, etc.) without additional channels~\cite{Raffelt:2006cw}. The well-known bounds from supernova SN1987A on axion like particle emission (which require that any new light boson not carry off too much energy from the supernova core) are likewise automatically satisfied. In essence, the gauged $Z_N$ does not introduce any new light propagating states, so it evades all astrophysical limits that are based on novel light particles.

Cosmologically, the absence of an axion means there are no axion oscillations or isocurvature fluctuations in the early universe. Axion models are tightly constrained by cosmic microwave background (CMB) observations because a misaligned axion field would produce cold dark matter isocurvature perturbations unless the axion decay constant $f_a$ is relatively low or inflation occurs at low scales~\cite{Wantz:2009it}. In D$\theta$P, there is no misalignment angle to worry about, and no isocurvature perturbation is generated (the $Z_N$ is a gauge symmetry, not a spontaneous global symmetry). This relieves the model of one of the primary cosmological headaches of axion models. Furthermore, there is no axion domain-wall problem, in PQ axion models with discrete symmetry $Z_N$ in the axion potential, if the symmetry breaks after inflation, one generically gets a network of $N$ domain walls attached to axion strings, which can dominate the energy density (the ``domain wall problem''~\cite{Sikivie:1982qv}). In D$\theta$P, by contrast, the $\theta$-parameter is not a field undergoing spontaneous symmetry breaking, so the universe is never populated by different $\theta$ vacua separated by axion domain walls. Instead, as we discuss next, the discrete symmetry structure yields a different kind of relic, dynamical membrane-like objects, which, however, can be rendered cosmologically harmless.


\subsection{Topological relics and vacuum structure.}

Gauging the discrete $\theta$-shift has an interesting effect on QCD’s vacuum structure. 
In ordinary QCD (with a fixed $\theta$), if $\theta$ were dynamically variable (as in an axion model), one would have a multivalued potential with minima periodically spaced in $\theta$~\cite{Marsh:2015xka,DiLuzio:2020wdo}. 
In our case, the theory sums over $\theta$ images shifted by $2\pi/N$, which effectively means that the $\theta$ variable is periodic on a smaller circle of size $2\pi/N$. 
The vacuum structure becomes an envelope of $N$ branches, but only the branch with the smallest $|\bar{\theta}|$ is realized as the true vacuum. 
One can, however, imagine field configurations or excited states that correspond to being on a “neighboring branch”. 
For example, there can be an excited vacuum with effective angle $\bar{\theta}_{\rm eff} = 2\pi/N$ (instead of $0$). 
The energy difference between this first excited branch and the true vacuum is of order 
\begin{equation}
\Delta E \sim \tfrac{1}{2}\chi_t (2\pi/N)^2
\end{equation}
for small $2\pi/N$. 

In the four dimensional path integral, a configuration that interpolates between $\bar{\theta}=0$ and $\bar{\theta}=2\pi/N$ across space is realized by an object analogous to a domain wall. 
It is a sheet or membrane across which the discrete topological charge (Pontryagin index) jumps by one unit. 
Indeed, in a Hamiltonian picture, one can think of “large” instantons or fractional topological flux spreading over a two dimensional surface that separates two regions. 
On one side of the membrane the vacuum is in branch $m$, and on the other side it is in branch $m+1$. 
Because $\theta$-shifts by $2\pi/N$ are gauged, such a membrane is a permissible dynamical object. 
It carries a quantized topological flux and can end on a $Z_N$ cosmic string (a vortex in the topological sector)~\cite{Vilenkin:2000jqa,Kibble:1976sj}. 
These membranes are not topological domain walls arising from spontaneous symmetry breaking, since the $Z_N$ is not broken and is a true gauge symmetry. 
They are therefore not stabilized by standard topological arguments. 
They are more akin to solitonic branes connecting different vacua of a multivalued potential~\cite{Krauss:1988zc}. 

The tension $\sigma$ of these membranes is a parameter that depends on the ultraviolet completion of the topological sector. 
In an effective description, one can write an action term for an extended membrane spanning area $A$ as $S_{\rm mem}=\sigma A$, so $\sigma$ has units of energy per unit area. 
Because the membranes must carry one unit of QCD topological charge across them, one might naively expect $\sigma$ to scale with the QCD scale. 
A simple estimate gives $\sigma\sim N \Lambda_{\text{QCD}}^3$ on dimensional grounds (the factor $N$ could appear, for instance, if one considers that $N$ such membranes stacked would accomplish a $2\pi$ shift). 
However, $\sigma$ is not computed by QCD alone, it also involves the dynamics of the new topological $Z_N$ sector. 
In “discrete clockwork” ultraviolet constructions, one can achieve an exponentially large effective $N$ with only moderately large microscopic $N_i$, at the cost of introducing multiple coupled $Z_{N_i}$ sectors~\cite{Kaplan:2015fuy,Giudice:2016yja}. 
In such cases, the membrane tension in the effective theory could be much larger than the QCD scale (depending on heavy parameters and couplings in the ultraviolet), or vice versa, so we treat $\sigma$ as essentially a free parameter subject to some broad constraints. 

From a cosmological perspective, we must ensure that these membranes do not lead to any unacceptable relic abundances or domain wall like problems~\cite{Vilenkin:2000jqa,Kawasaki:2013iha,Hiramatsu:2012gg}. 
Several factors ensure that they do not. 
First, as mentioned, there is no stage in the early universe where a continuous symmetry is broken and different regions of space fall into different $\theta$-vacua. 
Instead, at high temperature $T\gg \Lambda_{\text{QCD}}$, the effective QCD $\theta$-potential is nearly flat (since the topological susceptibility $\chi_t(T)$ is heavily suppressed)~\cite{Borsanyi:2015cka,Petreczky:2016vrs}. 
This means $\bar{\theta}$ is undetermined (and unimportant) in the deconfined phase. 
As the universe cools through $T\sim \Lambda_{\text{QCD}}$, $\chi_t$ grows and the $\theta$-vacuum structure re-emerges. 
At that point, however, the $Z_N$ gauge dynamics ensure that the universe chooses the $\bar{\theta}$ branch closest to $0$ as the unique vacuum. 
There is no domain formation. 
The entire universe (within one horizon volume after the QCD epoch) will settle into the same $\bar{\theta}\approx0$ branch because that is the absolute minimum of energy and there is no degeneracy, unlike an axion $U(1)$ where degenerate vacua exist when the axion mass turns on~\cite{Sikivie:1982qv,Marsh:2015xka}. 
Therefore, one does not expect a network of membranes crisscrossing the universe. 
The true vacuum pervades space from the start of confinement. 
In fact, the situation is even safer than in Peccei–Quinn axion models that undergo the symmetry breaking after inflation. 
Those models produce axion strings and domain walls that need careful consideration~\cite{Sikivie:1982qv,Hiramatsu:2012gg,Kawasaki:2013iha}. 
In contrast, D$\theta$P produces neither of these in the early universe. 

What about thermal or quantum production of membrane defects. 
One could imagine that after confinement there is some probability to nucleate a small bubble of a higher $\bar{\theta}$ vacuum (for example $\bar{\theta}=2\pi/N$) within our $\bar{\theta}\approx0$ vacuum. 
Such a bubble would be bounded by a membrane of tension $\sigma$, and the bubble interior would have a higher vacuum energy by $\Delta E$. 
If $\Delta E>0$, which it is since $\bar{\theta}=0$ is the true minimum, then a bubble of the false vacuum has a pressure deficit and will tend to collapse. 
However, it could be produced by quantum tunneling or thermal excitation if it is small enough. 

For zero temperature quantum tunneling in four dimensions, the nucleation rate per unit volume for a bubble of critical size $R_c$ is roughly 
$\Gamma/V \sim A\, e^{-B}$, with $B$ the Euclidean action of the 
$O(4)$-symmetric bounce solution~\cite{Coleman:1977py,Callan:1984sa}. 
In the thin wall approximation one finds
\begin{equation}
B_{4\mathrm{D}} \simeq \frac{27\pi^2 \sigma^4}{2\,(\Delta E)^3}\,,
\label{eq:bounce4D}
\end{equation}
where $\sigma$ is the membrane tension (surface energy per unit area) and 
$\Delta E$ is the energy density difference between the two vacua~\cite{Coleman:1977py}. 
Parametrically, $B \sim \sigma^4/\epsilon^3$ where 
$\epsilon \equiv \Delta E$. 
In our case, for large $N$,
\begin{equation}
\epsilon \approx \frac{1}{2}\,\chi_t \left(\frac{2\pi}{N}\right)^2
\;\sim\; \frac{(2\pi)^2}{2N^2}\,\chi_t \,.
\end{equation}
Using $\sigma \sim N\,\Lambda_{\rm QCD}^3$ and 
$\chi_t \sim \Lambda_{\rm QCD}^4$, we obtain
\begin{equation}
B \sim N^{10}
\end{equation}
up to numerical $2\pi$ factors. 
Even for moderate $N$ this exponent is gigantic. 
Taking order one prefactors, $N \sim 10^2$ already gives 
$B \sim 10^{20}$, while $N \sim 10^{10}$ (our preferred regime) would
correspond to $B \sim 10^{100}$. 
The associated tunneling rate is thus 
utterly negligible, and any membrane facilitated transitions between
$\theta$-branches are enormously suppressed. 
The $\bar\theta = 0$ vacuum
is essentially stable for all practical purposes, and higher $\bar\theta$
regions will not materialize via tunneling. 

One could also consider if thermal fluctuations at the QCD phase transition could excite some regions into an excited branch. 
However, at $T\sim T_c$ the free energy cost $\Delta F(T)$ for $\bar{\theta}\neq0$ is already appreciable, and it grows from $0$ to its $T=0$ value as confinement sets in. 
Lattice results show that $\chi_t(T)$ rises sharply below $T_c$~\cite{Petreczky:2016vrs,Borsanyi:2015cka}. 
Therefore any region with a nonzero $\bar{\theta}$ finds itself at a free energy disadvantage and will relax to $\bar{\theta}\approx0$. 
In a sense, the universe undergoes a form of “vacuum selection” at the QCD crossover. 
The initially flat $\theta$ potential at high $T$ tilts toward $\theta=0$ at low $T$, and the system rolls into the nearest minimum (which, due to the gauged $Z_N$, is unique and at $\theta\approx0$). 
This is in line with the statement that there is no spontaneous breaking of $Z_N$. 
The universe need not choose between degenerate vacua, there is only one vacuum once the coupling is turned on. 

In summary, the D$\theta$P mechanism avoids the classic axion domain wall problem and does not lead to a dangerous abundance of relic membranes or other topological defects. 
The membranes that do exist in principle are either not formed at all on cosmic scales, or if formed, they are either bounded by strings or collapse quickly~\cite{Vilenkin:2000jqa,Kibble:1976sj}. 
If $\sigma$ is very large, any would-be network of membranes would have extremely high tension and would quickly annihilate, or would be inflated away if inflation occurred at a later stage~\cite{Linde:1981zj}. 
If inflation happens before the QCD scale, then even the above considerations are moot, as $\bar{\theta}$ will be homogenized to the principal branch in our observable patch from the start. 
The key point is that because the $Z_N$ is a gauge symmetry and not a global symmetry, there is no enduring degeneracy of vacua. 
The universe has a unique ground state (up to irrelevant gauge copies). 
Thus, D$\theta$P is cosmologically safe, it solves the strong CP problem without introducing new cosmological problems, and in fact it ameliorates some problems (no isocurvature, no axion strings and no axion domain walls) that afflict the axion approach~\cite{Sikivie:2006ni,Marsh:2015xka}. 

Here, we note that the above phenomenological and cosmological features distinguish the D$\theta$P scenario from typical axion solutions. 
The neutron EDM is ultra suppressed and stays well below experimental bounds, with the suppression protected by a discrete gauge symmetry that is robust against radiative and gravity induced corrections. 
No light axion is present, so all direct and indirect searches for axions will return null results specific to the strong CP sector. 
The mechanism leaves weak CP violation untouched, and any tiny CP effects from the Standard Model (like CKM loop induced EDMs) remain, although they are far too small to detect. 
The theory predicts new membrane like objects instead of global axion domain walls, but these objects can be made harmless by appropriate parameter choices, for example large tension and thus exponentially suppressed nucleation~\cite{Coleman:1977py}. 
In essence, D$\theta$P provides a “cassette” of new topological physics that solves the strong CP puzzle in a ultraviolet protected way, while mimicking the success of axion models (nearly zero $\bar{\theta}$) without the need for an axion itself. 
Upcoming EDM experiments, which will probe down to the $10^{-27}$ to $10^{-28}$~$e$·cm level, are unlikely to find a neutron EDM signal if the D$\theta$P mechanism is at work (unless $N$ is unnaturally small)~\cite{Chupp:2017rkp}. 
This is a clear, testable distinction: any positive detection of a neutron EDM near the current sensitivity would falsify the D$\theta$P scenario or force $N$ to be only modest (which in turn could conflict with the theoretical requirement of anomaly freedom, see Section~\ref{sec5}). 
On the other hand, the continued absence of an EDM signal, while consistent with D$\theta$P, is not a smoking gun by itself. 
Many solutions to the strong CP problem predict the same. 
Rather, the unique predictions of D$\theta$P lie in its cosmological and topological implications, which we explore further in the next section, and in possible lattice signatures (Section~\ref{sec:diagnostics}) that could reveal the underlying $N$-branch structure of the QCD vacuum.

\section{Cosmology}
\label{sec:cosmo}

The expansion rate in a radiation-dominated universe is governed by the Friedmann equation, 
\begin{equation}
H^2 = \frac{8\pi\,\rho(T)}{3\,M_{\rm Pl}^2}\,,
\end{equation}
where $M_{\rm Pl}$ is the Planck mass and $\rho(T)$ is the energy density at temperature $T$~\cite{Kolb:1990vq,Dodelson:2003ft}. 
For a relativistic plasma with $\rho(T)=\frac{\pi^2}{30}\,g_*(T)\,T^4$ (effective degrees of freedom $g_*$), this gives 
\begin{equation}
H(T) = \sqrt{\frac{8\pi}{3M_{\rm Pl}^2}\,\frac{\pi^2}{30}\,g_*(T)\,T^4} 
= \frac{\pi}{\sqrt{90}}\;\sqrt{g_*(T)}\;\frac{T^2}{M_{\rm Pl}}\,.
\end{equation} 
Using $H\simeq \frac{1}{2t}$ during radiation domination, one finds the time-temperature relation 
\begin{equation}
t(T) \;\approx\; \frac{1}{2H(T)} \;=\; \frac{1}{2}\sqrt{\frac{3\,M_{\rm Pl}^2}{8\pi\,\rho(T)}} 
\;\approx\; 0.3\;\frac{M_{\rm Pl}}{\sqrt{g_*(T)}\,T^2}\,,
\end{equation}
so, for example, $T\sim1~{\rm MeV}$ corresponds to $t\sim1~\text{s}$~\cite{Kolb:1990vq,ParticleDataGroup:2022pth}. 

At high temperatures $T\gg\Lambda_{\rm QCD}$ (deconfined phase), QCD instanton fluctuations are rare and the topological susceptibility $\chi_t(T)$ is strongly suppressed. 
Lattice results indicate that $\chi_t$ falls as a power law at $T\gg T_c$ (with $T_c$ the QCD crossover temperature), 
\begin{equation}
\chi_t(T)\;\sim\;\chi_t(0)\,\Big(\frac{T}{T_c}\Big)^{-b}\,,
\end{equation} 
with an exponent $b$ of order $8$ (consistent with dilute-instanton gas and lattice simulations)~\cite{Borsanyi:2015cka,Petreczky:2016vrs,Borsanyi2016ksw}. 
As the universe cools through $T_c\sim\Lambda_{\rm QCD}$, $\chi_t(T)$ grows rapidly, approaching its $T=0$ value $\chi_t(0)\approx(75~\text{MeV})^4$~\cite{Vicari:2008jw}. 

Gauging the $Z_N$ subgroup of the $2\pi$ $\theta$-shift symmetry enforces a selection rule on topological charges: 
\begin{equation}
Q+\kappa_G Q_G = N\,K\,,
\end{equation} 
for some $K\in\mathbb{Z}$.  
In other words, only those field configurations survive in the path integral whose combined QCD topological charge $Q$ plus $\kappa_G$ times the charge $Q_G$ of the new $Z_N$ topological sector is an integer multiple of $N$~\cite{Cordova:2019uob,Hsin:2020nts}. 
This condition implies that a shift $\bar{\theta}\to\bar{\theta}+2\pi/N$ (accompanied by an appropriate discrete gauge transformation on $Q_G$) is a gauge symmetry of the theory, identifying $\bar{\theta}$ with $\bar{\theta}+2\pi/N$. 
Equivalently, the fundamental domain of the physical CP angle is reduced from $[0,2\pi)$ to $[0,2\pi/N)$. 
The Euclidean path integral can then be written as an orbifold sum over $\theta$-images: 
\begin{equation}
Z_{\rm gauged}(\bar{\theta}) \;=\; \frac{1}{N}\sum_{m=0}^{N-1} Z_{\rm QCD}\!\big(\bar{\theta}+2\pi m/N\big)\,. 
\end{equation} 
In a large spacetime volume (thermodynamic limit), this sum is dominated by the minimal free-energy contribution. 
Denoting $E_{\rm QCD}(\theta,T)=-\frac{T}{V}\ln Z_{\rm QCD}(\theta)$ as the vacuum energy density of QCD at angle $\theta$, the gauged theory’s vacuum energy is 
\begin{equation}
E(\bar{\theta},T)\;=\;\min_{m\in\mathbb{Z}}\;E_{\rm QCD}\!\Big(\bar{\theta}+\frac{2\pi m}{N}\,,\,T\Big)\,. 
\end{equation} 
This shows that the true vacuum will lie on the branch $m$ that brings $\bar{\theta}+2\pi m/N$ into the principal interval $[-\pi/N,\,+\pi/N]$. 
In particular, the effective CP angle in the vacuum is $\bar{\theta}_{\rm eff}\equiv \bar{\theta}+2\pi m_{\rm min}/N$ with $|\bar{\theta}_{\rm eff}|\le\pi/N$. 
Expanding the QCD vacuum energy for small $\bar{\theta}_{\rm eff}$, 
\begin{equation}
E_{\rm QCD}(\bar{\theta}_{\rm eff},T)\;\approx\;E_{\rm QCD}(0,T)\;+\;\frac{1}{2}\,\chi_t(T)\,\bar{\theta}_{\rm eff}^{\,2}\;+\;\mathcal{O}(\bar{\theta}_{\rm eff}^4)\,,
\end{equation} 
we see that the energy density rises quadratically away from the CP-conserving point~\cite{Witten:1998uka,Vicari:2008jw}. 
Therefore, the difference in vacuum energy density between the lowest branch ($\bar{\theta}_{\rm eff}\approx0$) and the first excited branch ($\bar{\theta}_{\rm eff}=2\pi/N$) at temperature $T$ is 
\begin{equation}
\Delta E(T)\;\equiv\;E(\bar{\theta}_{\rm eff}=2\pi/N,\,T)\,-\,E(0,T)\;\approx\;\frac{1}{2}\,\chi_t(T)\,\Big(\frac{2\pi}{N}\Big)^{2}\,. 
\end{equation} 
For $T\gg T_c$, $\chi_t(T)$ is tiny and $\Delta E(T)$ is essentially zero, all $\bar{\theta}$-branches are then quasi-degenerate in energy. 
But as $T$ drops below $\sim\Lambda_{\rm QCD}$, $\chi_t$ grows and $\Delta E(T)$ becomes significant, lifting the would-be $N$-fold vacuum degeneracy. 
In the $T\to0$ limit, $\Delta E$ approaches its maximum value $\frac{1}{2}\chi_t(0)\,(2\pi/N)^2$. 
The $Z_N$ being a gauge symmetry means there is no spontaneous breaking: only a single true vacuum remains (the branch with the smallest $|\bar{\theta}_{\rm eff}|$), and all other branches correspond to excited states. 
In particular, there are no stable domain walls. The would-be domain walls between different $\bar{\theta}$-vacua are rendered non-topological and can be eliminated by a gauge transformation of the new sector~\cite{Krauss:1988zc,Banks:2010zn,Harlow:2018tng}. 

One can nevertheless consider a transient configuration where a spatial region is in the $m=1$ excited branch (effective angle $\bar{\theta}_{\rm eff}=2\pi/N$) while the outside universe is in the $m=0$ vacuum ($\bar{\theta}_{\rm eff}\approx0$). 
The interface between these regions is a membrane (domain wall) across which the discrete topological charge jumps by one unit. 
Such a membrane carries a tension $\sigma(T)$ (energy per unit area) which, in an effective description, can be written as an action term $S_{\rm mem}=\sigma\,A$ for a membrane of area $A$. 
The value of $\sigma$ depends on ultraviolet details of the topological sector~\cite{Dvali:2005an,Freedman:1980us}. 
Dimensional analysis suggests a characteristic scale of order the QCD confinement scale, for example $\sigma\sim N\,\Lambda_{\rm QCD}^3$ in a minimal model, although extended constructions (such as clockwork chains of multiple $Z_{N_i}$) can raise or lower the effective tension~\cite{Kaplan:2015fuy,Giudice:2016yja}. 
At very high $T$, when $\Delta E(T)\approx0$, the membrane tension is irrelevant (one has $\sigma(T\gg T_c)\approx0$ because there is no distinction between vacua). 
Once confinement sets in and $\Delta E(T)$ grows, the membrane develops a nonzero tension, and more importantly, a pressure difference $\Delta E(T)$ arises across it. 
The higher-energy phase (the $m=1$ branch) feels a net pressure $\Delta E$ pushing the wall toward the lower-energy phase (the $m=0$ vacuum). 
Consequently, any $m\neq0$ region is metastable and will tend to shrink. 

The decay of a false-vacuum region can proceed via thermal or quantum nucleation of a bubble of true vacuum. 
The decay rate per unit volume can be estimated as $\Gamma/V \sim A\,\exp[-B(T)]$, where $B(T)$ is the minimal Euclidean action (the bounce) of the critical bubble~\cite{Coleman:1977py,Callan:1977pt,Coleman:1980aw}. 
In the thin-wall limit (valid when $\Delta E \ll \sigma\,(\Lambda_{\rm QCD})^3$), the bounce action is approximately 
\begin{equation}
B(T)\;\approx\;\frac{16\pi}{3}\,\frac{\sigma(T)^3}{[\Delta E(T)]^2}\,.
\end{equation} 
As $T$ decreases and $\Delta E(T)$ grows, the exponent $B(T)$ drops precipitously, meaning that any remaining false-vacuum regions can rapidly disappear. 
Indeed, the vacuum selection effectively takes place as soon as the QCD topological potential becomes appreciable. 
In an expanding universe, when the temperature cools to $T\sim T_c$, even a tiny bias $\Delta E(T)$ is enough to drive the entire Hubble volume into the $m=0$ lowest-energy vacuum~\cite{Kolb:1990vq}. 
There is no epoch in which different Hubble regions are permanently trapped in different $\bar{\theta}$-vacua, so no domain wall network can form~\cite{Vilenkin:2000jqa}. 
Any membrane-like defects that do momentarily appear (for instance, as transient field configurations carrying one unit of topological charge) are unstable and will collapse or annihilate shortly after formation. 
Crucially, the $\bar{\theta}$-angle is “locked in” to the principal branch long before primordial nucleosynthesis. 
For the enormous values of $N$ required by the strong CP bound (for example $N\gtrsim10^{10}$ to achieve $|\bar{\theta}|<10^{-10}$), the vacuum selection completes at a temperature far above $T\sim1~\text{MeV}$, ensuring that no dangerous relic walls survive to the era of BBN~\cite{ParticleDataGroup:2022pth}.

\section{Comparison with Peccei-Quinn Solutions}

The Peccei-Quinn (PQ) mechanism~\cite{Peccei:1977hh} introduces a new global $U(1)_{PQ}$ symmetry spontaneously broken at a high scale $f_a$, yielding a pseudo-Goldstone boson (axion) that couples to $G\tilde{G}$ and dynamically relaxes the effective $\theta$ angle to zero~\cite{Weinberg:1977ma,Wilczek:1977pj}. The axion field $a(x)$ can be introduced as $\theta \to \bar{\theta} + a(x)/f_a$, so that the QCD vacuum energy as a function of the axion expectation takes the form 
\begin{equation}
\label{PQpot}
V_{PQ}(a) \;=\; \chi_t\,\Big[1 - \cos\!\Big(\frac{a}{f_a} - \bar{\theta}\Big)\Big]~,
\end{equation}
where $\chi_t$ is the topological susceptibility (approximately $\chi_t \sim \Lambda_{\text{QCD}}^4$) setting the scale of the CP-violating vacuum energy. The potential in Eq.~\eqref{PQpot} is $2\pi f_a$-periodic in $a$ and is minimized when $a/f_a$ cancels the CP-violating angle: $\langle a/f_a \rangle + \bar{\theta} = 2\pi k$ for some integer $k$. In particular, for $\bar{\theta}$ in $[-\pi,\pi]$ one finds a unique minimum at $a/f_a = -\,\bar{\theta}$, so that the physical $\bar{\theta}_{\rm PQ}$ is driven to zero in the vacuum. Expanding $V_{PQ}(a)$ about the minimum yields the axion mass $m_a$ in terms of $\chi_t$ and $f_a$
\begin{equation}
\label{axionmass}
m_a^2 f_a^2 \;=\; \frac{\partial^2 V_{PQ}}{\partial a^2}\Big|_{a=\langle a\rangle} \;=\; \chi_t~, \qquad 
m_a \;\approx\; \frac{\sqrt{\chi_t}}{\,f_a\,}~,
\end{equation}
showing that $m_a$ is parametrically small for large $f_a$. For a realistic QCD axion with $f_a \gtrsim 10^{11}$~GeV, one finds $m_a \lesssim 10^{-4}$~eV, far below hadronic scales. The structure of $V_{PQ}(a)$ reflects the color anomaly of $U(1)_{PQ}$, often quantified by the domain-wall number $N_{DW}$. If $N_{DW}>1$, the cosine potential has $N_{DW}$ degenerate vacua (corresponding to $a/f_a = 2\pi k/N_{DW}$ for $k=0,\dots,N_{DW}-1$) that are all physically equivalent up to a discrete shift $a\to a+2\pi f_a$~\cite{Wilczek:1977pj}. In minimal models one has $N_{DW}=1$, so that the axion has a single physical vacuum and no domain-wall proliferation. In all cases, however, the PQ mechanism ensures that the effective $\bar{\theta}$ is relaxed to $\sim 0$
\begin{equation}
\bar{\theta}_{PQ} \;\sim\; 0 \qquad \text{(vacuum relaxation)}~,
\end{equation}
solving the strong CP problem to all orders in QCD.

By contrast, the discrete-$\theta$ projection (D$\theta$P) mechanism~\cite{Witten:1998uka,Cordova:2019uob} solves the strong CP problem without a dynamical axion field, instead employing a gauged discrete symmetry to ``lock’’ the $\theta$ angle to a small value. In this framework, a finite subgroup $\mathbb{Z}_N$ of the shift symmetry $\theta \to \theta + 2\pi$ is promoted to a gauge symmetry, typically by coupling QCD to a topological 3-form sector that enforces $\theta$-periodicity mod $2\pi/N$. Physically, this means that $\theta$ is only defined modulo $2\pi/N$. The theory identifies $\theta$ and $\theta + 2\pi/N$ as the same vacuum angle. As a result, the vacuum will settle in the $\theta$-sector closest to zero. In particular, if $E_{QCD}(\theta)$ denotes the vacuum energy (density) of QCD as a function of the bare $\theta$ angle, then in the $\mathbb{Z}_N$-gauged theory the physical vacuum energy is obtained by an ``orbifold projection’’ over the $N$ images of $\theta$~\cite{Witten:1998uka} 
\begin{equation}\label{Egauged}
E_{D\theta P}(\theta) \;=\; \min_{m\,\in\,\mathbb{Z}}\,E_{QCD}\!\Big(\theta + \frac{2\pi m}{N}\Big)~,
\end{equation}
choosing the branch $m$ that minimizes the energy. Equivalently, one may restrict $\theta$ to lie in the fundamental domain $-\pi/N \le \theta \le \pi/N$, since any larger $|\theta|$ can be reduced by a gauge transformation $m$. The vacuum is therefore selected at an effective angle $\bar{\theta}_{\rm eff}$ satisfying $|\bar{\theta}_{\rm eff}| \le \pi/N$. In particular, starting from an arbitrary $\bar{\theta}$ (including the contributions from quark phases), the theory is driven to a vacuum with 
\begin{equation}\label{thetaBound}
|\bar{\theta}_{D\theta P}| \;\le\; \frac{\pi}{N}~,
\end{equation}
as an exact consequence of the $\mathbb{Z}_N$ gauge symmetry. This non-perturbative bound implements a small $\bar{\theta}$ without any fine-tuning of parameters: even a $\mathcal{O}(1)$ bare $\theta$ is ``cleaned’’ to $\mathcal{O}(\pi/N)$ in the true vacuum. For large $N$, the residual angle $|\bar{\theta}_{D\theta P}|$ can be made arbitrarily small; for example, gauging a $\mathbb{Z}_{N}$ with $N\sim 10^{11}$ already guarantees $|\bar{\theta}|\lesssim 10^{-10}$, safely within current experimental limits on strong CP violation. We emphasize that the $\bar{\theta}$ suppression in D$\theta$P arises from a gauge symmetry rather than a global one, so it is radiatively and gravitationally stable. No operator that violates $\theta\to\theta+\frac{2\pi}{N}$ invariance can be generated in the low-energy theory, even by Planck-scale physics, since $\mathbb{Z}_N$ is an exact gauge symmetry (in particular, the $\theta$ term itself is unobservable modulo $2\pi/N$). The only conceivable explicit $\theta$-dependent operators are those proportional to $e^{iN\theta}$ (or cosines of $N\theta$), which are permitted by the $\mathbb{Z}_N$ selection rule but either vanish in the vacuum (since $\sin N\bar{\theta}_{\rm eff}=0$) or induce only extremely suppressed effects for large $N$. In short, D$\theta$P enforces $\bar{\theta}_{D\theta P}\sim \pi/N$ via a discrete gauge projection:
\begin{equation}
\bar{\theta}_{D\theta P} \;\sim\; \frac{\pi}{N} \qquad \text{(gauge projection)}~,
\end{equation}
providing an exponentially small $\theta$ angle without any light axion field. Figure~\ref{fig:potentials} illustrates the essential difference between the PQ and D$\theta$P approaches in terms of the vacuum potential. 

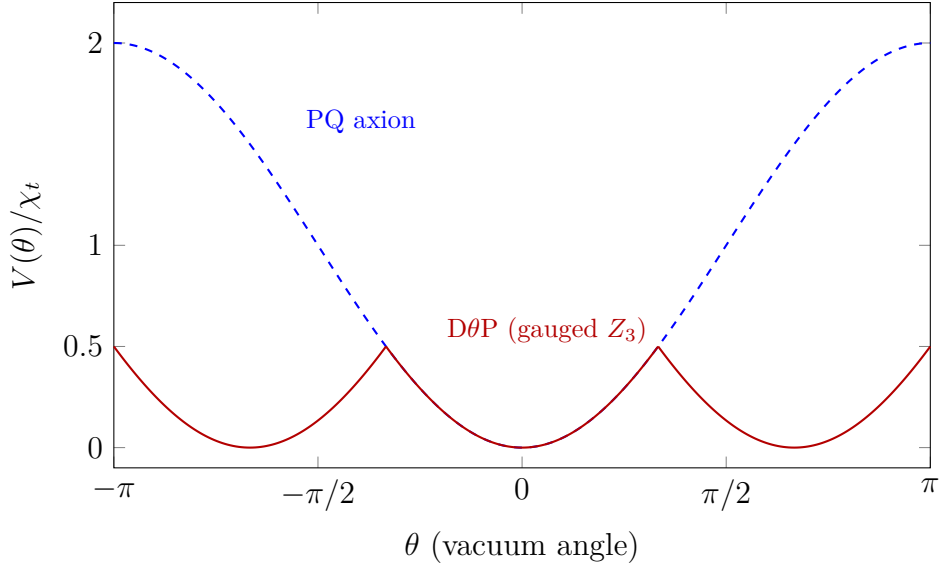
\begin{figure}[t]
    \centering
    \begin{tikzpicture}[scale=1]
      \begin{axis}[
        width=0.8\linewidth,
        height=0.5\linewidth,
        xlabel={$\theta$ (vacuum angle)},
        ylabel={$V(\theta)/\chi_t$},
        xmin=-180, xmax=180,
        ymin=-0.1, ymax=2.2,
        xtick={-180,-90,0,90,180},
        xticklabels={$-\pi$,$-\pi/2$,0,$\pi/2$,$\pi$},
        ytick={0,0.5,1,2},
        legend style={at={(0.98,0.85)},anchor=east, draw=none, font=\footnotesize}
      ]
        \addplot[domain=-180:180, samples=201, smooth, thick, dashed, color=blue] {1 - cos(x)}; 
        \addplot[domain=-180:-60, samples=50, smooth, thick, color=red!70!black] {1 - cos(x + 120)}; 
        \addplot[domain=-60:60,  samples=100, smooth, thick, color=red!70!black] {1 - cos(x)};
        \addplot[domain=60:180,  samples=50, smooth, thick, color=red!70!black] {1 - cos(x - 120)};
        \node[anchor=south west, color=blue] at (axis cs:-100,1.5) {\footnotesize PQ axion};
        \node[anchor=north east, color=red!70!black] at (axis cs:60,0.7) {\footnotesize D$\theta$P (gauged $Z_3$)};
      \end{axis}
    \end{tikzpicture}
    \caption{Comparison of the PQ axion potential and the discrete-$\theta$ projected potential. \emph{Blue dashed curve:} $V_{PQ}(\theta)/\chi_t = 1-\cos\theta$, the standard cosine potential (for illustration we set $\bar{\theta}=0$ and $N_{DW}=1$ so that $\theta=a/f_a$). The axion has a single minimum at $\theta=0$ and oscillates with full $2\pi$ periodicity, implying a mass $m_a^2=\chi_t/f_a^2$ and multi-vacuum structure only if $N_{DW}>1$. \emph{Red solid curve:} $V_{D\theta P}(\theta)/\chi_t$ for a gauged $Z_3$ (discrete $\theta$-periodicity $2\pi/3$). The physical vacuum energy is the lower envelope of the $1-\cos\theta$ branches shifted by $\pm 2\pi/3$, resulting in a unique minimum at $\theta=0$ and a maximal $|\theta|=\pi/3$ in any vacuum. In the $Z_N$ case, the effective $\bar{\theta}$ is confined to $[-\pi/N,\pi/N]$, and would-be $\theta$ vacua are identified by the gauge symmetry (no degenerate minima).}
    \label{fig:potentials}
\end{figure}

Both mechanisms solve the strong CP problem but in distinctly different ways, leading to several concrete differences in phenomenology and model-building. First, the D$\theta$P scenario does not require an approximate global symmetry and hence sidesteps the axion quality problem entirely. The $\bar{\theta}$ suppression arises from exact discrete gauge invariance, so there is no need to engineer a nearly perfect $U(1)_{PQ}$ symmetry that could be spoiled by Planck-suppressed operators~\cite{Witten:1998uka,Cordova:2019uob}. In PQ axion models, by contrast, the $U(1)_{PQ}$ is global and must be extraordinarily precise: even a higher-dimensional operator suppressed by $M_{\rm Pl}$ can re-induce a nonzero $\bar{\theta}$. For example, a dimension-$d$ Planck-scale term proportional to $a^d/M_{\rm Pl}^{\,d-4}$ shifts the axion potential by a tiny amount $\Delta V \sim \Lambda_{\text{QCD}}^4 \cos(a/f_a + \delta)$, which would generically push the vacuum away from $\bar{\theta}=0$ by an angle of order $(f_a/M_{\rm Pl})^{d-4}$ if not highly suppressed. This is the axion ``quality’’ problem. Obtaining $|\bar{\theta}| < 10^{-10}$ typically demands very large $d$ (or extra symmetries) to sufficiently quash these explicit breakings. No such concern arises in D$\theta$P, since any operator that is not exactly $\mathbb{Z}_N$-invariant is forbidden; the discrete gauge symmetry protects the $\theta$ angle up to arbitrary high scale. 

Second, the field content and low-energy degrees of freedom differ significantly. PQ solutions introduce a new light pseudoscalar, the axion, which must be very light and very weakly coupled. This leads to a host of experimental and observational constraints on the axion: laboratory searches (helioscopes, haloscopes, spin precession experiments, etc.) and astrophysical bounds (stellar cooling, supernova energy loss, cosmic background signals) severely restrict the allowed axion parameter space. In the minimal D$\theta$P solution, no light scalar is present at all; the strong CP angle is alleviated by a topological sector (e.g. a 3-form gauge field or equivalent discrete degrees of freedom), which does not produce a light propagating particle. Consequently, D$\theta$P evades all constraints that apply to a QCD axion. There is no axion-like particle to detect in the laboratory, no anomalous energy loss from stars due to axion emission, and no axion dark matter. The absence of a light axion is a major distinguishing feature: whereas a PQ model predicts a spin-0 boson with coupling $1/f_a$ that could be discovered (or constrained) by various experiments~\cite{Peccei:1977hh,Wilczek:1977pj}, an axionless solution like D$\theta$P predicts null results in all such searches. If upcoming axion dark matter and helioscope experiments continue to see no signal despite covering the classic QCD axion parameter space, it would strengthen the case for an axionless solution such as D$\theta$P.

Third, the two solutions differ in their implications for the neutron electric dipole moment (nEDM). In PQ models, the axion relaxes $\bar{\theta}$ to essentially zero, so the leading contribution to the neutron EDM from the QCD $\theta$-term is absent. The neutron EDM in such scenarios arises only from subleading sources (e.g. the CKM phase contributing via higher-order loops) and is extremely small: $|d_n|_{PQ} \approx 0$ to within present precision. Quantitatively, one expects $|d_n|_{PQ} \sim 10^{-32}$-$10^{-31}\,e\cdot\text{cm}$, many orders of magnitude below the current experimental bound $|d_n|<1.8\times10^{-26}\,e\cdot\text{cm}$. This tiny value is effectively unobservable in ongoing EDM experiments. In contrast, the D$\theta$P mechanism permits a small but nonzero $\bar{\theta}$ in general, bounded by $|\bar{\theta}| \le \pi/N$. If $N$ is chosen just large enough to satisfy the current bound (so that $\bar{\theta}_{D\theta P} \sim 10^{-10}$), the resulting neutron EDM could lie near the $10^{-26}\,e\cdot\text{cm}$ level. More optimistically, if $N$ is even larger, the EDM will be further suppressed, effectively approaching zero for experimental purposes. Importantly, however, D$\theta$P offers a clear scaling prediction: the CP-violating effects descend as $1/N$. If one could probe a range of neutron EDM sensitivities corresponding to smaller and smaller $\bar{\theta}$, any detection (or further non-detection) would discriminate between the scenarios. Figure~\ref{fig:edm} illustrates this behavior: the D$\theta$P neutron EDM $|d_n|$ scales as $\sim |\bar{\theta}_{D\theta P}| \propto 1/N$, whereas the PQ solution predicts essentially $|d_n|=0$ independent of any $N_{DW}$ (aside from tiny Standard Model contributions). In other words, the PQ mechanism completely eliminates $\bar{\theta}$-induced EDMs, while D$\theta$P leaves a residual $\bar{\theta}$ of order $\pi/N$ that in principle could generate a small $d_n$ (though for large $N$ it remains far below current limits). From a practical standpoint, both mechanisms predict a neutron EDM well below present bounds - a “null test’’ of strong CP violation - but the PQ axion makes $\bar{\theta}$ dynamically zero (hence $d_n$ exactly zero at tree-level), whereas D$\theta$P enforces only an upper limit on $\bar{\theta}$ (hence $d_n$ is suppressed but not identically zero). Ongoing and future EDM experiments, which aim to improve sensitivity by one or more orders of magnitude, will continue to probe the $\bar{\theta}$-parameter space. If a nonzero neutron EDM is observed in the $\bar{\theta}$-dominated regime ($\sim10^{-27}$-$10^{-26}\,e\cdot\text{cm}$) and can be attributed to a QCD $\theta$-term, it might suggest a small but non-vanishing $\bar{\theta}$ as could occur in D$\theta$P (or in a PQ model with a slight quality problem). On the other hand, a persistent null result simply pushes $N$ higher (or the PQ quality requirements more stringent), keeping both solutions viable but with different outlooks for discovery in other channels (axion searches vs. topological effects).

\begin{figure}[t]
    \centering
    \begin{tikzpicture}[scale=1]
      \begin{axis}[
        width=0.75\linewidth,
        height=0.5\linewidth,
        xlabel={Symmetry order $N$ (or $N_{DW}$)},
        ylabel={Normalized $|d_n|$},
        xmin=1, xmax=20,
        ymin=0, ymax=1.05,
        xtick={1,5,10,15,20},
        ytick={0,0.25,0.5,0.75,1.0},
        legend style={at={(0.98,0.95)},anchor=north east, draw=none, font=\footnotesize}
      ]
        \addplot[smooth, thick, color=red!70!black] coordinates {
            (1,1.0) (2,0.5) (3,0.333) (4,0.25) (5,0.2) (6,0.167) (8,0.125) (10,0.1) (12,0.0833) (15,0.0667) (20,0.05)
        };
        \addplot[domain=1:20, samples=2, thick, dashed, color=blue] {0.0};
        \node[anchor=south west, color=red!70!black] at (axis cs:6,0.2) {\footnotesize D$\theta$P: $|d_n| \propto 1/N$};
        \node[anchor=south east, color=blue] at (axis cs:19,0.05) {\footnotesize PQ: $|d_n| \approx 0$};
      \end{axis}
    \end{tikzpicture}
    \caption{Qualitative comparison of neutron EDM expectations in PQ vs. D$\theta$P solutions, as a function of the relevant symmetry order. For D$\theta$P (red solid line), the neutron EDM $|d_n|$ arises from the residual $\bar{\theta}$ and is therefore proportional to $|\bar{\theta}_{D\theta P}| \sim \pi/N$. (Here we normalize $|d_n|$ to the value corresponding to $|\bar{\theta}|=\pi$; thus $|d_n|/|d_n(\pi)| \approx |\bar{\theta}|/\pi$ on the horizontal axis.) Increasing the discrete symmetry order $N$ leads to a linear suppression of $|d_n|$. In contrast, a PQ axion solution (blue dashed line) drives $\bar{\theta}\to 0$ essentially exactly, predicting $|d_n|\approx 0$ for all domain-wall numbers $N_{DW}$ (aside from negligible SM contributions). Even if one allows for a small PQ symmetry-breaking (``axion quality’’) effect inducing a tiny $\bar{\theta}$, the PQ scenario would still produce a flat, nearly zero $d_n$ curve on this plot. Both mechanisms predict a neutron EDM far below current bounds, but their scaling behaviors differ markedly, offering a potential diagnostic with sufficient experimental improvement.}
    \label{fig:edm}
\end{figure}
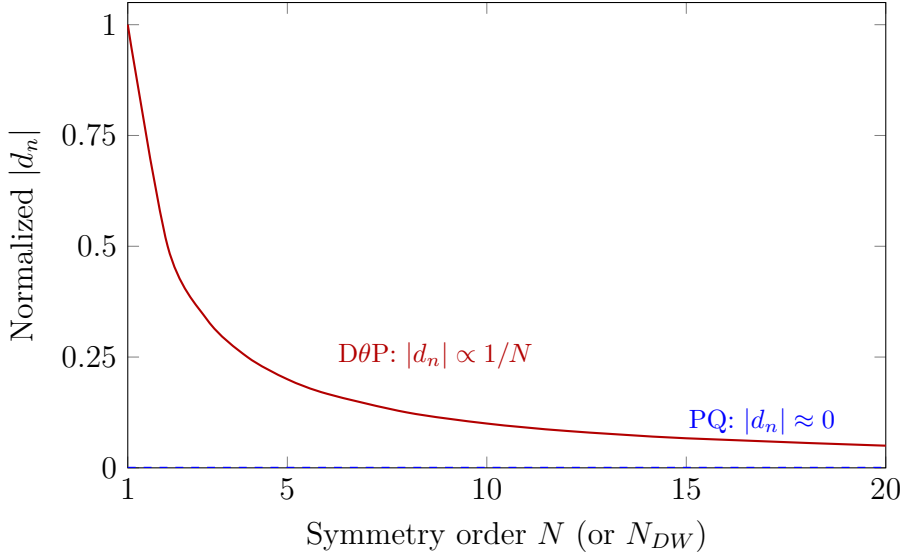

Finally, there are important distinctions in cosmology. PQ axion models are known to face cosmological hurdles associated with the axion field. If $N_{DW}>1$ and the PQ phase transition occurs after inflation, the random choice of $N_{DW}$ vacua across the universe leads to the formation of stable domain walls bounded by axion strings when $m_a$ turns on at the QCD scale. This domain wall network carries immense energy and is ruled out by cosmology (the ``axion domain-wall problem’’~\cite{Peccei:1977hh,Wilczek:1977pj}). Even for $N_{DW}=1$, where stable domain walls are absent, the axion field may produce significant isocurvature perturbations if it is present during inflation. A light axion is practically massless during inflation, so quantum fluctuations imprint a patchy $\bar{\theta}$ (axion misalignment angle) across the universe; these isocurvature fluctuations are tightly constrained by CMB observations. As a result, PQ axion models either require inflation to occur after the PQ breaking (to dilute any topological relics and homogenize the axion field) or must fine-tune the inflationary Hubble scale to suppress axion fluctuations. The D$\theta$P mechanism, on the other hand, avoids these cosmological issues. Because the $\theta$-shift symmetry is gauged and never truly spontaneously broken, there is no epoch in which $N$ distinct $\theta$ vacua exist in different regions of space. The would-be degenerate vacua are identified as a single state by the $\mathbb{Z}_N$ gauge symmetry, and no stable domain walls form between them. In effect, what would have been a domain wall in an ungauged theory is now unobservable or can be continuously deformed away by a gauge transformation in D$\theta$P. (If the discrete symmetry has a residual anomaly, one may get transient membrane-like objects, but these are not stable topological defects; they carry gauge flux and can annihilate or collapse without causing a domain-wall problem.) Thus, an axionless, $\mathbb{Z}_N$-protected universe does not produce cosmic strings or domain walls that threaten cosmology. Likewise, there is no axion field during inflation to generate isocurvature perturbations. The only new ingredients in D$\theta$P are massive topological sectors (e.g. heavy 0-form or 3-form gauge fields and associated membranes) which, if anything, tend to dilute away rather than dominate the universe. The D$\theta$P mechanism avoids the classic axion cosmological problems, no axion domain walls, no axion strings, and no isocurvature fluctuations. It achieves the strong CP solution in a way that is coexistent with standard cosmology, with the added implication that the universe contains no axion dark matter and no relic axion signals. These clear differences from PQ scenarios underscore how Discrete $\theta$ Projection provides a conceptually distinct and phenomenologically safe route to solving the strong CP problem without axions.

\section{Large $N$ and UV Completions}
\label{sec:LargeN-UV}
Having established that the discrete-$\theta$ projection dynamically confines the physical angle to $|\bar\theta_{\rm eff}|\le \pi/N$ and thus strongly motivates parametrically large $N$, we now show how multi-branch vacuum structure, anomaly safe topological selection rules, and discrete clockwork chains can generate exponentially large effective orders $N_{\rm eff}$ in consistent ultraviolet completions~\cite{Kaplan:2015fuy,Gaiotto:2017yup}. 
Eq.~\eqref{eq:ThetaEffBound} motivates taking $N$ extremely large. 
Here we show how a multi-branch structure yields an exponentially large effective $N$ in consistent UV constructions (“discrete clockwork”). 
We begin with the vacuum energy dependence on $\theta$ for finite $N$ and examine the $N\to\infty$ limit. 
We then derive the topological selection rules and clockwork mechanisms that achieve enormous $N_{\rm eff}$ while preserving all anomaly cancellation conditions~\cite{tHooft:1973alw,Witten:1979vv,Witten:1998uka,Bonati:2016tvi}.

\subsection{Multi-Branch Vacuum Energy and Large-\emph{N} Limit}

Gauging a $Z_N$ $\theta$-shift forces an orbifold identification on the $\theta$-circle. 
Path integrals at $\theta$ and $\theta+2\pi m/N$ (for any integer $m$) are identified and summed. 
The partition function is
\begin{equation}
Z_{\rm gauged}(\theta)=\frac{1}{N}\sum_{m=0}^{N-1}Z_{\rm YM}\!\left(\theta+\frac{2\pi m}{N}\right),
\end{equation}
where $Z_{\text{YM}}(\theta)$ is the Yang–Mills partition function at angle $\theta$~\cite{Witten:1998uka,Vicari:2008jw}. 
In the large volume (thermodynamic) limit, this leads to a multi-branch vacuum energy representation. 
Writing $E_{\text{YM}}(\theta) = -\lim_{V\to\infty}\frac{1}{V}\ln Z_{\text{YM}}(\theta)$ for the vacuum energy density of Yang–Mills at $\theta$, one finds
\begin{equation}\label{eq:EnvelopeEnergy}
E(\theta)=\min_{m\in\mathbb Z}E_{\text{YM}}\!\left(\theta+\frac{2\pi m}{N}\right),
\end{equation}
that is, the physical $E(\theta)$ is the lower envelope of $E_{\text{YM}}(\theta)$ over all $\theta$-images shifted by $2\pi/N$~\cite{Bonati:2016tvi}. 
Equivalently, $\theta$ is now periodically identified only modulo $2\pi/N$, so there is an $N$-fold branch structure in $E_{\text{YM}}$, but only the energetically minimal branch contributes to $E(\theta)$. 
Crucially, this envelope is flatter than any single branch of $E_{\text{YM}}$. 
In particular, $E(\theta)$ is $2\pi$-periodic and CP-even, with a unique minimum at $\theta=0$, but its curvature around $\theta=0$ is suppressed by $N^2$ relative to the original Yang–Mills curvature~\cite{Witten:1979vv,Vicari:2008jw}. 

To see this, note that for large $N_c$ (number of colors), one can express the Yang–Mills vacuum energy as
\begin{equation}
E_{\text{YM}}(\theta)=A N_c\Lambda^4 f(\theta/N_c),
\end{equation}
with $f(x)$ a $2\pi$-periodic even function ($f'(0)=0$ and $f''(0)=\chi_t/\Lambda^4$ relates to the topological susceptibility $\chi_t$)~\cite{tHooft:1973alw,Witten:1979vv}. 
The $N$ gauging does not alter this large-$N_c$ form but forces $\theta$ to lie in the interval $[-\pi/N,\pi/N]$, the principal branch. 
Among the set $\{\theta+2\pi m/N: m=0,\dots,N-1\}$ there is a unique representative $\theta_{\rm eff}\in[-\pi/N,\pi/N]$ that minimizes $|\theta+2\pi m/N|$. 
Then
\begin{equation}
E(\theta)=E_{\text{YM}}(\theta_{\rm eff})=A N_c\Lambda^4 f(\kappa N_c \theta_{\rm eff}),
\end{equation}
where $\kappa$ is a model-dependent constant (for pure $SU(N_c)$, $\kappa=1$). 
Since $|\theta_{\rm eff}|\le\pi/N$, at any physical $\theta$ the theory sits on the branch with the smallest $|\theta_{\rm eff}|$. 
Expanding for small angles,
\begin{equation}
E_{\text{YM}}\!\left(\theta+\frac{2\pi m}{N}\right)\approx E_{\text{YM}}(0)+\frac{1}{2}\chi_t\!\left(\theta+\frac{2\pi m}{N}\right)^{2}+\cdots,
\end{equation}
and the envelope \eqref{eq:EnvelopeEnergy} is minimized by the $m$ that makes $|\theta+2\pi m/N|$ as small as possible. 
Therefore the physical vacuum selects $\theta_{\rm eff}=\theta+2\pi m^*(\theta)/N$ with $|\theta_{\rm eff}|\le\pi/N$, implying the vacuum selection rule
\begin{equation}
\label{eq:ThetaEffBound}
|\theta_{\rm eff}|\le\frac{\pi}{N}.
\end{equation}
In particular, the physically realized angle $\bar\theta_{\rm phys}$ coincides with $\theta_{\rm eff}$, proving that
\begin{equation}
|\bar\theta_{\rm phys}|\le\frac{\pi}{N},
\end{equation}
and thus recovering the suppression bound (1.2) as an exact consequence of the orbifolded $\theta$ dynamics. 
As $N$ increases, the bound tightens. 
Effectively, the CP-violating angle is dynamically driven to zero in the large-$N$ limit. 
The multi-branch structure also ensures that $E(\theta)$ becomes extremely flat (see Fig.~\ref{fig:envelope}). 
As $N\to\infty$, one can choose $m$ such that $\theta_{\rm eff}\to 0$ for any fixed $\theta$, so the envelope $E(\theta)\to E_{\text{YM}}(0)$ for all $\theta$~\cite{Bonati:2016tvi}. 
This is an asymptotically smooth dependence where the theory is virtually CP-conserving. 
Equivalently, the topological susceptibility of the $Z_N$ gauged theory, $\chi_N = \partial^2 E/\partial\theta^2|_{\theta=0}$, is suppressed by $1/N^2$ relative to $\chi_t$. 
For large but finite $N$, $E(\theta)$ is continuous and differentiable except at $\theta=(2\ell+1)\pi/N$, where the minimal branch jumps by one unit ($m\to m\pm1$). 
At those points $E(\theta)$ has a cusp with a discontinuous second derivative, but the cusp height $\Delta E$ shrinks as $\Delta E\approx\frac{1}{2}\chi_t(\pi/N)^2=\mathcal O(1/N^2)$. 
Thus even moderate $N$ yields an almost flat $E(\theta)$. 
This multi-branched yet ultra-flat structure ensures the exponential smallness of $\bar\theta$. 
The theory automatically relaxes to the branch with smallest $\theta_{\rm eff}$, that is, the principal branch with $\theta_{\rm eff}\approx0$, while other branches correspond to metastable free-energy densities $E(\theta_{\rm eff}\approx 2\pi k/N)$ separated by thin domain walls~\cite{Bonati:2016tvi,Gaiotto:2017yup}.

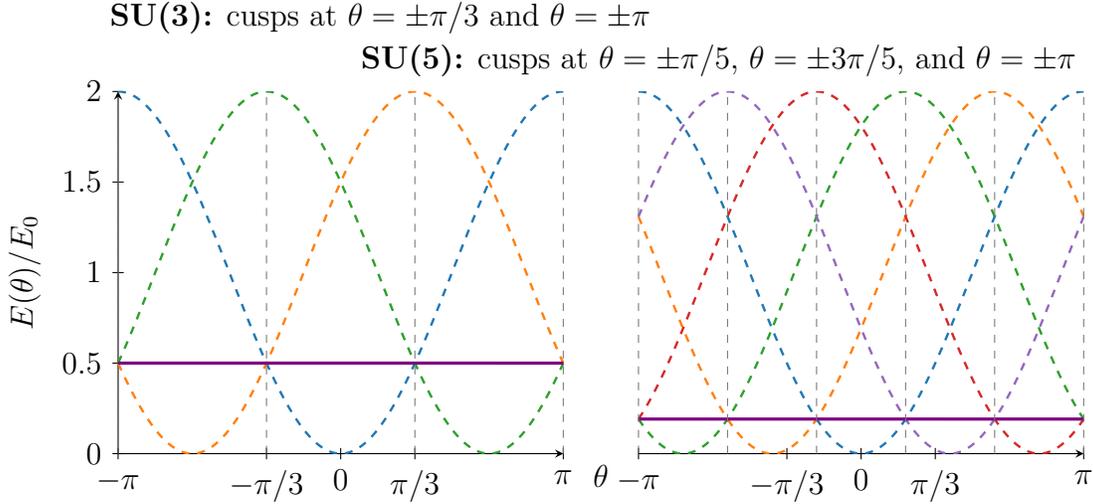
\begin{figure}[htbp]
\centering
\begin{tikzpicture}
  \pgfplotsset{compat=1.18}
  \usepgfplotslibrary{groupplots}
  \definecolor{colEnv}{RGB}{128,0,128}        
  \definecolor{colBrA}{RGB}{31,119,180}       
  \definecolor{colBrB}{RGB}{255,127,14}       
  \definecolor{colBrC}{RGB}{44,160,44}        
  \definecolor{colBrD}{RGB}{214,39,40}        
  \definecolor{colBrE}{RGB}{148,103,189}      
  \definecolor{colGuides}{RGB}{120,120,120}   

  \pgfmathsetmacro{\p}{3.141592653589793}

  \begin{groupplot}[
    group style={group size=2 by 1, horizontal sep=1cm, ylabels at=edge left, yticklabels at=edge left},
    width=7.5cm, height=6.4cm,
    xmin=-\p, xmax=\p, ymin=0, ymax=2,
    xtick={-\p, -\p/3, 0, \p/3, \p},
    xticklabels={$-\pi$, $-\pi/3$, $0$, $\pi/3$, $\pi$},
    ytick={0,0.5,1,1.5,2},
    ylabel={$E(\theta)/E_0$},
    axis x line=bottom, axis y line=left,
    tick style={black},
    every tick/.style={black},
    label style={black},
    title style={black},
    grid=none
  ]

  \nextgroupplot[title={}, xlabel={}, axis y line=left]
    \addplot[colBrA, dashed, line width=0.9pt, domain=-\p:\p, samples=201, smooth] {1 - cos(deg(x))};              
    \addplot[colBrB, dashed, line width=0.9pt, domain=-\p:\p, samples=201, smooth] {1 - cos(deg(x + 2*\p/3))};     
    \addplot[colBrC, dashed, line width=0.9pt, domain=-\p:\p, samples=201, smooth] {1 - cos(deg(x + 4*\p/3))};     

    \addplot[colEnv, very thick, domain=-\p:-\p/3, samples=2] {1 - cos(deg(x + 2*\p/3))};  
    \addplot[colEnv, very thick, domain=-\p/3:\p/3, samples=2] {1 - cos(deg(x))};          
    \addplot[colEnv, very thick, domain=\p/3:\p, samples=2] {1 - cos(deg(x + 4*\p/3))};    

    \draw[colGuides, dashed] (axis cs:-\p/3,0) -- (axis cs:-\p/3,2);
    \draw[colGuides, dashed] (axis cs: \p/3,0) -- (axis cs: \p/3,2);
    \draw[colGuides, dashed] (axis cs:-\p,0)   -- (axis cs:-\p,2);
    \draw[colGuides, dashed] (axis cs: \p,0)   -- (axis cs: \p,2);

  \nextgroupplot[title={}, xlabel={}, axis y line=none, yticklabels=\empty]
    \addplot[colBrA, dashed, line width=0.9pt, domain=-\p:\p, samples=201, smooth] {1 - cos(deg(x))};              
    \addplot[colBrB, dashed, line width=0.9pt, domain=-\p:\p, samples=201, smooth] {1 - cos(deg(x + 2*\p/5))};     
    \addplot[colBrC, dashed, line width=0.9pt, domain=-\p:\p, samples=201, smooth] {1 - cos(deg(x + 4*\p/5))};     
    \addplot[colBrD, dashed, line width=0.9pt, domain=-\p:\p, samples=201, smooth] {1 - cos(deg(x + 6*\p/5))};     
    \addplot[colBrE, dashed, line width=0.9pt, domain=-\p:\p, samples=201, smooth] {1 - cos(deg(x + 8*\p/5))};     

    \addplot[colEnv, very thick, domain=-\p:-3*\p/5, samples=2] {1 - cos(deg(x + 4*\p/5))};  
    \addplot[colEnv, very thick, domain=-3*\p/5:-\p/5, samples=2] {1 - cos(deg(x + 2*\p/5))};
    \addplot[colEnv, very thick, domain=-\p/5:\p/5, samples=2] {1 - cos(deg(x))};            
    \addplot[colEnv, very thick, domain=\p/5:3*\p/5, samples=2] {1 - cos(deg(x + 8*\p/5))};  
    \addplot[colEnv, very thick, domain=3*\p/5:\p, samples=2] {1 - cos(deg(x + 6*\p/5))};    

    \draw[colGuides, dashed] (axis cs:-\p/5,0)   -- (axis cs:-\p/5,2);
    \draw[colGuides, dashed] (axis cs: \p/5,0)   -- (axis cs: \p/5,2);
    \draw[colGuides, dashed] (axis cs:-3*\p/5,0) -- (axis cs:-3*\p/5,2);
    \draw[colGuides, dashed] (axis cs: 3*\p/5,0) -- (axis cs: 3*\p/5,2);
    \draw[colGuides, dashed] (axis cs:-\p,0)     -- (axis cs:-\p,2);
    \draw[colGuides, dashed] (axis cs: \p,0)     -- (axis cs: \p,2);

  \end{groupplot}

  \node[anchor=north] at ($(group c1r1.south)!0.5!(group c2r1.south)$) {$\theta$};

  \node[anchor=south] at ($(group c1r1.north)!0.5!(group c2r1.north)$) {%
    \begin{tabular}{c}
        \hspace{-6cm}
      \textbf{SU(3):} cusps at $\theta=\pm\pi/3$ and $\theta=\pm\pi$ \\ \hspace{3cm}
      \textbf{SU(5):} cusps at $\theta=\pm\pi/5$, $\theta=\pm 3\pi/5$, and $\theta=\pm\pi$
   \end{tabular}
  };
\end{tikzpicture}
\caption{Colored schematic multi-branch vacuum energy $E(\theta)$ obtained as the lower envelope of shifted branch functions $E_0\!\left[1-\cos(\theta + 2\pi k/N)\right]$ for $N=3$ and $N=5$. 
Dashed colored curves show individual branches. 
The solid purple curve is the physical envelope. 
Vertical gray dashed lines mark cusp locations at $\theta=(2\ell+1)\pi/N$ and $\theta=\pm\pi$. 
The ground state lies in $[-\pi/N,\pi/N]$, realizing $|\theta_{\mathrm{eff}}|\le \pi/N$.}
\label{fig:envelope}
\end{figure}

\subsection{Anomaly Constraints and Topological Selection Rules}

Gauging the discrete shift symmetry modifies the global topology. 
Effectively one studies an $SU(N_c)/Z_N$ theory, rather than a strictly simply connected $SU(N_c)$ theory~\cite{Aharony:2013hda,Witten:1979vv,tHooft:1973alw}. 
The instanton number $Q=\frac{1}{8\pi^2}\!\int\!\text{Tr}\,F\wedge F$ is no longer an unrestricted integer. 
Only $Q\bmod N$ is invariant, which implies a fractional quantization
\begin{equation}\label{eq:SelectionRule}
Q\in\frac{1}{N}\mathbb Z.
\end{equation}
The path integral in the $Z_N$ gauged theory projects onto $Q\in N\mathbb Z$ sectors. 
Since $\exp(i\theta Q)$ must remain single valued under $\theta\to\theta+2\pi/N$, one must have $e^{i(2\pi/N)Q}=1$, so $Q\in N\mathbb Z$~\cite{Kapustin:2013uxa,Gaiotto:2014kfa}. 
Thus, the $N$ branches labeled by $\tilde Q=Q\bmod N$ collapse into a single physical vacuum ($\tilde Q=0$). 
This forces $\bar\theta_{\rm phys}$ into $[-\pi/N,\pi/N]$. 
The would-be $\theta\sim\pi$ vacuum reappears only as a metastable state separated by a membrane. 

Gauging $Z_N$ must respect all ’t Hooft anomalies, including those mixing with gravity~\cite{tHooft:1973alw,Freed:2014iua,Freed:2016rqq}. 
Let $Q_G$ denote the gravitational Pontryagin number. 
A $2\pi/N$ shift in $\theta$ generally induces $Q\to Q+1$ and a gravitational phase shift, so consistency demands
\begin{equation}\label{eq:GravitySelection}
Q+\kappa_G Q_G = N K,\qquad K\in\mathbb Z,
\end{equation}
where $\kappa_G\in\mathbb Z$ is fixed by anomaly inflow~\cite{Cordova:2019uob,Hsin:2020nts}. 
Thus, gauge and gravitational instantons contribute consistently only when $Q+\kappa_G Q_G$ is a multiple of $N$. 
These quantization conditions are exactly maintained in the discrete topological action, for example in the term $\frac{N}{2\pi}\!\int\!B\wedge da$ whose integer coefficient enforces $Q\in N\mathbb Z$ and is radiatively stable~\cite{Krauss:1988zc,Banks:2010zn}. 
Analogous gravitational couplings $\frac{\kappa_G}{2\pi}\!\int\!a\wedge\text{Tr}(R\wedge R)$ are quantized by anomaly cancellation~\cite{Cordova:2019uob,Freed:2014iua,Freed:2016rqq}. 
Hence the discrete $\theta$ projection is fully gauge and gravity safe. 

\subsection{Discrete Clockwork and Exponential $N_{\rm eff}$}

Large discrete gauge groups naturally emerge in ultraviolet theories such as extra-dimensional orbifolds, brane systems, or flux compactifications~\cite{Dvali:2005an,Kapustin:2013uxa}. 
A particularly efficient way to obtain exponentially large $N$ is via a clockwork mechanism in the topological sector~\cite{Kaplan:2015fuy,Giudice:2016yja}. 
Consider $K$ compact $BF$ sectors:
\begin{equation}\label{eq:BFclockwork}
S_{\rm clockwork}=\sum_{i=1}^K\frac{N_i}{2\pi}\int B_i\wedge da_i+\frac{1}{2\pi}\int(a_1+q_2 a_2+\cdots+q_K a_K)\wedge C_{\rm QCD},
\end{equation}
where each $BF$ term enforces a $Z_{N_i}$ invariance and $C_{\rm QCD}$ satisfies $dC_{\rm QCD}=\mathrm{Tr}(F\wedge F)$. 
QCD couples only to the effective combination $a_{\rm eff}=a_1+q_2 a_2+\cdots+q_K a_K$. 
With appropriate integer coefficients $q_i$, a single diagonal subgroup $Z_{N_{\rm eff}}$ of $Z_{N_1}\times\cdots\times Z_{N_K}$ remains, acting as $\theta\to\theta+2\pi/N_{\rm eff}$ with
\begin{equation}\label{eq:NeffProduct}
N_{\rm eff}=\prod_{i=1}^K N_i.
\end{equation}
Thus, even with modest $N_i\sim10^2$ to $10^4$, a chain of $K\sim10$ sectors yields $N_{\rm eff}\sim10^{10}$ to $10^{12}$, exponentially enhancing $\bar\theta$ suppression~\cite{Giudice:2016yja,Ahmed:2016viu}. 

As an explicit example, for $K=3$ with $q_2=N_1$ and $q_3=N_1N_2$, one finds $a_{\rm eff}=a_1+N_1a_2+N_1N_2a_3$. 
A $2\pi/N_{\rm eff}$ shift of $a_1$ can be compensated by fractional shifts of $a_2$ and $a_3$ that preserve each $BF$ term, leaving $a_{\rm eff}$ shifted by exactly $2\pi/N_{\rm eff}$. 
Hence $\theta$ inherits a $2\pi/N_{\rm eff}$ periodicity, producing the same envelope potential as Eq.~\eqref{eq:EnvelopeEnergy} with $N\to N_{\rm eff}$ and a suppression $|\bar\theta|\le\pi/N_{\rm eff}$. 
Figure~\ref{fig:clockworkNeff} illustrates this exponential enhancement. 

\begin{figure}[htbp]
\centering
\begin{tikzpicture}
  \pgfplotsset{compat=1.18}
  \definecolor{colTen}{RGB}{31,119,180}   
  \definecolor{colThree}{RGB}{214,39,40}  
  \definecolor{colGuide}{RGB}{120,120,120}

  \begin{axis}[
    width=14.2cm, height=14.2cm,
    xmin=0, xmax=10,
    ymin=0, ymax=10,
    xtick={0,2,4,6,8,10},
    ytick={0,2,4,6,8,10},
    xlabel={$K$ (number of sectors)},
    ylabel={$\log_{10} N_{\text{eff}}$},
    axis x line=bottom, axis y line=left,
    tick style={black},
    legend pos=south east,
    legend cell align=left,
    grid=none,
    grid style={line width=0.25pt, draw=gray!40},
    legend style= {draw=none}
  ]

    \addplot[colTen, very thick, domain=0:10, samples=2] {x};
    \addplot[colThree, very thick, dashed, domain=0:10, samples=2] {ln(3)/ln(10) * x};

    \addplot[colGuide, densely dotted, thick] coordinates {(10,0) (10,10)};
    \addplot[colTen, mark=*, mark size=2.2pt] coordinates {(10,10)};
    \addplot[colThree, mark=*, mark size=2.2pt] coordinates {(10,{(ln(3)/ln(10))*10})};

    \node[anchor=west, text=colTen] at (axis cs:1,9.3) {$N_0=10\ \Rightarrow\ \log_{10}N_{\text{eff}}=K$};
    \node[anchor=west, text=colThree] at (axis cs:1,8.2) {$N_0=3\ \Rightarrow\ \log_{10}N_{\text{eff}}=K\log_{10}3$};
    \node[anchor=south, text=black] at (axis cs:10,10) {$10^{10}$};
    \node[anchor=south west, text=black] at (axis cs:10,{(ln(3)/ln(10))*10}) {$3^{10}\approx 5.9\times10^{4}$};

    \legend{{$N_0 = 10$},{$N_0 = 3$}}
  \end{axis}
\end{tikzpicture}
\caption{Colored logarithmic enhancement of the effective symmetry order via a multi-sector clockwork. Plotted is $\log_{10}N_{\text{eff}}$ versus $K$ for $N_0=10$ (solid blue) and $N_0=3$ (dashed red), illustrating $N_{\text{eff}}=N_0^K$. Even modest $K$ yields an exponentially large $N_{\text{eff}}$.}
\label{fig:clockworkNeff}
\end{figure}
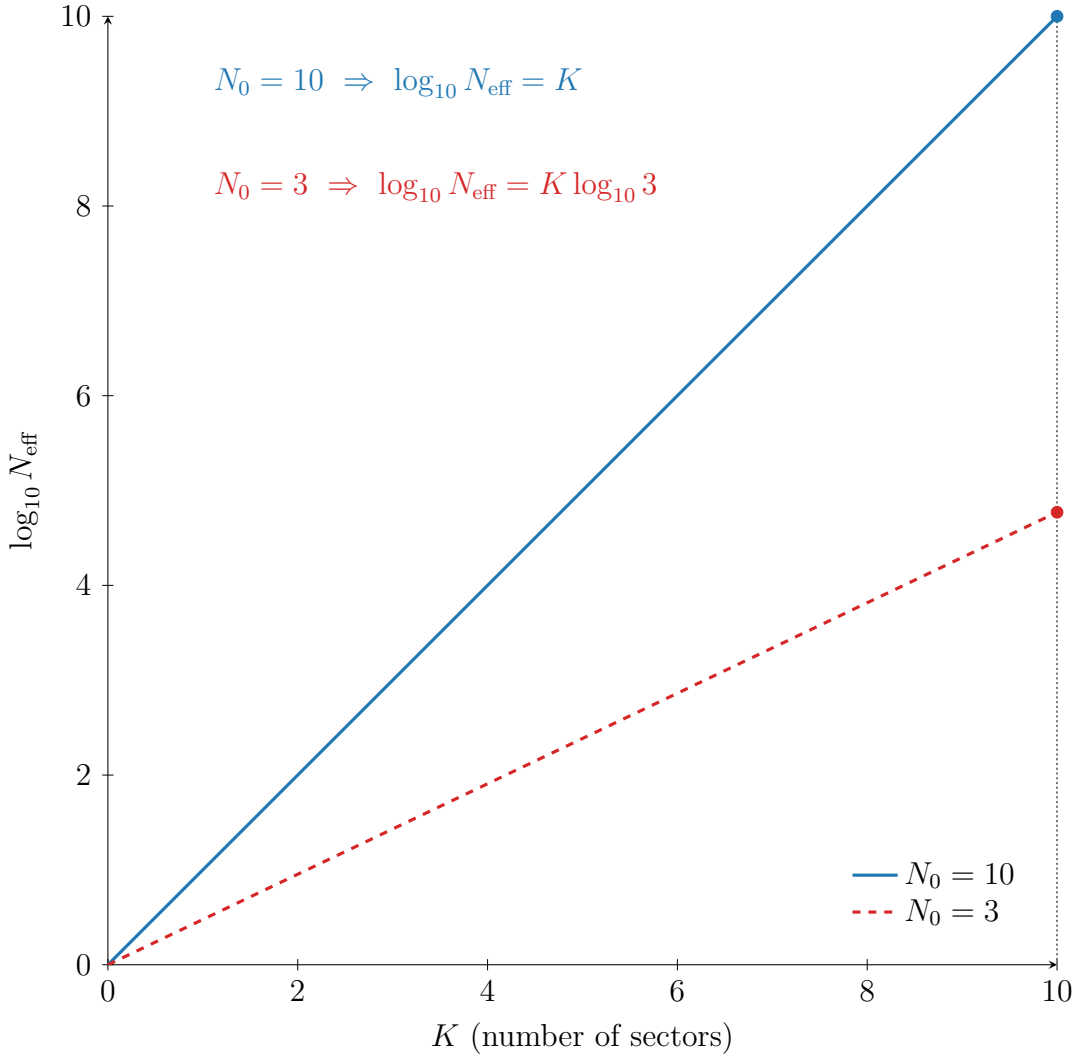

Each clockwork sector satisfies local anomaly constraints ($Q_i+\kappa_{G,i}Q_G^{(i)}\in N_i\mathbb Z$). 
This ensures that the effective diagonal combination coupling to QCD obeys the same condition for $N_{\rm eff}$. 
This structure can arise from higher-dimensional inflow or brane setups with quantized five dimensional Chern–Simons couplings, providing a holographic UV completion~\cite{Witten:1998uka,Kaloper:2008fb,Dvali:2005an}. 

\subsection{Metastable States, Membranes, and Energy Scales}

Even though gauging $Z_N$ removes exact $\theta$ degeneracy, remnants appear as $N-1$ metastable states and domain walls between them. The $m\neq m^*(\theta)$ branches in Eq.~\eqref{eq:EnvelopeEnergy} define local extrema with excess energy 
\begin{equation}
\Delta E_k=E_{\text{YM}}(2\pi k/N)-E_{\text{YM}}(0)\approx\frac{1}{2}\chi_t(2\pi k/N)^2.
\end{equation} 
These metastable states are connected by membranes carrying topological charge $\Delta Q=k$. The lowest-action process is the nucleation of a $k=1$ bubble bounded by a membrane of charge $\Delta Q=1$. Its tension scales as
\begin{equation}
\label{eq:MembraneTension}
\sigma\sim\frac{\Lambda}{g}\Delta E\sim N\Lambda^3,
\end{equation}
mirroring the $N_c$ scaling of large-$N_c$ domain walls in Yang-Mills \cite{Witten:1998uka,Bigazzi:2015bna}. The metastable spectrum exhibits 
\begin{equation}
\Delta E_k/E(0)\sim\pi^2 k^2/(2N^2),
\end{equation}
giving ultra-small splittings: for $N=10^{12}$, $\Delta E/E(0)\sim10^{-24}$. Thus, $E(\theta)$ appears almost analytic on the lattice, with curvature $\sim1/N^2$ and only a faint kink at $\theta=\pi/N$, testable via precise measurements of $\chi_t$.

Phenomenologically, the physical $\bar\theta$ is limited to $|\bar\theta_{\rm phys}|\le\pi/N$. Solving the strong CP problem requires $N\gtrsim10^{10}$, e.g.
\begin{equation}
N=\lceil \pi\times10^{10}\rceil\approx3.14\times10^{10},
\end{equation}
yielding $|\bar\theta_{\rm phys}|\lesssim3\times10^{-11}$, comfortably below experimental bounds $|\bar\theta|<10^{-10}$. For $N=10^{12}$ one obtains $|\bar\theta|\le3\times10^{-12}$. Such $N$ arise naturally from modest clockworks (e.g.\ $K=4$ with $N_i=10^3$). The corresponding neutron EDM $d_n\approx2.4\times10^{-16}\bar\theta\,e\cdot\text{cm}$ \cite{Crewther:1979pi} yields $d_n\lesssim8\times10^{-27}\,e\cdot\text{cm}$ for $N=3.14\times10^{10}$ and $d_n\lesssim8\times10^{-29}\,e\cdot\text{cm}$ for $N=10^{12}$, well below current limits ($1.8\times10^{-26}\,e\cdot\text{cm}$). Thus, ongoing null results in axion or EDM searches remain fully consistent with $Z_N$-gauged $\theta$-projection, which predicts no light axion and an exponentially suppressed $\bar\theta$, in contrast to Peccei-Quinn scenarios.

{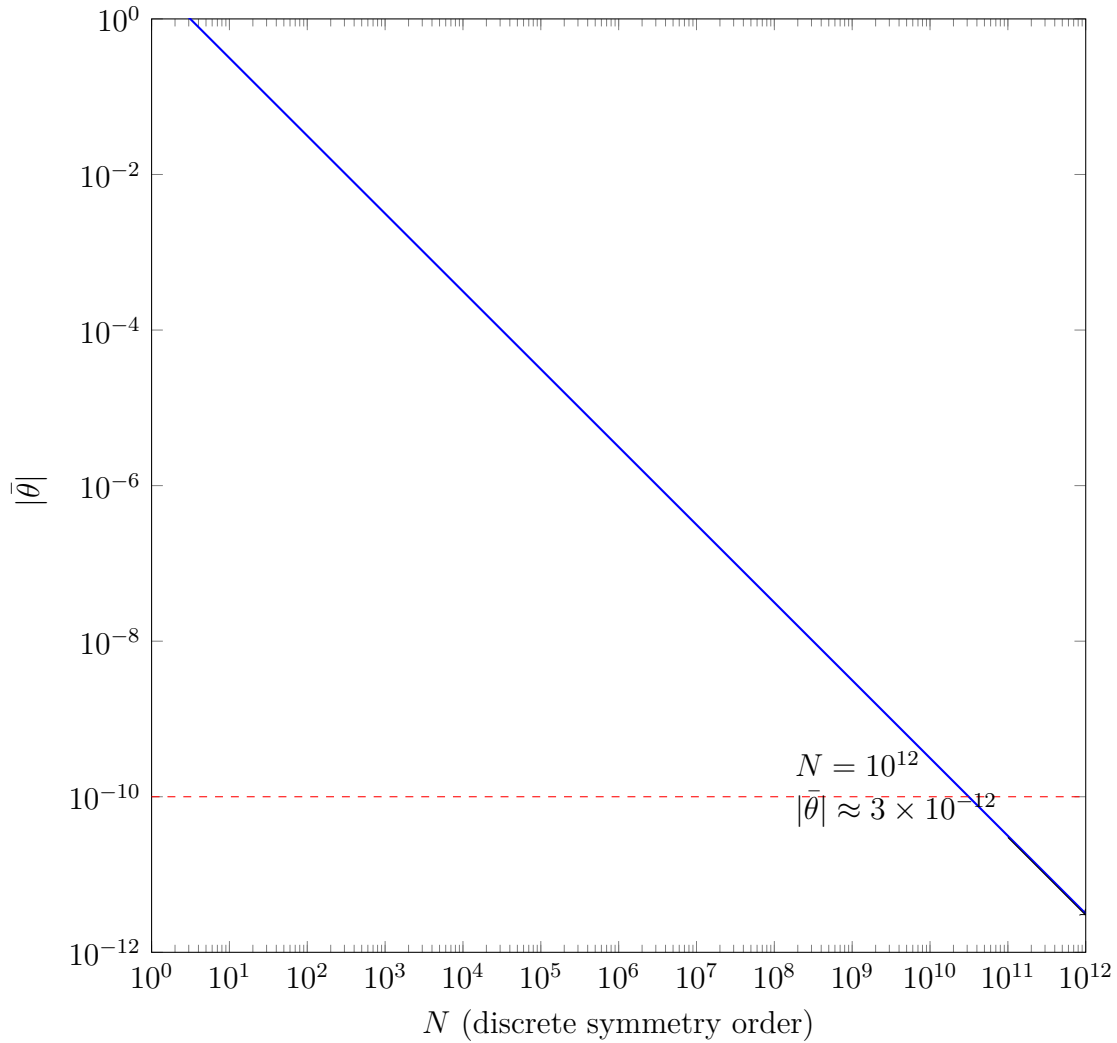
\begin{figure}[htbp]
\centering
\begin{tikzpicture}
    \begin{axis}[
        width=14cm, height=14cm,
        xmode=log, ymode=log,
        xmin=1e0, xmax=1e12,
        ymin=1e-12, ymax=1e0,
        xtick={1e0,1e1,1e2,1e3,1e4,1e5,1e6,1e7,1e8,1e9,1e10,1e11,1e12},
        ytick={1e0,1e-2,1e-4,1e-6,1e-8,1e-10,1e-12},
        xlabel={$N$ (discrete symmetry order)},
        ylabel={$|\bar{\theta}|\ $},
        legend style={draw=none}
    ]
        \addplot[blue, thick, domain=1:1e12, samples=200] {3.1416/x};
        \addplot[red, dashed, domain=1:1e12, samples=2] {1e-10};

        \node[anchor=south east, align=left] (largeNlabel)
          at (axis cs:1e11,3e-11) {$N=10^{12}$\\$|\bar{\theta}|\approx 3\times 10^{-12}$};
        \draw[->] (largeNlabel.south east) -- (axis cs:1e12, 3e-12);
    \end{axis}
\end{tikzpicture}
\caption{ Gauge-protected suppression of $|\bar{\theta}|$ as a function of $N$. The bound $|\bar{\theta}| = \pi/N$ (blue curve) falls below the current experimental limit $|\bar{\theta}| < 10^{-10}$ (red dashed line) once $N \gtrsim 3\times 10^{10}$. For example, $N = 10^{12}$ yields $|\bar{\theta}| \approx 3\times 10^{-12}$, far below the current bound.}
\end{figure}
}


\section{Diagnostics and Tests}
\label{sec:diagnostics}

The discrete gauging projects the Yang-Mills (YM) $\theta$-circle to an orbifold of circumference $2\pi/N$ and induces an $N$-branched structure in the vacuum energy~\cite{Witten:1998uka,Bonati:2016tvi}. 
Two inequivalent curvatures enter naturally and must not be conflated. 
The first is the local curvature inside a fixed cell, obtained by differentiating with respect to the cell coordinate that couples directly to the YM topological charge $Q$. 
This curvature equals the YM susceptibility $\chi_t$ by the anomalous Ward identity~\cite{Luscher:1981zq,Vicari:2008jw}. 
The second is the global, coarse-grained curvature defined with respect to the orbifold coordinate that parametrizes the full reduced $\theta$-circle of length $2\pi$ after the $\mathbb{Z}_N$ identification, it is suppressed by $N^2$. 
We now derive both statements from the exact projected partition function and prove that they are compatible.

Let 
\begin{equation}
Z_{\mathrm{YM}}(\vartheta)
=\sum_{Q\in\mathbb{Z}}Z_Q\,e^{i\vartheta Q}
\label{eq:ZYM}
\end{equation}
be the Yang-Mills partition function at vacuum angle $\vartheta$, with sector weights $Z_Q$ and 
\begin{equation}
\label{eq:10.2}
Q=\frac{1}{8\pi^2}\int \mathrm{Tr}\,F\wedge F \in \mathbb{Z}.
\end{equation}
Gauging the $\mathbb{Z}_N$ subgroup of the $2\pi$ shift identifies $\vartheta\sim\vartheta+2\pi/N$ and projects onto $Q\equiv0\pmod N$~\cite{Kapustin:2013qsa,Hsin:2020nts}. 
Equivalently, the projected partition function at microscopic angle $\theta$ is
\begin{equation}
Z_N(\theta)
=\frac{1}{N}\sum_{m=0}^{N-1}Z_{\mathrm{YM}}\!\left(\theta+\frac{2\pi m}{N}\right)
=\sum_{k\in\mathbb{Z}}Z_{Nk}\,e^{i\theta Nk}.
\label{eq:ZN}
\end{equation}

Define the thermodynamic vacuum energy
\begin{equation}
E(\theta)
=-\lim_{V\to\infty}\frac{1}{V}\log Z_N(\theta).
\end{equation}
Using
\begin{equation}
\log\sum_m e^{-V E_{\mathrm{YM}}(\theta+2\pi m/N)}
=-V\min_m E_{\mathrm{YM}}(\theta+2\pi m/N)+o(V),
\end{equation}
one obtains the infimal convolution (envelope) form
\begin{equation}
E(\theta)
=\min_{m\in\mathbb{Z}}E_{\mathrm{YM}}\!\left(\theta+\frac{2\pi m}{N}\right),\qquad V\to\infty,
\label{eq:envelope}
\end{equation}
so the unique ground state at any $\theta$ lies on the branch whose argument is closest to $0$ modulo $2\pi$~\cite{Bonati:2016tvi,Gaiotto:2017yup}. 
Writing that selected representative as the “cell” (principal-branch) coordinate
\begin{equation}
\theta_{\mathrm{cell}}
=\mathrm{PV}\!\left(\theta+\frac{2\pi m^{*}(\theta)}{N}\right)\in\left[-\frac{\pi}{N},\frac{\pi}{N}\right],
\label{eq:thetachange}
\end{equation}
the energy inside a cell is simply $E(\theta)=E_{\mathrm{YM}}(\theta_{\mathrm{cell}})$ and cusps occur when $m^{*}$ jumps at $\theta=(2\ell+1)\pi/N$. 
This construction is exactly the one used in Secs.~\ref{sec:3.6} and \ref{sec:diagnostics} of the manuscript. Eq.~\eqref{eq:envelope} is the envelope limit discussed below Eq.~\eqref{eq:ZYM}, and Eq.~\eqref{eq:thetachange} coincides with the principal-cell definition leading to $|\theta_{\mathrm{cell}}|\le\pi/N$.

Two angular parametrizations are now natural and must be kept distinct. 
The first is the microscopic cell coordinate $\theta_{\mathrm{cell}}$, which couples to the microscopic topological charge $Q$ in~\eqref{eq:ZN}. 
The second is the orbifold coordinate $\hat\theta=N\theta_{\mathrm{cell}}\in[-\pi,\pi]$, which spans the reduced $2\pi$ circle after the $\mathbb{Z}_N$ gauging and couples to the coarse-grained charge $K\equiv Q/N$. 
These are not two different theories; they are two coordinate choices on the same projected ensemble $Z_N$ in~\eqref{eq:ZN}, related by the simple rescaling $\hat\theta=N\theta_{\mathrm{cell}}$.

Curvatures follow from standard cumulant relations. 
Differentiating~\eqref{eq:ZN} with respect to the source that couples to $Q$ (that is $\theta_{\mathrm{cell}}$) gives
\begin{equation}
\frac{\partial E}{\partial \theta_{\mathrm{cell}}}\Big|_{0}
=iV\,\frac{\langle Q\rangle_N}{V}\Big|_{\theta=0}=0,\qquad
\frac{\partial^2 E}{\partial \theta_{\mathrm{cell}}^2}\Big|_{0}
=\frac{1}{V}\langle Q^2\rangle_{N,\theta=0}^{c}.
\label{eq:cellcurv1}
\end{equation}
In the large-volume saddle, the sector weights in~\eqref{eq:ZN} are Gaussian in $Nk$
\begin{equation}
Z_{Nk}\propto \exp\!\left[-\frac{(Nk)^2}{2V\chi_t}+\cdots\right].
\end{equation}
Hence the projection to $Q\in N\mathbb{Z}$ leaves the microscopic variance unchanged,
\begin{equation}
\frac{1}{V}\langle Q^2\rangle_{N,\theta=0}^{c}=\chi_t,
\label{eq:localvar}
\end{equation}
and therefore the local (cellwise) curvature equals the Yang-Mills topological susceptibility,
\begin{equation}
\frac{\partial^2 E}{\partial \theta_{\mathrm{cell}}^2}\Big|_{0}=\chi_t.
\label{eq:cellcurv2}
\end{equation}
This is precisely what Sec.~\ref{sec:3.6} asserts inside a fixed principal cell the envelope coincides with the underlying YM branch and its curvature equals $\chi_t=\int d^4x\,\langle q(x)q(0)\rangle_c$~\cite{Luscher:1981zq,DelDebbio:2004ns}. 
The same conclusion follows directly from the envelope form~\eqref{eq:envelope}: for fixed $m^*$ one has $E(\theta)=E_{\mathrm{YM}}(\theta_{\mathrm{cell}})$, so $\partial_{\theta_{\mathrm{cell}}}^2E|_0=\partial_{\vartheta}^2E_{\mathrm{YM}}|_{\vartheta=0}=\chi_t$.

Switch now to the orbifold coordinate $\hat\theta=N\theta_{\mathrm{cell}}$, whose source couples to $K=Q/N$. 
Differentiating with respect to $\hat\theta$ yields
\begin{equation}
\frac{\partial E}{\partial \hat\theta}\Big|_{0}
=iV\,\frac{\langle K\rangle_N}{V}\Big|_{\theta=0}=0,\qquad
\frac{\partial^2 E}{\partial \hat\theta^2}\Big|_{0}
=\frac{1}{V}\langle K^2\rangle_{N,\theta=0}^{c}
=\frac{1}{N^2}\frac{1}{V}\langle Q^2\rangle_{N,\theta=0}^{c}
=\frac{\chi_t}{N^2}.
\label{eq:orbicurv}
\end{equation}
Equivalently, since $\hat\theta=N\theta_{\mathrm{cell}}$, one has the chain rule identity
\begin{equation}
\frac{\partial}{\partial \hat\theta}
=\frac{1}{N}\frac{\partial}{\partial \theta_{\mathrm{cell}}}
\quad\Rightarrow\quad
\frac{\partial^2E}{\partial \hat\theta^2}
=\frac{1}{N^2}\frac{\partial^2E}{\partial \theta_{\mathrm{cell}}^2},
\label{eq:chain}
\end{equation}
which reproduces~\eqref{eq:orbicurv}. 
Equations~\eqref{eq:cellcurv2}–\eqref{eq:chain} prove the pair of statements
\begin{equation}
\frac{\partial^2E}{\partial \theta_{\mathrm{cell}}^2}\Big|_{0}=\chi_t,\qquad
\frac{\partial^2E}{\partial \hat\theta^2}\Big|_{0}=\frac{\chi_t}{N^2},
\label{eq:pair}
\end{equation}
and make explicit why there is no contradiction. 
The first curvature is local, it is taken with respect to the microscopic cell coordinate that couples directly to $Q$. 
The second is global; it is taken with respect to the orbifold angle that scans the reduced circle and couples to $K=Q/N$. 
The $N^{-2}$ factor is purely kinematic, arising from the change of variable $\hat\theta=N\theta_{\mathrm{cell}}$. 
This is exactly the distinction articulated in the paper’s Sec.~\ref{sec:diagnostics}, Eqs.~\eqref{eq:10.2} to \eqref{eq:thetachange}, where the two coordinates are introduced and the two curvatures are read off as the variances of $Q$ and $K$ in the same projected ensemble. 
Figure~\ref{fig:local-vs-global} (for $N=3$) illustrates that the local parabolic curvature inside each cell remains $\chi_t$, while the global modulation over the full $2\pi$ is flattened by $N^{-2}$~\cite{Vicari:2008jw,DelDebbio:2004ns}. 

It is instructive to phrase these results as susceptibilities. 
Define the local susceptibility
\begin{equation}
\chi_{\mathrm{local}}
\equiv\frac{1}{V}\int d^4x\,\langle q(x)q(0)\rangle_c
=\frac{1}{V}\langle Q^2\rangle_{N,\theta=0}^{c}
=\chi_t,
\label{eq:localchi}
\end{equation}
and the global susceptibility with respect to the orbifold angle
\begin{equation}
\chi_{\mathrm{global}}
\equiv\frac{\partial^2E}{\partial \hat\theta^2}\Big|_{0}
=\frac{1}{V}\langle K^2\rangle_{N,\theta=0}^{c}
=\frac{1}{N^2}\chi_{\mathrm{local}}.
\label{eq:globalchi}
\end{equation}
Both are computed in the same $\mathbb{Z}_N$ projected ensemble~\eqref{eq:ZN}. 
The projection leaves the microscopic fluctuations of $Q$ intact, hence $\chi_{\mathrm{local}}=\chi_t$, but suppresses the variance of the coarse-grained charge $K=Q/N$ by $N^{-2}$. 
In the thermodynamic limit this picture coexists with the envelope representation~\eqref{eq:envelope} inside any open cell the function $E(\theta)$ is analytic with curvature $\chi_t$. At the cell boundaries $\theta=(2\ell+1)\pi/N$ neighboring branches exchange dominance, producing a cusp in $\partial_\theta E$ but not altering the conclusions at $\theta=0$~\cite{Bonati:2016tvi,Gaiotto:2017yup}. 

Putting everything together resolves the apparent inconsistency between Eq.~\eqref{eq:axion-on-shell} and Eq.~\eqref{eq:10.2} of the manuscript. 
Equation~\eqref{eq:axion-on-shell} is a cellwise statement, hence differentiating the envelope with respect to the variable that probes the underlying YM branch gives curvature $\chi_t$. 
Equation~\eqref{eq:10.2} defines a global variable that scans the reduced circle, differentiating with respect to that variable introduces the Jacobian $1/N$ on first derivatives and $1/N^2$ on second derivatives. 
The global curvature suppression is therefore nothing but the kinematic rescaling $\hat\theta=N\theta_{\mathrm{cell}}$. 
It does not, and cannot, contradict the local equality of curvatures proved in Sec.~\ref{sec:3.6}. 
Operationally this gives a sharp lattice diagnostic: measure $(1/V)\langle Q^2\rangle$ and $(1/V)\langle(Q/N)^2\rangle$ in the projected theory at $\theta=0$. 
The former equals $\chi_t$ while the latter equals $\chi_t/N^2$. 
The same identities follow directly from~\eqref{eq:ZN}–\eqref{eq:thetachange} and hold at any finite volume; in $V\to\infty$ they coexist with the envelope picture and the selection rule $|\theta_{\mathrm{cell}}|\le\pi/N$~\cite{DelDebbio:2004ns,Luscher:1981zq,FlavourLatticeAveragingGroupFLAG:2021npn}. 

The curvature at $\theta=0$ with respect to the microscopic cell angle is $\chi_t$. 
The curvature with respect to the orbifold angle is $\chi_t/N^2$. 
Both descend from the same projected partition function, the former probing fluctuations of $Q$, the latter probing fluctuations of $K=Q/N$. 
The suppression is exclusively a consequence of the coordinate choice that spans the orbifolded $\theta$-circle and does not modify the local Yang-Mills curvature inside a single topological cell. 
The apparent discrepancy is thus eliminated, and the equalities $\chi_{\mathrm{local}}=\chi_t$ and $\chi_{\mathrm{global}}=\chi_t/N^2$ provide a precise and testable diagnostic.

It is useful to formalize the orbifold projection as a linear averaging on the partition function and as an infimal convolution on the energy. 
Define the discrete projector $\mathsf{P}_N$ acting on $2\pi$ periodic functions by
\begin{equation}
(\mathsf{P}_N f)(\theta)\;\equiv\;\frac{1}{N}\sum_{m=0}^{N-1} f\!\left(\theta+\frac{2\pi m}{N}\right),
\qquad
Z_{N}(\theta)=(\mathsf{P}_N Z_{\rm YM})(\theta).
\label{eq:PN}
\end{equation}
In the $V\to\infty$ limit,
\begin{equation}
E(\theta)\;=\;-\frac{1}{V}\log (\mathsf{P}_N e^{-V E_{\rm YM}})(\theta)
\;\xrightarrow[V\to\infty]{}\;\inf_{m\in\mathbb{Z}}\,E_{\rm YM}\!\left(\theta+\frac{2\pi m}{N}\right),
\label{eq:infconv}
\end{equation}
that is, the physical energy is the lower envelope of the shifted YM branches. 
The global curvature suppression \eqref{eq:globalchi} can be read either from the charge-variance statement or directly from the rescaling $E(\hat\theta)=E(\hat\theta/N)$ near the origin; both viewpoints are equivalent in the projected ensemble~\cite{Bonati:2016tvi}. 
The effect of the discrete $\theta$ projection on the Yang-Mills vacuum energy is illustrated in Fig.~\ref{fig:envelope1} for a toy model with
$E(\theta)=1-\cos\theta$, which is $2\pi$ periodic, even, and convex near the origin, as in the large-$N_c$ picture~\cite{tHooft:1973alw,Witten:1979vv}. 
Gauging a $\mathbb{Z}_N$ subgroup of the $\theta$-shift symmetry produces an $N$-branch structure $E(\theta+2\pi m/N)$, $m=0,\dots,N-1$, and the physical vacuum energy $E_N(\theta)$ is given by the lower envelope of these shifted branches. 
For $N=4$, the blue curve shows the original Yang-Mills branch, while the colored dashed curves represent the shifted branches $E(\theta+2\pi m/4)$. 
The thick black curve is $E_4(\theta)$, obtained by taking, at each $\theta$, the minimum over $m$ of these branch energies, in agreement with the projected partition function construction in Eq.~\eqref{eq:ZN}.

Near $\theta=0$ the $m=0$ branch is uniquely minimal, so $E_4(\theta)$ coincides with the original $E(\theta)$ and reproduces the usual local curvature $\chi_t$ in the microscopic cell. 
At $|\theta|=\pi/4$ the $m=0$ and $m=1$ branches become degenerate, leading to a cusp in $E_4(\theta)$; for larger $|\theta|$ the envelope follows the next branch, corresponding to an effective angle $\theta_{\mathrm{eff}}=\theta-2\pi/4$. 
All cusps occur at $|\theta|=(2\ell+1)\pi/4$, and between them the envelope is analytic, tracking a single translated Yang-Mills branch. 
This construction makes the projection bound $|\bar{\theta}|\le\pi/N$ geometrically manifest. The effective vacuum angle $\bar{\theta}=\theta_{\mathrm{eff}}$ that labels the occupied branch always lies in the fundamental interval $[-\pi/N,\pi/N]$, while the global curvature in the orbifold coordinate is suppressed by $1/N^2$ compared to the microscopic Yang-Mills curvature, as quantified in Eqs.~\eqref{eq:localchi}–\eqref{eq:globalchi}.

Let $\langle\cdot\rangle_{\rm YM}$ denote YM expectation at $\theta=0$. 
Using Eq.~\eqref{eq:ZN} and the Gaussian saddle for large $V$, the $Q$-distribution in the $\mathbb{Z}_N$ ensemble is the YM distribution conditioned to $Q\in N\mathbb{Z}$. 
Writing $Q=Nk$ and expanding the YM weight around $k=0$ gives $Z_{Nk}\propto\exp[-N^2 k^2/(2V\chi_t)]$, hence
\begin{equation}
\big\langle Q^2\big\rangle_{N,\theta=0}
= N^2\,\big\langle k^2\big\rangle_{N,\theta=0}
= N^2\times \frac{V\chi_t}{N^2}
= \big\langle Q^2\big\rangle_{\rm YM,\theta=0}.
\label{eq:varQ}
\end{equation}
Equation \eqref{eq:varQ} shows that the microscopic variance of $Q$ is unchanged by the projection, while the diagonal variance,
\begin{equation}
\big\langle K^2\big\rangle_{N,\theta=0}=\frac{1}{N^2}\,\big\langle Q^2\big\rangle_{\rm YM,\theta=0},
\label{eq:varK}
\end{equation}
is reduced by $N^2$. 
Combining \eqref{eq:varQ} – \eqref{eq:varK} with \eqref{eq:cellcurv2} – \eqref{eq:orbicurv} reproduces \eqref{eq:chilocal} – \eqref{eq:chiglobal}.

We now record the diagnostic formulas that separate the two curvatures. 
The local susceptibility is the standard connected two-point function of the YM topological density,
\begin{equation}
\chi_{\rm local}\;\equiv\;\frac{1}{V}\int_V d^4x\;\langle q(x)\,q(0)\rangle_c
\;=\;\left.\frac{\partial^2 E}{\partial\theta_{\rm cell}^{\,2}}\right|_{\theta_{\rm cell}=0}
\;=\;\chi_t\,.
\label{eq:chilocal}
\end{equation}
The global susceptibility is the curvature with respect to the orbifold coordinate (the angle of circumference $2\pi$ after the $\mathbb{Z}_N$ identification),
\begin{equation}
\chi_{\rm global}\;\equiv\;\left.\frac{\partial^2 E}{\partial\hat\theta^{\,2}}\right|_{\hat\theta=0}
\;=\;\frac{1}{V}\,\langle K^2\rangle_{N,\theta=0}^{c}
\;=\;\frac{1}{N^2}\,\chi_{\rm local}\,.
\label{eq:chiglobal}
\end{equation}
Equations \eqref{eq:chilocal}–\eqref{eq:chiglobal} are exact in the projected theory and hold at any $V$. 
They provide a clean way to quote Eq.~\eqref{eq:10.2} without ambiguity. The suppressed curvature is the curvature with respect to the orbifold angle, that is the fluctuation of the diagonal charge $K=Q/N$. 
The microscopic curvature inside a cell remains $\chi_t$ and saturates the anomalous Ward identity~\cite{Luscher:1981zq,Bilal:2008qx}.

Gauging the discrete shift enforces the selection rule $Q\in N\mathbb{Z}$ and identifies all angles separated by $2\pi/N$~\cite{Kapustin:2013qsa,Hsin:2020nts}. 
Instanton physics inside a cell is unaltered, the local correlator $\langle q(x)q(0)\rangle_c$ and hence $\chi_{\rm local}$ retain their YM values. 
The gauging removes are global fluctuations of $Q$ modulo $N$. 
Measuring topological response with respect to the orbifold angle is tantamount to probing the variance of $K=Q/N$, which is $N^{-2}$ of the YM variance. 
Thus large-scale curvature is suppressed even though local topological fluctuations are intact.

A lattice implementation that distinguishes the two curvatures follows directly from \eqref{eq:chilocal}–\eqref{eq:chiglobal}. 
To measure $\chi_{\rm local}$, simulate the projected theory at $\theta_{\rm cell}=0$ with the discrete two-form sector included (or project a YM ensemble to $Q\in N\mathbb{Z}$ in analysis) and integrate the connected correlator $\langle q(x)q(0)\rangle_c$~\cite{DelDebbio:2004ns,Luscher:1981zq,FlavourLatticeAveragingGroupFLAG:2021npn}. 
Equivalently, extract $\partial_{\theta_{\rm cell}}^2 E|_{0}$ by standard small-$\theta$ reweighting restricted to the principal cell $|\theta_{\rm cell}|<\pi/N$. 
To measure $\chi_{\rm global}$, couple the source to the orbifold angle $\hat\theta\in[-\pi,\pi]$ and determine $\partial_{\hat\theta}^2 E|_0$. 
By \eqref{eq:orbicurv}–\eqref{eq:chiglobal} this is identical to computing the variance of $K=Q/N$,
\begin{equation}
\chi_{\rm global}\;=\;\frac{1}{V}\,\Big\langle \big(Q/N\big)^2\Big\rangle_{N,\theta=0}
\;=\;\frac{1}{N^2}\,\frac{1}{V}\,\langle Q^2\rangle_{\rm YM,\theta=0}\,,
\label{eq:latticerel}
\end{equation}
which provides a robust cross-check of the $1/N^2$ law. 
In practice, both determinations share the same underlying measurements of $Q$ and of the two-point function $q(x)q(0)$, differing only by whether one differentiates with respect to $\theta_{\rm cell}$ or to $\hat\theta$.

The identities above hold at finite $V$. 
As $V\to\infty$, the projected free energy equals the lower envelope \eqref{eq:infconv}, the ground state lies in the principal cell $|\bar\theta_{\rm eff}|\le\pi/N$, and metastable branches are split from the minimum by
\(
\Delta E_k\simeq \frac12\chi_t (2\pi k/N)^2
\),
while the domain wall tension scales linearly in $N$ in the topological sector~\cite{Bonati:2016tvi,Gaiotto:2017yup}. 
As $N\to\infty$ at fixed $\chi_t$, the envelope becomes flat on $\mathcal{O}(1)$ angular scales, $\chi_{\rm global}\to 0$ with $\chi_{\rm global}\sim\chi_t/N^2$, whereas $\chi_{\rm local}=\chi_t$ remains finite.

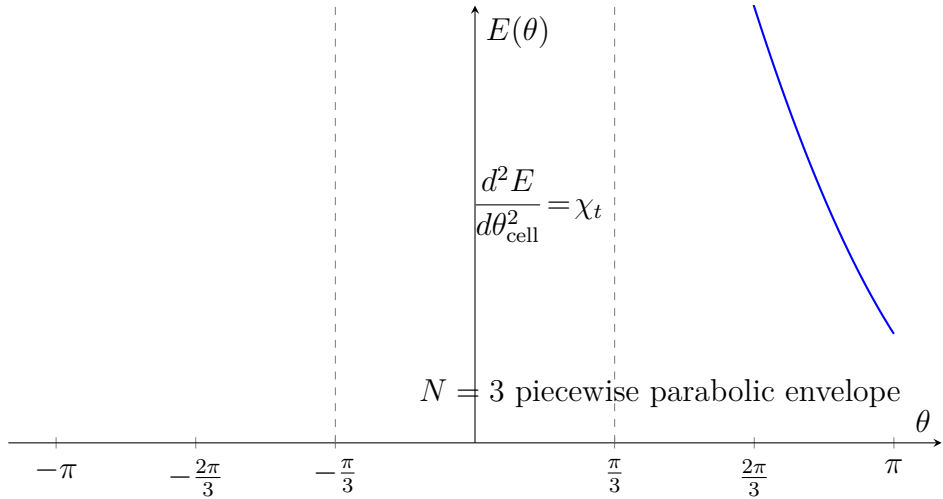
\begin{figure}[t]
\centering
\begin{tikzpicture}
\begin{axis}[
  width=0.9\textwidth,height=7.4cm,
  xmin=-3.5, xmax=3.5, ymin=0, ymax=2.2,
  axis lines=middle, xlabel={$\theta$}, ylabel={$E(\theta)$},
  xtick={-3.1416,-2.0944,-1.0472,0,1.0472,2.0944,3.1416},
  xticklabels={$-\pi$,$-\tfrac{2\pi}{3}$,$-\tfrac{\pi}{3}$,$0$,$\tfrac{\pi}{3}$,$\tfrac{2\pi}{3}$,$\pi$},
  ytick=\empty, domain=-3.1416:3.1416, samples=800, smooth]
\addplot+[no marks,thick,blue] {(x+2*pi)^2 < (pi/3)^2 ? 0.5*(x+2*pi)^2 :
                                (x+4*pi/3)^2 < (pi/3)^2 ? 0.5*(x+4*pi/3)^2 :
                                (x+2*pi/3)^2 < (pi/3)^2 ? 0.5*(x+2*pi/3)^2 :
                                (x)^2 < (pi/3)^2 ? 0.5*(x)^2 :
                                (x-2*pi/3)^2 < (pi/3)^2 ? 0.5*(x-2*pi/3)^2 :
                                (x-4*pi/3)^2 < (pi/3)^2 ? 0.5*(x-4*pi/3)^2 :
                                0.5*(x-2*pi)^2};
\draw[dashed,gray] (axis cs:-1.0472,0) -- (axis cs:-1.0472,2.2);
\draw[dashed,gray] (axis cs: 1.0472,0) -- (axis cs: 1.0472,2.2);
\node at (axis cs:-0.1,1.2) [anchor=west] {$\displaystyle \frac{d^2E}{d\theta_{\rm cell}^2}\!=\!\chi_t$};
\node at (axis cs:-0.5,0.25) [anchor=west] {$N=3$ piecewise parabolic envelope};
\end{axis}
\end{tikzpicture}
\caption{Schematic envelope $E(\theta)=\min_m \tfrac12\chi_t(\theta+2\pi m/N)^2$ for $N=3$. The local curvature inside a cell (solid parabolic segment) equals $\chi_t$. Measuring curvature with respect to the orbifold angle $\hat\theta=N\theta_{\rm cell}$ yields a global curvature suppressed by $N^{-2}$.}
\label{fig:local-vs-global}
\end{figure}

Equations \eqref{eq:chilocal}–\eqref{eq:chiglobal} supply the precise operational meaning of “curvature suppression” used in Sec.~\ref{sec:diagnostics}. 
It is suppression in the orbifold variable, i.e. in the diagonal response to the projected, $2\pi$-periodic angle, while the microscopic cell-wise curvature is unsuppressed and controlled by the YM anomaly. 
This is the unique way to read both Eq.~\eqref{eq:axion-on-shell} and the global scaling quoted in Sec.~\ref{sec:diagnostics} as simultaneously true. 
The definitions above also match lattice practice, $\chi_{\rm local}$ is the standard integrated correlator of $q$, and $\chi_{\rm global}$ is the variance of $Q/N$ or, equivalently, the second derivative of the free energy with respect to the orbifold angle (the angle of the reduced circle)~\cite{Vicari:2008jw,DelDebbio:2004ns,FlavourLatticeAveragingGroupFLAG:2021npn}. 
Together with the exact projector \eqref{eq:ZN} and the envelope limit \eqref{eq:infconv}, they provide a complete and testable set of diagnostics for the discrete projection mechanism~\cite{Veneziano:1979ec,Witten:1998uka,Gaiotto:2014kfa,Cordova:2019uob,Hason:2020yqf}.

\section{Conclusions}
\label{sec:conclusions}

We have shown that gauging a finite subgroup $Z_{N}$ of the $2\pi$ $\theta$-shift in QCD produces an orbifolded path integral that selects, in the thermodynamic limit, the lowest branch of the multi-valued Yang--Mills vacuum energy. This yields
\(
E(\theta)=\min_{m\in\mathbb Z}E_{\rm YM}\!\left(\theta+\frac{2\pi m}{N}\right),
\)
which projects the physical theory onto the principal cell $|\bar\theta_{\rm eff}|\le\pi/N$ and solves the strong CP problem without introducing a light axion. 
The bound is the dynamics. It is protected by a discrete gauge symmetry rather than an approximate global symmetry. It persists under renormalization and in curved space because the topological couplings are integer quantized and therefore nonrenormalizable in the sense that they are immutable under small deformations. We established the underlying selection rule in a local effective field theory with compact higher-form fields,
\(
Q+\kappa_{G}Q_{G}=NK, \, \mathrm{with} \, K\in\mathbb Z,
\)
and we demonstrated that the integer data $(N,\kappa_{G})$ follow from higher-form gauge invariance and from five-dimensional anomaly inflow.  This ensures gauge-gravity consistency while enforcing that only sectors with $Q\in N\mathbb Z$ survive. Equivalently, $\theta$ is a coordinate on a circle orbifolded by $Z_{N}$, with the vacuum always realigning to the representative closest to $0$. We quantified the infrared consequences of this projection. The envelope suppresses the curvature of $E(\theta)$ by $N^{-2}$ in the small-angle regime, with $\chi_{N}\sim \chi_{t}/N^{2}$. 
It lifts would-be $\theta\sim2\pi k/N$ vacua to exponentially nearby metastable branches with gaps
\(
\Delta E_{k}\simeq\frac{1}{2}\,\chi_{t}\left(\frac{2\pi k}{N}\right)^{2},
\)
and it replaces axionic domain walls by membrane-like excitations whose tension scales as $\sigma\sim N\Lambda^{3}$ and whose nucleation is catastrophically suppressed for large $N$. Vacuum selection therefore completes long before Big Bang Nucleosynthesis once $\chi_{t}(T)$ turns on at the QCD epoch. We exhibited ultraviolet realizations that deliver exponentially large suppression with modest microscopic inputs via discrete clockwork chains of compact $BF$ sectors. In these models the diagonal gauge symmetry seen by QCD has
\(
N_{\rm eff}=\prod_{i}N_{i}
\)
and thus $|\bar\theta_{\rm phys}|\le\pi/N_{\rm eff}$. 
This construction preserves all local and mixed anomalies site by site and in the diagonal by design. It can be engineered in higher-dimensional inflow or brane-inspired settings. These constructions are intended as controlled four-dimensional effective descriptions of QCD coupled to compact topological sectors, rather than as complete embeddings of the full Standard Model plus gravity. Identifying such fully microscopic realizations in a given ultraviolet framework is an interesting problem that we leave for future work. Phenomenologically the framework makes sharp, falsifiable predictions. 
The QCD-induced neutron EDM scales as
\(
|d_{n}|\sim (1 \text{ to } 3)\times10^{-16}\,\frac{\pi}{N}\,e\cdot{\rm cm}
\)
and becomes unobservably small for $N\gtrsim 10^{10}$. CKM-induced contributions remain as in the Standard Model. All axion-like signals are absent because the shift is gauged and not global, so no pseudo Nambu-Goldstone mode exists and no $1/p^{2}$ pole appears in $\langle q,q\rangle$. On the lattice one should observe a flattened $\theta$ dependence with cusps at $\theta=(2\ell{+}1)\pi/N$, a suppressed curvature consistent with $\chi_{N}\propto N^{-2}$, and a branch envelope matching the projected multi-branch energy.  Conceptually, discrete $\theta$ projection reframes strong CP as a question of generalized gauge structure rather than of approximate global symmetries. 
It therefore avoids axion quality issues and makes the smallness of $\bar\theta$ a corollary of integer quantized topological data $(N,\kappa_{G})$. Looking ahead, it will be valuable to systematize anomaly-safe $Z_{N}$ embeddings across QCD-like theories with dynamical flavors, to map the space of inflow completions and two-group structures compatible with the selection rule, and to develop precision lattice diagnostics that separate envelope curvature from higher-cumulant contributions. It will also be valuable to integrate the mechanism in unified or composite ultraviolet models where discrete clockwork arises naturally. In parallel, sustained null results in axion searches and continued tightening of EDM bounds will progressively corner axionic solutions while remaining entirely compatible with the gauge-protected, axionless dynamics advocated here.

\section*{Acknowledgements}
We are grateful to 

\appendix

\section{Proof of the projection bound}
\label{app:proof}

Let $E(\theta)$ denote the vacuum energy density as a function of the CP angle $\theta$. 
In a pure Yang-Mills theory (or QCD in the $m_{q}\to 0$ limit), $E(\theta)$ is known to be a $2\pi$-periodic, continuous function. 
It is even, in the sense that $E(-\theta)=E(\theta)$, and it has a global minimum at $\theta=0$ by CP symmetry. 
Moreover, $E(\theta)$ is convex for $|\theta|$ sufficiently small. 
Indeed, $E''(0)=\chi_t>0$, where $\chi_t$ is the topological susceptibility, given by the anomalous Ward identity. 
Thus, near $\theta=0$ one has $E(\theta)\approx \frac{1}{2}\chi_t\,\theta^2+\dots$ with positive curvature. 

Analyses at large $N_c$ further reveal that $E(\theta)$ can be written as the lower envelope of an infinite family of $2\pi$-shifted branches. 
In particular, one may introduce an auxiliary $2\pi$-periodic even function $f(x)$, smooth for $x\in\mathbb{R}$, such that
\begin{equation}\label{E_multi_branch}
E(\theta) = \min_{k\in\mathbb{Z}}\, A\, f\!\Big(\frac{\theta+2\pi k}{c}\Big),
\end{equation}
for some constants $A,c>0$ (with $A$ scaling like $N_c^2$ and $c\sim N_c$ in Yang-Mills). 
This multi-branch representation, first derived by Veneziano, Di Vecchia, Witten and collaborators, makes the convex and periodic structure explicit. 
Each branch $E_k(\theta)\equiv A\,f((\theta+2\pi k)/c)$ is even and convex on $[-\pi,\pi]$, and $E(\theta)$ is their pointwise minimum. 
In particular, $E(\theta)$ is continuous and convex on $[-\pi,\pi]$, with possible non-analytic kinks (cusps) only where two branches intersect. 
These intersections occur at the half-period points $\theta=(2\ell+1)\pi$ for integers $\ell$. 
At those points, by symmetry, $E(\pi^-)=E(\pi^+)$ and a first-order transition can arise. 
Indeed, for $N_c>2$ one expects a cusp at $\theta=\pi$ related to spontaneous CP breaking, as argued by Witten and observed in lattice-inspired models. 
Away from such points, $E(\theta)$ is smooth (infinitely differentiable) and strictly increasing for $\theta>0$, since the slope $E'(\theta)$ goes from $0$ at $\theta=0$ to $E'(\pi^-)>0$. 

We now gauge a discrete $\mathbb{Z}_N$ subgroup of the shift symmetry $\theta\mapsto\theta+2\pi$. 
Physically, this means identifying $\theta$ with $\theta+2\pi/N$ as equivalent angles, an orbifold of the $\theta$-circle. 
Equivalently, one introduces an $N$-fold covering of the $\theta$ parameter space by coupling to a discrete topological sector, as in a $\mathbb{Z}_N$ $0$-form or $1$-form gauge theory. 
This procedure results in a modified vacuum energy function, which we denote $E_N(\theta)$. 
It is defined as the minimum of $E(\theta)$ over all $\mathbb{Z}_N$-related angles:
\begin{equation}\label{E_N_def}
E_N(\theta) = \min_{m\in\{0,1,\dots,N-1\}}\, E\!\Big(\theta+\frac{2\pi m}{N}\Big).
\end{equation}
By construction, $E_N(\theta)$ inherits $2\pi$-periodicity from $E(\theta)$ and is an even function as well, since the original $E(\theta)$ is even and the set of shifts $\{2\pi m/N\}$ is symmetric. 
Equation~\eqref{E_N_def} can be interpreted as follows. 
For any given $\theta$, one considers the $N$ distinct angles $\theta+\frac{2\pi m}{N}$ (modulo $2\pi$) that are related by the discrete gauge symmetry. 
One evaluates $E(\theta)$ on each of these angles and then takes the lowest value. 
This orbifold projection enforces that the theory relaxes to the vacuum branch whose $\theta$-value lies closest to $0$, the CP-conserving point. 

It is useful to define the principal value operation $\mathrm{PV}(\phi)$, which yields the unique representative of any angle $\phi$ in the interval $(-\pi,\pi]$. 
For each $m=0,1,\dots,N-1$, let us denote
\begin{equation}
x_m(\theta) \equiv \mathrm{PV}\!\Big(\theta + \frac{2\pi m}{N}\Big),
\end{equation}
so that each $x_m(\theta)\in(-\pi,\pi]$ and represents the angle $\theta+\frac{2\pi m}{N}$ brought to the principal $[-\pi,\pi]$ range. 
Because adding $2\pi/N$ successively generates a complete set of $N$ evenly spaced points on the circle, the $N$ values $x_0(\theta),x_1(\theta),\dots,x_{N-1}(\theta)$ partition the circle. 
They satisfy $x_{m+1}(\theta)=x_m(\theta)+2\pi/N$ (mod $2\pi$), differing pairwise by multiples of $2\pi/N$. 
Geometrically, if we divide the interval $(-\pi,\pi]$ into $N$ consecutive subintervals of length $2\pi/N$, for example $[-\pi/N,\pi/N]$, $[\pi/N,3\pi/N]$, up to $[(2N-1)\pi/N,\pi]$ and their negatives, exactly one of the $x_m(\theta)$ will fall into each such subinterval. 
In particular, there is exactly one representative $x_{m}(\theta)$ that lies in the central interval $[-\pi/N,\pi/N]$. 
We denote this special representative by
\begin{equation}
\theta_{\mathrm{eff}} \equiv \theta_{\mathrm{eff}}(\theta) = \text{the unique } x_m(\theta) \text{ with } x_m(\theta)\in[-\pi/N,\pi/N].
\end{equation}
By construction $|\theta_{\mathrm{eff}}|\le \pi/N$. 
We emphasize that $\theta_{\mathrm{eff}}(\theta)$ is simply $\theta$ reduced modulo $2\pi/N$ to the range $[-\pi/N,\pi/N]$. 
Intuitively, one shifts $\theta$ by an appropriate multiple of $2\pi/N$ to bring it as close to zero as possible on the circle. 

Now, since $E(\theta)$ is an even function and is monotonically increasing for $\theta>0$ on the interval $[0,\pi]$, owing to convexity and the minimum at $0$, the value of $E(\theta)$ grows as $|\theta|$ increases from $0$ up to $\pi$. 
In particular, if $\phi_1,\phi_2\in[0,\pi]$ with $\phi_1<\phi_2$, then $E(\phi_1)<E(\phi_2)$, and similarly $E(-\phi_1)<E(-\phi_2)$ by evenness. 
Among the $N$ candidates $x_m(\theta)$, the one with the smallest magnitude $|\theta_{\mathrm{eff}}|$ therefore yields the lowest energy. 
For any other $x_{m'}(\theta)$, we have $|x_{m'}(\theta)|\ge|\theta_{\mathrm{eff}}|$, hence
\begin{equation}
E(x_{m'}(\theta))\ge E(|x_{m'}(\theta)|)\ge E(|\theta_{\mathrm{eff}}|)=E(\theta_{\mathrm{eff}}).
\end{equation}
Here we used evenness to drop the sign inside $E$, and the fact that $E(\theta)$ is non-decreasing on $[0,\pi]$ to compare $E(|x_{m'}|)$ with $E(|\theta_{\mathrm{eff}}|)$. 
It follows that the minimum in Eq.~\eqref{E_N_def} is achieved precisely by $m=m_*(\theta)$ such that $x_{m_*}(\theta)=\theta_{\mathrm{eff}}$. 
In other words, the projected energy can be written in the simple form
\begin{equation}\label{E_projected_eff}
E_N(\theta) = E\!\big(\theta_{\mathrm{eff}}(\theta)\big), \qquad \text{with } |\theta_{\mathrm{eff}}(\theta)| \le \frac{\pi}{N}.
\end{equation}
This establishes that the effective $\theta$-angle in the $\mathbb{Z}_N$-gauged theory can always be chosen, by a gauge transformation in the topological sector, to lie in the interval $[-\pi/N,\pi/N]$. 
In particular, the absolute size of the physical vacuum angle is bounded by $\pi/N$:
\begin{equation}\label{theta_bound}
|\bar{\theta}| \equiv |\theta_{\mathrm{eff}}| \le \frac{\pi}{N}.
\end{equation}
This inequality is the projection bound. 
It holds for any input value of $\theta$ (the theory’s bare $\theta$-parameter) and for any $2\pi$-periodic, even, convex $E(\theta)$ with minimum at $0$. 
The bound is saturated in scenarios where $\theta$ lies exactly halfway between two adjacent multiples of $2\pi/N$. 
For example, if $\theta=(2\ell+1)\pi/N$ for some integer $\ell$, then $\theta_{\mathrm{eff}}=\pi/N$ (up to a sign) and $|\bar{\theta}|=\pi/N$. 
At those special points, two different choices of $m$ (namely $m=\ell$ and $m=\ell+1$) yield $\theta_{\mathrm{eff}}=+\pi/N$ and $\theta_{\mathrm{eff}}=-\pi/N$ respectively, which are equal in magnitude. 
Consequently $E_N(\theta)$ is continuous at $\theta=(2\ell+1)\pi/N$, since $E(\pi/N)=E(-\pi/N)$ by evenness, but it is not differentiable there. 
The system has two degenerate vacua, related by $\theta_{\mathrm{eff}}\to -\theta_{\mathrm{eff}}$, and the slope of $E_N(\theta)$ jumps, forming a cusp (a kink in the energy profile). 
These cusp points $|\theta|=(2\ell+1)\pi/N$ for $\ell=0,1,\dots$ are the only non-analytic features of $E_N$. 
Away from those points, a single branch $x_{m_*}(\theta)$ dominates and $E_N(\theta)$ exactly coincides with a smooth piece of the original function $E(\theta+2\pi m_*/N)$. 
Thus $E_N(\theta)$ is an even, $2\pi$-periodic function that is piecewise smooth and convex, with a sequence of cusps that reflect the discrete $\mathbb{Z}_N$ identification. 
Figure~\ref{fig:envelope1} illustrates this construction for a representative example.

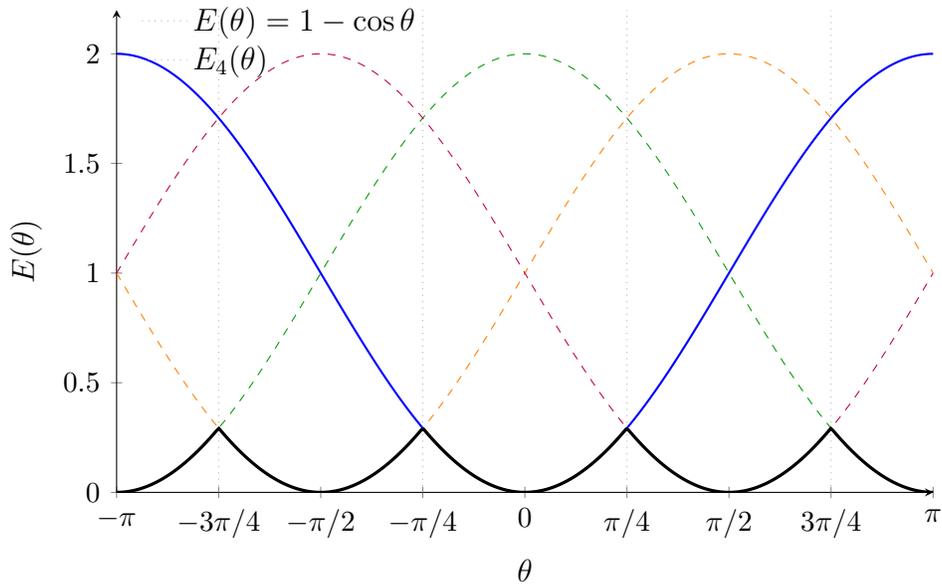
\begin{figure}[h]
\centering
\begin{tikzpicture}
  \begin{axis}[
    width=0.8\textwidth,
    height=8cm,
    xmin=-3.1416, xmax=3.1416,
    ymin=0, ymax=2.2,
    axis x line=bottom,
    axis y line=left,
    xlabel={$\theta$},
    ylabel={$E(\theta)$},
    xtick={-3.1416,-2.3562,-1.5708,-0.7854,0,0.7854,1.5708,2.3562,3.1416},
    xticklabels={$-\pi$,$-3\pi/4$,$-\pi/2$,$-\pi/4$,0,$\pi/4$,$\pi/2$,$3\pi/4$,$\pi$},
    legend cell align=left,
    legend style={at={(0.02,1.03)},anchor=north west,draw=none,fill=none},
    ticklabel style={font=\small},
  ]

    \addplot[gray!60,dotted] coordinates {(-0.7854,0) (-0.7854,2.2)};   
    \addplot[gray!60,dotted] coordinates {( 0.7854,0) ( 0.7854,2.2)};   
    \addplot[gray!60,dotted] coordinates {(-2.3562,0) (-2.3562,2.2)};   
    \addplot[gray!60,dotted] coordinates {( 2.3562,0) ( 2.3562,2.2)};   

    \addplot[blue,thick,domain=-pi:pi,samples=400]
      {1 - cos(deg(x))};
    \addlegendentry{$E(\theta)=1-\cos\theta$}

    \addplot[orange,dashed,domain=-pi:pi,samples=400]
      {1 - cos(deg(x + pi/2))};        
    \addplot[green!60!black,dashed,domain=-pi:pi,samples=400]
      {1 - cos(deg(x + pi))};          
    \addplot[purple,dashed,domain=-pi:pi,samples=400]
      {1 - cos(deg(x + 3*pi/2))};      

    \addplot[black,very thick,domain=-pi:pi,samples=400]
      {min(
         min(1 - cos(deg(x)),       1 - cos(deg(x + pi/2))),
         min(1 - cos(deg(x + pi)),  1 - cos(deg(x + 3*pi/2)))
       )};
    \addlegendentry{$E_{4}(\theta)$}

  \end{axis}
\end{tikzpicture}
\caption{Vacuum energy $E(\theta)$ (blue curve) for a $2\pi$-periodic, even, convex function (here $E(\theta)=1-\cos\theta$ for illustration) and its projected counterpart $E_{4}(\theta)$ after gauging a $\mathbb{Z}_4$ subgroup of the $\theta$-shift symmetry (black curve). The projected function $E_{4}(\theta)$ is obtained as the lower envelope of the $N$ shifted branches $E(\theta+2\pi m/N)$ (colored dashed curves, for $m=0,1,2,3$) that compose the multi-valued structure of $E(\theta)$. At each $\theta$, the physical energy $E_{4}(\theta)$ is the minimum of these branch values, which geometrically corresponds to dropping down to the lowest lying curve. For small $|\theta|$, the $m=0$ branch (blue) is clearly the lowest, so $E_{4}(\theta)=E(\theta)$ in the central region. However, beyond the first cusp at $\theta=\pi/4$, the branch corresponding to $\theta_{\mathrm{eff}}=\theta-2\pi/4$ (orange dashed curve) becomes the lowest. At $\theta=\pi/4$ exactly, one has two degenerate choices $\theta_{\mathrm{eff}}=\pm\pi/4$ giving the same energy, hence a non-smooth kink in $E_{4}(\theta)$. All such cusps occur at $|\theta|=(2\ell+1)\pi/4$ (vertical dotted lines), between which $E_{4}(\theta)$ follows a single analytic branch of the original function. The projection bound $|\bar{\theta}| \le \pi/N$ is visually evident: the effective vacuum angle $\bar{\theta}=\theta_{\mathrm{eff}}$ (the horizontal distance from the origin to the rightmost point of the black curve) is constrained to the fundamental domain $[-\pi/N,\pi/N]$.}
\label{fig:envelope1}
\end{figure}

Hence, gauging the $\mathbb{Z}_N$ shift symmetry forces the theory onto the vacuum branch closest to $\theta=0$. 
This ensures that the physical CP angle $\bar{\theta}$ is at most $\pi/N$ in magnitude. 
This result holds nonperturbatively and is protected by the discrete gauge symmetry, which forbids any continuous interpolation between branches. 
As $N$ increases, the bound~\eqref{theta_bound} becomes more restrictive. 
In fact, in the $N\to\infty$ limit one recovers $\bar{\theta}\to 0$ as the unique selection, solving the strong CP problem without requiring an axion. 
We stress that the projection mechanism is robust against radiative or gravitational corrections, since $E_N(\theta)$ remains the minimum of $E(\theta+2\pi m/N)$ to all orders. 
Any perturbation affecting $E(\theta)$ simply feeds into each branch equally and cannot upset the minimum selection rule. 
Thus, under the mild and well-motivated assumptions of periodicity, evenness, and convexity of $E(\theta)$, all supported by large-$N_c$ and lattice evidence and by general anomaly arguments, the discrete-$\theta$ projection yields a nonperturbative guarantee that $|\bar{\theta}| \le \pi/N$. 
This is the desired bound, which completes the proof~\cite{Witten:1998uka,Veneziano:1979ec,DiVecchia:1980yfw,Gaiotto:2014kfa,Cordova:2019uob,Hason:2020yqf}.

\section{Topological EFT details}
\label{app:teft}

We begin with the extended topological action that couples a new 1-form and 2-form gauge field to QCD:
\begin{equation}
S \;=\; \frac{N}{2\pi}\,\int B\wedge d a \;+\; \frac{1}{2\pi}\,\int a\wedge C_{\mathrm{QCD}}\,. 
\end{equation} 
Here $a$ is a compact $U(1)$ one-form gauge field and $B$ is a compact two-form gauge field (sometimes called a $B$-field or gerbe connection). 
The gauge transformations are $a\to a+d\lambda$ for $\lambda$ a $2\pi$-periodic $0$-form and $B\to B+d\Lambda$ for $\Lambda$ a $2\pi$-periodic $1$-form \cite{Gaiotto:2014kfa,HsinLam2021}. 
Varying $B$ enforces the equation of motion $\frac{N}{2\pi}d a=0$, which implies that $d a$ is a flat 2-form field strength. 
In a topologically nontrivial sector this means that $d a$ represents a cohomology class of finite order. 
Since $a$ is compact, its flux is quantized in integer units, so $\frac{1}{2\pi}\int_{\Sigma_2}d a\in\mathbb{Z}$ for any closed 2-surface $\Sigma_2$. 
Together with $N\,d a=0$ we infer that $\frac{1}{2\pi}\int_{\Sigma_2}d a$ is in fact an integer multiple of $\frac{1}{N}$. 
In other words, $[d a]$ lies in the torsion subgroup of $H^2(M,\mathbb{Z})$ of order $N$: the field strength $d a$ is not only closed but also $\mathbb{Z}_N$ quantized in such a way that $N\,[d a]=0$ in cohomology. 
Thus integrating out $B$ restricts the path integral to field configurations where $d a$ is a discrete $\mathbb{Z}_N$ flux. 
Equivalently, $d a$ is a 2-form $\mathbb{Z}_N$ gauge field. 
This achieves a discrete quantization of the $a$-field field strength and reduces the continuous $\theta$-parameter freedom to a discrete set of effective angles, as we will see in detail below \cite{Tanizaki:2018wtg}. 

The action $S$ is invariant under the local gauge symmetries of $a$ and $B$ up to boundary terms. 
Under $a\to a+d\lambda$, one finds $B\wedge d a\to B\wedge d a + B\wedge d^2\lambda = B\wedge d a$, since $d^2\lambda=0$, and $a\wedge C_{\mathrm{QCD}}\to a\wedge C_{\mathrm{QCD}} + d\lambda\wedge C_{\mathrm{QCD}}$. 
The latter variation is a total derivative because $d\lambda\wedge C_{\mathrm{QCD}} = d(\lambda\,C_{\mathrm{QCD}}) - \lambda\,dC_{\mathrm{QCD}}$, and $dC_{\mathrm{QCD}}$ is (by definition of $C_{\mathrm{QCD}}$) proportional to the topological charge density $\mathrm{Tr}(F\wedge F)$, which integrates to $0$ on a closed manifold or yields a boundary term on manifolds with boundary. 
Similarly, under the 1-form gauge transformation $B\to B+d\Lambda$, the variation $S\to S+\frac{N}{2\pi}\int d\Lambda\wedge d a$ equals $\frac{N}{2\pi}\int d(\Lambda\wedge d a)$, a surface term which vanishes for appropriate boundary conditions. 
These gauge invariances ensure that $a$ and $B$ enter only through $d a$ and $dB$, so that only the integer cohomology classes of their field strengths have physical significance, as expected for compact electrodynamic form fields. 

More subtle is the behavior under large gauge transformations of the QCD gauge field, which shift the Chern-Simons 3-form $C_{\mathrm{QCD}}$ by an exact piece. 
If $g$ is a gauge transformation in $SU(N)$ of winding number $1$, that is a transformation carrying one unit of topological charge, then the Chern-Simons term changes as $C_{\mathrm{QCD}}\to C_{\mathrm{QCD}}+d\Omega_2$, where $\Omega_2$ is a 2-form defined on coordinate patches that encodes the transition function for the Chern–Simons 3-form \cite{Alvarez-Gaume:1983ihn}. 
By construction $d\Omega_2$ represents the generator of $H^4(M,\mathbb{Z})$ corresponding to one unit of instanton number. 
More concretely, $d\Omega_2$ integrates to $2\pi$ on a 4-volume whose boundary is a 3-sphere on which $g$ is defined, reproducing the $2\pi\nu$ shift of the $\theta$-term for an instanton number $\nu=1$ configuration. 
Under this large gauge transformation $g$, the second term in $S$ varies as 
\begin{equation}
\Delta S \;=\; \frac{1}{2\pi}\int a\wedge d\Omega_2 \;=\; \frac{1}{2\pi}\int d(a\wedge \Omega_2)\;-\;\frac{1}{2\pi}\int d a\wedge \Omega_2\,.
\end{equation} 
The first term $\frac{1}{2\pi}\int d(a\wedge \Omega_2)$ is a boundary term, vanishing on closed spacetime or canceling between patches, while the second term gives the bulk contribution. 
Now, using the constraint from the $B$-field, we know that $d a$ is $\mathbb{Z}_N$ valued; in particular, $\frac{1}{2\pi}\int_{\Sigma_2}d a = k/N$ for some integer $k$ on any closed $\Sigma_2$. 
Taking $\Sigma_2$ to be a surface linking the worldvolume where $\Omega_2$ is supported, for example on a manifold $M=\mathbb{R}^4$ where $\Omega_2$ can be thought of as localized in the neighborhood of a three-sphere $S^3$ surrounding the instanton insertion and $\Sigma_2$ is any two-sphere linking that $S^3$, we have
\[
\int d a\wedge \Omega_2 = \left(\int_{\Sigma_2}d a\right)\left(\int_{S^2} \Omega_2\right) = \frac{2\pi k}{N}\cdot 1,
\]
since $\int_{S^2}\Omega_2=2\pi$ for a unit winding. 
Therefore $\Delta S = -\,\frac{1}{2\pi}\int d a\wedge \Omega_2 = -\,\frac{2\pi k}{N}$. 
The path integral thus acquires a phase $e^{i\Delta S} = e^{-2\pi i k/N}$. 
For a single instanton, that is for $k=1$ in the minimal sector where $\frac{1}{2\pi}\int d a = 1/N$, this phase is $e^{-2\pi i/N}$, which is a nontrivial $N$th root of unity. 
This indicates that under the large $SU(N)$ gauge transformation $g$, the functional integrand is not strictly invariant but picks up a phase. 
However, this phase precisely reflects the fact that we have promoted the shift $\theta\to\theta+\frac{2\pi}{N}$ to a gauge symmetry. 
In the original theory without the $a,B$ fields, a $2\pi$ shift of the $\theta$ angle is an invariance because the instanton number $Q\in\mathbb{Z}$. 
In an $SU(N)/\mathbb{Z}_N$ theory where $Q$ can be fractional (in units of $1/N$), a $2\pi$ shift is no longer invariant and only a $2\pi N$ shift would be. 
This is exactly why we have introduced the discrete gauge symmetry in the first place. 
The phase $e^{-2\pi i/N}$ is the hallmark of this discrete gauge symmetry. 
It says that performing a large $SU(N)$ gauge transformation (instantonic transition) is equivalent, in the enlarged theory, to performing a gauge transformation in the $0$-form $\mathbb{Z}_N$ symmetry (shifting $\theta$ by $2\pi/N$) that compensates for this phase. 

In the path integral one must sum over all gauge-equivalent field configurations. 
Thus one must include not only the original sector but also the sector after an instanton accompanied by the $2\pi/N$ shift, and so on. 
Summing over these $\mathbb{Z}_N$ related sectors, an operation often described as an orbifold sum over $\theta$-angles, yields 
\begin{equation}
Z_{\mathrm{orbifold}} \;=\; \frac{1}{N}\sum_{m=0}^{N-1} Z\!\Big(\theta + \frac{2\pi m}{N}\Big)\,,
\end{equation} 
so that the only contributions which survive are those invariant under $\theta\to\theta+\frac{2\pi}{N}$. 
Equivalently, writing the original partition function as $Z(\theta)=\sum_Q e^{i\theta Q}\,Z_Q$ where $Q\in\mathbb{Z}$ labels the total topological charge sector, the orbifold-projected partition function is 
\begin{equation}
Z_{\mathrm{orbifold}} \;=\; \frac{1}{N}\sum_{m=0}^{N-1}\sum_{Q\in\mathbb{Z}} e^{i(\theta+\frac{2\pi m}{N})Q}\,Z_Q \;=\; \sum_{Q\in\mathbb{Z}} e^{i\theta Q}\left(\frac{1}{N}\sum_{m=0}^{N-1}e^{2\pi i m Q/N}\right)Z_Q\,.
\end{equation} 
The inner sum $\frac{1}{N}\sum_{m=0}^{N-1}e^{2\pi i m Q/N}$ vanishes for any $Q$ not divisible by $N$, and equals $1$ if $Q=N\ell$ for some $\ell\in\mathbb{Z}$. 
Therefore
\[
Z_{\mathrm{orbifold}} = \sum_{\ell\in\mathbb{Z}} e^{i\theta (N\ell)}\,Z_{Q=N\ell}.
\]
In particular, all sectors with $Q\notin N\mathbb{Z}$ are projected out, establishing the selection rule $Q\in N\mathbb{Z}$. 
Physically, this means that any gauge-field configuration contributing to the path integral must have instanton number a multiple of $N$. 
This is an exact enforcement of “$\theta$ can only be $0$ mod $2\pi/N$” at the nonperturbative level. 
The surviving $\theta$ dependence is only through the combination $N\theta$. 
Indeed, under $\theta\to\theta+\frac{2\pi}{N}$ the theory is now fully gauge equivalent, so the fundamental domain of the effective vacuum angle is $\theta_{\mathrm{eff}}\in[0,\frac{2\pi}{N})$. 
If CP is also imposed, mapping $\theta\to-\theta$, one may further restrict to $|\theta_{\mathrm{eff}}|\le \pi/N$. 
Thus the strong CP angle is rendered effectively $N$ times smaller. 
The theory has been discretely gauged such that $\theta$ and $\theta+\frac{2\pi}{N}$ describe the same vacuum. 
Consequently any $\theta$ within this range, in particular $\theta$ sufficiently close to $0$, is physically indistinguishable from $0$. 
This is the mechanism by which the discrete $\theta$-projection solves the strong CP problem without axions: the dangerous $\theta$ parameter is quantized to a small discrete subgroup and can naturally relax to zero (mod $2\pi/N$), eliminating CP-violating effects beyond the $\mathcal{O}(\theta^{2})$ level expected for $|\theta_{\mathrm{eff}}|\le \pi/N$. 

We now connect this BF-type construction to a more group-theoretic description in terms of the gauge structure $SU(N)/\mathbb{Z}_N$ and higher-form symmetries. 
Gauging the $\mathbb{Z}_N$ 1-form center symmetry of an $SU(N)$ gauge theory is known to produce an $SU(N)/\mathbb{Z}_N$ theory. 
In that description, the $SU(N)$ principal bundle is modified by a $\mathbb{Z}_N$ valued second Stiefel–Whitney class $w_2$ that obstructs lifting to an $SU(N)$ bundle \cite{Gaiotto:2014kfa}. 
In the continuum, this procedure introduces a 2-form gauge field $B$ for the center and a coupling $\frac{N}{2\pi}\int B\wedge \mathrm{Tr}(F)$, which is essentially the first term of $S$ above (here $\mathrm{Tr}(F)$ is the generator of $H^2$ for the center symmetry current) \cite{Gaiotto:2014kfa,Hsin:2020nts}. 
However, in a pure Yang-Mills theory one may also consider a $\mathbb{Z}_N$ discrete $\theta$ angle: a coupling of the form $i\,p\int w_2\smile w_2$ in the action (with $p\in\{0,1,\dots,N-1\}$) which weights gauge field configurations according to the $\mathbb{Z}_N$ characteristic class $w_2$. 
Equivalently, one can use a 4d topological term $\frac{2\pi p}{N}\int B\wedge B$ when one uses a 2-form $B$ \cite{Tanizaki:2018wtg,Hsin:2020nts}. 
Different choices of this discrete $\theta$ parameter $p$ define physically distinct $SU(N)/\mathbb{Z}_N$ theories that are not continuously connected. 

In our construction, the second term $\frac{1}{2\pi}\int a\wedge C_{\mathrm{QCD}}$ together with the BF coupling can be understood as a concrete gauge-theoretic realization of such a discrete $\theta$ term. 
Integrating out $a$ yields $dB=-\frac{1}{N}C_{\mathrm{QCD}}$, or equivalently $N\,dB + C_{\mathrm{QCD}}=0$. 
Taking an exterior derivative, this condition becomes $N\,d^2B + dC_{\mathrm{QCD}}=N\,d^2B + \mathrm{Tr}(F\wedge F)=0$. 
This is precisely the modified Bianchi identity that enforces that the fractional instanton number carried by an $SU(N)/\mathbb{Z}_N$ bundle, given by $\frac{1}{8\pi^2}\int \mathrm{Tr}(F\wedge F)$, is correlated with the $\mathbb{Z}_N$ 2-form flux $dB$ in such a way that the total topological charge is an integer. 
In fact, one can show that $e^{iS}$ is exactly the partition function of an $SU(N)/\mathbb{Z}_N$ theory with a particular discrete $\theta$ choice, specifically the one that ensures a $\mathbb{Z}_N$ identification of $\theta$ vacua \cite{HsinLam2021}. 

This equivalence is encapsulated in the language of higher-form symmetries and their extensions. 
The theory exhibits an interplay between a 0-form $\mathbb{Z}_N$ symmetry, the discrete shift of $\theta$, and a 1-form $\mathbb{Z}_N$ symmetry, the center symmetry. 
The BF coupling introduces a 2-group structure that unifies these into a single algebraic framework \cite{Cordova:2019uob,Hsin:2020nts}. 
The Postnikov class that classifies this 2-group extension is an element of $H^3(B\mathbb{Z}_N,\mathbb{Z}_N)\cong \mathbb{Z}_N$, which can be identified with the integer $N$ (mod $N$) in our construction. 
Concretely, this means that performing a 0-form transformation (shifting $\theta$ by $2\pi/N$) followed by a 1-form transformation (acting on a Wilson loop by a center element) is not exactly the same as doing them in the opposite order. 
Instead, the difference is a nontrivial transformation in the overlap, quantified by the cohomology class in $H^3(\mathbb{Z}_N,\mathbb{Z}_N)$. 
This is the hallmark of a two-group: the 0-form and 1-form symmetries are fused by an extension, here protected by the discrete topological term. 

In physical terms, the mixed anomaly that would otherwise relate $\theta$-dependence and center symmetry, for example the well-known obstruction to having both CP symmetry and unbroken center symmetry at $\theta=\pi$ when $N$ is even, as discussed by Gaiotto and collaborators and others \cite{Gaiotto:2014kfa,Cordova:2019uob}, is canceled by the inflow from the bulk five-dimensional topological action that one can associate to $S$. 
Our BF term can be seen as deriving from a five dimensional protected topological phase, an invertible TQFT, with action $\frac{N}{2\pi}\int B\wedge d a$ whose boundary is the four dimensional gauge theory, analogous to how a Green-Schwarz mechanism cancels continuous anomalies via a five dimensional Chern-Simons term \cite{Alvarez-Gaume:1983ihn}. 
The quantization of the Postnikov class ensures that all gauge transformation phases $e^{-2\pi i/N}$ are accounted for and made harmless by the enlarged gauge symmetry. 

The integration of the $B$-field forces $d a$ to be an $N$ torsion flux, which in turn guarantees that large $SU(N)$ gauge transformations shift $\theta$ by $2\pi/N$ and produce a path-integral phase that is an $N$th root of unity. 
By interpreting this phase as a gauge redundancy and summing over it, that is performing the orbifold projection, the theory allows only configurations with total instanton number $Q$ divisible by $N$ and identifies all $\theta$ angles differing by $2\pi/N$. 
This implements the discrete $\theta$-gauging condition: effectively, $\theta$ becomes an angular variable with period $2\pi/N$. 
The physical $\theta_{\mathrm{eff}}$ is confined to $[-\pi/N,\pi/N]$, and CP is automatically preserved when $\theta_{\mathrm{eff}}=0$, since $\theta_{\mathrm{eff}}$ can only vanish modulo $2\pi$ in the gauged theory. 
The entire construction is dynamically neutral, because it involves no propagating light fields and only topological degrees of freedom. 
It is equivalent to formulating QCD as an $SU(N)/\mathbb{Z}_N$ gauge theory with a suitable discrete topological term. 
The bound $|\theta_{\mathrm{eff}}|\le \pi/N$ arises naturally as the fundamental domain of the orbifolded $\theta$. 
Thus a small effective $\theta$, solving the strong CP problem, is enforced by the theory’s topological consistency and gauge structure rather than by fine-tuning.

\section{Lattice implementation strategy}
\label{app:lattice}

We discretize the four dimensional Euclidean spacetime on a hypercubic lattice and introduce two new gauge fields. 
The first is an Abelian one form $a_\ell$ on each oriented link $\ell$ (representing a $U(1)$ gauge connection). 
The second is a discrete two form $B_p$ on each plaquette $p$ (representing a $\mathbb{Z}_N$ “BF’’ flux variable). 
In continuum language, $a$ and $B$ are lattice 1- and 2-cochains, respectively. 
The variable $a_\ell \in [0,2\pi)$ is a compact $U(1)$ phase on link $\ell$, and $B_p \in \{0,1,\dots,N-1\}$ is an integer valued $\mathbb{Z}_N$ flux on plaquette $p$. 

Their gauge transformations are defined as follows. 
For any 0-form gauge function $\Lambda_x \in \mathbb{R}/2\pi\mathbb{Z}$ on sites $x$, the 1-form $a$ transforms as 
\begin{equation}
a_{\ell: x\to y} \to a_{\ell} + \Lambda_y - \Lambda_x,
\end{equation}
which is analogous to a usual $U(1)$ lattice gauge field. 
For any 1-form gauge parameter $\Xi_\ell \in \{0,\dots,N-1\}$ (an integer on each link), the 2-form $B$ transforms as 
\begin{equation}
B_p \to B_p + \sum_{\ell\in\partial p}\Xi_\ell \quad (\text{mod } N),
\end{equation}
that is, the $\mathbb{Z}_N$ plaquette variable shifts by the integer sum of $\Xi$ on its bounding links. 
By construction, these gauge symmetries ensure that $a$ and $B$ only appear in the action through invariant combinations (essentially the lattice field strength and its dual). 
In particular, the BF coupling introduced below will be exactly invariant under $a\to a+d\Lambda$ and $B\to B+\delta\Xi$, where $d$ and $\delta$ are the lattice exterior derivative and coboundary operators, acting as $\partial$ on chains. 

We now formulate the lattice BF action that couples $a$ and $B$. 
Following the continuum prescription for a $B_2\wedge f_2$ term $\frac{iN}{2\pi}\int B\wedge d a$ (with $N$ a normalization chosen for later convenience), we define the discrete action as a sum over three cells (elementary cubes) $c$ of the lattice:
\begin{equation}
S_{BF} \;=\; \frac{2\pi i}{N}\;\sum_{c}\;B(\partial c)~,
\label{eq:C1}
\end{equation}
where $B(\partial c)\equiv \sum_{p\in \partial c} B_p$ is the total $\mathbb{Z}_N$ flux through the faces of the three cell $c$. 
In algebraic topology terms, if $B$ is viewed as a 2-cochain on the lattice, then $B(\partial c)=\langle B,\partial c\rangle$ is the natural pairing of $B$ with the boundary of the three chain $c$. 
Equivalently, using the lattice coboundary operator $\delta$ (dual to $\partial$), one can write
\begin{equation}
S_{BF}=\frac{2\pi i}{N}\langle B,\delta a\rangle,
\end{equation}
since $\delta a$ is a 2-cocycle on plaquettes that captures the lattice field strength of $a$. 
This $S_{BF}$ is the exact lattice analog of the continuum topological action $S_{\text{BF}}=\frac{iN}{2\pi}\int B\wedge d a$. 
Expanding $B(\partial c)$ in local coordinates, one finds 
\begin{equation}
B(\partial c)\approx \frac{1}{4}\epsilon^{\mu\nu\rho\sigma}B_{\mu\nu}(\partial_\rho a_\sigma)
\end{equation}
on each cell (with $\epsilon^{\mu\nu\rho\sigma}$ the Levi-Civita symbol), so that 
\begin{equation}
\sum_c B(\partial c)\to \int d^4x\,B_{\mu\nu}\partial_{\rho}a_{\sigma}\,\epsilon^{\mu\nu\rho\sigma}=\int B\wedge d a
\end{equation}
in the continuum limit. 
The factor $2\pi/N$ in \eqref{eq:C1} is chosen such that $B_p$ being an integer reproduces the standard $2\pi$ periodicity of the $B$-field in the continuum action. 
Thus, $S_{BF}$ provides a gauge invariant and dimensionless lattice implementation of the $B_2\wedge f_2$ coupling. 
By inspection, $S_{BF}$ is invariant under the 1-form $a$ gauge symmetry. 
Any shift $a\to a+d\Lambda$ contributes $\delta a\to\delta a$ and hence $B(\partial c)$ is unchanged, since $\sum_{p\in\partial c}(\delta\Lambda)_p=0$ for a closed surface. 
The action is also invariant under the 2-form $B$ gauge symmetry. 
Any shift $B\to B+\delta\Xi$ yields $B(\partial c)\to B(\partial c)+(\delta^2\Xi)(c)$, and $\delta^2=0$ on any three chain so the action is unchanged. 
These invariances confirm that our lattice BF term correctly gauges a discrete $\mathbb{Z}_N$ 1-form symmetry, as expected from continuum analyses \cite{Gaiotto:2014kfa,Kapustin:2013uxa}. 

The BF term \eqref{eq:C1} acts as a Lagrange multiplier that enforces a quantization condition on the $U(1)$ field $a$. 
Because $B_p$ enters $S_{BF}$ linearly, summing over the $\mathbb{Z}_N$ variable $B_p$ imposes a delta function constraint. 
In fact, for each three cell $c$ one has 
\begin{equation}
\frac{1}{N}\sum_{B(\partial c)=0}^{N-1} e^{\frac{2\pi i}{N}B(\partial c)\,k} \;=\; \delta_{k,\,0\;(\text{mod }N)}~,
\end{equation}
for any integer $k$. 
In the path integral, this means that only field configurations with $B(\partial c)\equiv 0\pmod N$ for every $c$ will contribute nonzero weight. 
But $B(\partial c)=\langle B,\delta a\rangle_c = \langle \delta B, a\rangle_c$ by lattice duality, which is analogous to integration by parts, using $\langle B,\delta a\rangle=\langle \delta B,a\rangle$ on a periodic lattice. 
The equation of motion obtained by varying $a_\ell$ in $S_{BF}$ is therefore $\delta B=0$ (mod $N$) on every link. 
This means the $B$-field is a closed 2-form: $\delta B=0$ is the lattice Bianchi identity for $B$. 
It implies that the $\mathbb{Z}_N$ flux through any closed surface vanishes mod $N$, so $B$ can be regarded as the field strength of an emergent $\mathbb{Z}_N$ 1-form gauge field (sometimes called a 2-form gauge field in continuum parlance \cite{Gaiotto:2014kfa}). 
Dually, the equation of motion obtained by summing over $B_p$ (or varying $B_p$ viewed as an integer Lagrange multiplier) enforces $\delta a \equiv 0\pmod{2\pi/N}$ on each plaquette. 
In other words, the $U(1)$ holonomy around every plaquette is an integer multiple of $2\pi/N$. 
The $a$-field flux is quantized in units of $2\pi/N$. 
We can write this constraint as
\begin{equation}
\delta a \;=\; \frac{2\pi}{N}\,m_{p}~, \qquad m_{p}\in\mathbb{Z}
\label{eq:C2}
\end{equation}
for each plaquette $p$. 
Thus $d a = 0$ in the continuum sense, except that configurations carrying fractional $U(1)$ flux (fractions of the $2\pi$ quantum) are permitted when $m_{p}\neq 0$. 
This is the hallmark of gauging a $\mathbb{Z}_N$ subgroup of a $U(1)$ gauge theory. 
The field strength is no longer unconstrained, as in a pure $U(1)$ theory, but must take values in the subgroup $\{0,\frac{2\pi}{N},\frac{4\pi}{N},\dots\}$, implementing a discrete flux quantization. 
In our context, since $a$ will couple to the QCD topological charge, the condition \eqref{eq:C2} ensures that only bundles with total flux divisible by $N$ are summed over, a mechanism essential to the $\theta$-angle projection. 
We emphasize that $S_{BF}$ by itself is a discrete topological action, a $\mathbb{Z}_N$ two form gauge theory, also known as a $\mathbb{Z}_N$ SPT phase in four dimensions, with no dependence on the metric or local dynamics \cite{Kapustin:2013uxa}. 
Its effect is to enforce a $\mathbb{Z}_N$ flux constraint on $a$, which can be viewed as a coupling of $a$ to a background $\mathbb{Z}_N$ 1-form symmetry current \cite{Gaiotto:2014kfa}. 
This $\mathbb{Z}_N$ gauged sector is what will enforce the discrete $\theta$-periodicity on the QCD vacuum. 

Next, we incorporate the QCD gauge fields and the $\theta$ term into this lattice framework. 
Let $U_\ell \in \mathrm{SU}(N_c)$ be the usual lattice gauge link variables for color gauge theory (we suppress color indices), and let $S_{\text{YM}}[U]$ be the Wilson action (or any suitable lattice action) for the Yang–Mills field. 
The topological term in the continuum is $S_\theta = i\,\theta\,Q = i\,\theta \int d^4x\,q(x)$, where $q(x)=\frac{1}{16\pi^2}\Tr(F_{\mu\nu}\widetilde{F}^{\mu\nu})$ is the topological charge density and $Q=\int q(x)d^4x\in \mathbb{Z}$ is the instanton number. 
On the lattice, an exactly gauge-invariant local definition of $q(x)$ is challenging, but $Q$ can be defined globally (for example by cooling or through the index of the Dirac operator) such that in the continuum limit it approaches the continuum $Q$ (for reviews, see \cite{Luscher:1981zq,Creutz:1983njd}). 
We do not need the explicit form of $q(x)$ on the lattice. 
Instead, we introduce a Chern–Simons three form $K(c)$ on each elementary three cell $c$, which serves as a local representative (potential) for the topological charge. 
In continuum terms, $K_{\text{cont}}(x)$ is a three form with $dK = \Tr(F\wedge F)/8\pi^2 = q(x)\,d^4x$. 
On the lattice, we require that the coboundary of the three cochain $K(c)$ yields the topological charge on the surrounding four cell(s). 
Concretely, for each elementary four dimensional hypercube (four cell) $C$ in the lattice, we demand
\begin{equation}
\delta K (C) \;=\; \sum_{c\,\subset\,\partial C} K(c) \;=\; q(C)~,
\label{eq:C3}
\end{equation}
where $q(C)$ is an integer-valued lattice topological charge assigned to that four cell, such that $\sum_C q(C)=Q$ over the whole lattice. 
The existence of such an integer $q(C)$ and of $K(c)$ is guaranteed for smooth gauge fields in the continuum limit by the fact that $\Tr(F\wedge F)$ is exact (the instanton number can be seen as the difference of Chern-Simons invariants of the three dimensional boundaries). 
In practice one can construct $K(c)$ from the plaquettes and links of the cube $c$, for example by a lattice transcription of the continuum Chern-Simons current
\begin{equation}
K_\mu=\frac{1}{8\pi^2}\epsilon_{\mu\nu\rho\sigma}\Tr\!\left[A_\nu F_{\rho\sigma} - \frac{i}{3}A_\nu A_\rho A_\sigma\right]
\end{equation}
in an axial gauge \cite{Luscher:1981zq}. 
We assume $K(c)$ has been chosen such that $\delta K(C)=q(C)$ exactly. 
This can always be done up to trivial gauge transformations on the boundaries of four cells, which will be accounted for by $a$-field gauge freedom shortly. 

With these ingredients, we can now couple the $U(1)$ field $a_\ell$ to the QCD topological density in a gauge invariant way. 
We add to the action a term
\begin{equation}
S_{\text{top}} \;=\; i \sum_{C} a(\partial C)\;K(C)~,
\label{eq:C4}
\end{equation}
where the sum is over all four cells $C$ and $a(\partial C)$ denotes the circulation of the 1-form $a$ around the edges of the four cell $C$. 
More precisely, if $\partial C$ is the closed three dimensional boundary of the hypercube $C$, consisting of three cells $c$, and if those three cells have oriented link loops $\ell\in \partial c$, then $a(\partial C)\equiv \sum_{\ell\in \partial C} a_\ell$ is the total $a$-flux along the boundary. 
This is well defined since $\partial(\partial C)=0$. 
The form \eqref{eq:C4} is a discretized version of the coupling $\int a\wedge \frac{\Tr(F\wedge F)}{8\pi^2}$ in the continuum. 
To see this, note that on each four cell $C$, $\sum_{\ell\in \partial C} a_\ell$ is analogous to $\int_{\partial C} a = \int_C d a$, and $\sum_{C}K(C)$ with $\delta K = q$ is analogous to $\int \Tr(F\wedge F)$. 
Thus $S_{\text{top}} = i\int a\wedge (\Tr(F\wedge F)/8\pi^2)$ in the continuum limit. 
Importantly, $S_{\text{top}}$ is invariant under all gauge symmetries, provided we accompany the usual $\theta$-shift with the appropriate transformation of $a$. 

We now verify the gauge invariances in turn. 
Under an $\mathrm{SU}(N_c)$ gauge transformation, $K(c)$ shifts by a coboundary while $q(C)$ remains invariant, and large gauge transformations shift $Q$ by integers. 
This shift is compensated by large $\mathbb{Z}_N$ transformations of $a$ of the form $a_\ell\to a_\ell+2\pi/N$. 
The shift of $S_{\text{top}}$ is $\Delta S_{\text{top}}=i(2\pi/N)Q$, which cancels a $\theta$ shift $\Delta S_\theta=-i(2\pi/N)Q$ for $\Delta\theta=2\pi/N$. 
Hence, the transformation $\theta\to\theta+2\pi/N$, $a_\ell\to a_\ell+2\pi/N$ leaves the full action invariant. 
This effectively gauges the discrete $\theta$-shift symmetry and reduces the physical periodicity to $2\pi/N$. 
Similarly, small 0-form gauge transformations $a\to a+d\Lambda$ shift $S_{\text{top}}$ by boundary terms that are integer multiples of $2\pi i$, leaving the Boltzmann weight unchanged. 
The term $S_{\text{top}}$ is trivially invariant under $B\to B+\delta\Xi$, since $S_{\text{top}}$ does not depend on $B$. 
Therefore, the total lattice action 
\begin{equation}
S_{\text{YM}}+S_{\theta}+S_{BF}+S_{\text{top}}
\end{equation}
preserves $\mathrm{SU}(N_c)$, 1-form, and 2-form gauge invariance exactly. 

The full partition function reads
\begin{multline}
Z(\theta) = \frac{1}{N^{N_p}}\sum_{\{B_p\}\in (\mathbb{Z}_N)^{N_p}}\;\int [DU][Da]\;\exp\Big\{-S_{\mathrm{YM}}[U] - i\,\theta \sum_{C} q(C)\;+\;\\ S_{BF}[B,a] \;+\; S_{\text{top}}[a,K(U)]\Big\}~,
\label{eq:C5}
\end{multline}
where $N_p$ is the number of plaquettes. 
Performing the sums over $B_p$ and the integrals over $a_\ell$ enforces the constraints $\delta a = \frac{2\pi}{N}m_p$ and $\delta B=0$, ensuring that only values of $Q$ divisible by $N$ contribute. 
This leads to
\begin{equation}
Z(\theta) \;=\; \frac{1}{N}\sum_{m=0}^{N-1} \; Z_{\text{QCD}}(\theta + 2\pi m/N)~,
\label{eq:C6}
\end{equation}
where $Z_{\text{QCD}}(\theta)$ is the partition function of ordinary Yang–Mills theory at vacuum angle $\theta$. 
Equation \eqref{eq:C6} expresses the orbifold sum that projects the theory onto $\theta$-sectors modulo $2\pi/N$. 
The gauging of the discrete $\theta$-shift symmetry thus manifests as a sum over $N$ shifted images of $\theta$, removing redundant branches of the vacuum. 

From \eqref{eq:C6}, the vacuum energy density follows as
\begin{equation}
E(\theta) \;=\; \min_{m\in\mathbb{Z}}\;E_{\text{QCD}}(\theta + 2\pi m/N)~,
\label{eq:C7}
\end{equation}
which represents the multi-branched energy envelope. 
For $|\theta|\le \pi/N$, $E(\theta)=E_{\text{QCD}}(\theta)$, and for larger $\theta$ the theory transitions between adjacent branches at $\theta=(2\ell+1)\pi/N$. 
Expanding $E(\theta)$ for small $\theta$ gives
\begin{equation}
E(\theta) \;\approx\; E(0)\;+\;\frac{1}{2}\,\chi_N\,\theta^2 \;+\; \mathcal{O}(\theta^4)~, \qquad 
\chi_N \;=\; \frac{\chi_t}{N^2}~,
\label{eq:C8}
\end{equation}
where $\chi_t=\partial^2 E_{\text{QCD}}/\partial\theta^2|_{\theta=0}$ is the conventional Yang–Mills topological susceptibility. 
The factor $1/N^2$ reflects the suppression of topological fluctuations due to the projection. 
The $\mathbb{Z}_N$ gauging requires $Q$ to appear in multiples of $N$, exponentially suppressing instantons of charge below $N$ and flattening the $\theta$ potential near $\theta=0$. 
Consequently, CP violation observables such as the neutron EDM are suppressed by $\sim 1/N$, yielding an effective dynamical relaxation of $\bar\theta$ toward zero. 

The lattice theory defined by $S_{\mathrm{YM}}[U] + S_{BF}[B,a] + S_{\text{top}}[a,K]$ provides a concrete implementation of discrete $\theta$-projection without axions. 
By gauging the $\theta$-shift symmetry $\mathbb{Z}_N$, we explicitly derived how the path integral sums over $\theta$-images to enforce $Q\in N\mathbb{Z}$ and to select the CP-conserving branch of the vacuum ($m=0$) nonperturbatively. 
The correspondence between the lattice action \eqref{eq:C1} and \eqref{eq:C4} and its continuum counterpart $i\frac{N}{2\pi}\int B\wedge d a + i\int a\wedge\frac{F\tilde F}{8\pi^2}$ is manifest in our construction, including the normalization of $B$ and $a$ to yield the standard $\theta$ periodicity. 
All gauge invariances, including the new 1-form and 2-form symmetries, are preserved exactly on the lattice, protecting the discrete topological term from any local quantum corrections. 
The extraction of $E(\theta)$ from the partition function via $E(\theta) = -\frac{1}{V}\ln Z(\theta)$ reproduces the expected orbifold structure and the $1/N^2$ suppression of $\chi_N$. 
These results demonstrate that coupling QCD to a $\mathbb{Z}_N$ two form gauge field via a BF term and a discrete topological interaction nonperturbatively projects the theory onto the $\theta$-vacuum nearest zero, achieving an exponentially small $\bar\theta$ without introducing an axion.

\section{Discrete clockwork for large $N$}

We now give a complete derivation of the clockwork enhancement of the discrete $\theta$ gauging advertised in the main text. 
We start from the multi-sector topological action
\begin{equation}
\label{eq:D1}
S \;=\; \sum_{i=1}^{K}\frac{N_i}{2\pi}\int_{M_4} B_i\wedge \mathrm{d}a_i \;+\; \frac{1}{2\pi}\int_{M_4}\big(a_1+q_2 a_2+\cdots+q_K a_K\big)\wedge C_{\rm QCD},
\end{equation}
where $a_i$ are compact $U(1)$ one-form gauge fields, $B_i$ are compact two-form gauge fields, $C_{\rm QCD}$ is the Chern–Simons three form of QCD with $\mathrm{d}C_{\rm QCD}= {\rm Tr}\,F\wedge F$, and $N_i,q_i\in\mathbb{Z}_{>0}$. 
Throughout, integrals are over a closed, oriented four-manifold $M_4$ and wedge products are implicit. 
The first set of terms implements $K$ independent $\mathbb{Z}_{N_i}$ one-form gauge symmetries (discrete BF theories), while the second term couples the diagonal linear combination
\begin{equation}
a_{\rm eff}\;\equiv\;a_1+q_2 a_2+\cdots+q_K a_K
\end{equation}
to QCD. 
This is the discrete, higher-form analogue of clockwork couplings familiar from axion and $p$-form constructions \cite{KaloperSorbo2009}, now embedded in the generalized symmetry framework \cite{Gaiotto:2014kfa,HsinLam2021,Cordova:2019uob}. 

Varying \eqref{eq:D1} we obtain, for each sector $i$,
\begin{equation}
\label{eq:EOMs}
\begin{aligned}
\delta_{B_i}S: N_i\,\mathrm{d}a_i=0,\\
\delta_{a_i}S: N_i\,\mathrm{d}B_i+\alpha_i\,C_{\rm QCD}=0,\\
\alpha_1=1,\;\alpha_i=q_i\;(i\ge2).
\end{aligned}
\end{equation}
Thus $a_i$ are flat $U(1)$ connections whose field strengths $f_i\equiv \mathrm{d}a_i$ represent $\mathbb{Z}_{N_i}$ valued torsion classes. 
On any closed two cycle $\Sigma_2$ one has
\begin{equation}
\label{eq:fluxquant}
\oint_{\Sigma_2} f_i \;=\; \frac{2\pi}{N_i}\,\ell_i,\qquad \ell_i\in\mathbb{Z}\;\;(\mathrm{mod}\,N_i).
\end{equation}
The second equation in \eqref{eq:EOMs} is the local constraint that enforces that $C_{\rm QCD}$ is trivial in cohomology modulo all $N_i$, in a manner that we make precise below. 


Consider a large $\mathrm{SU}(3)_c$ gauge transformation of unit winding. 
Under such a transformation $C_{\rm QCD}\to C_{\rm QCD}+\mathrm{d}\Omega_2$ with $\int_{S^2}\mathrm{d}\Omega_2=2\pi$ on any $S^2$ linking the instanton insertion. 
Using \eqref{eq:fluxquant} and integrating by parts, the topological action \eqref{eq:D1} shifts by
\begin{equation}
\Delta S \;=\; \frac{1}{2\pi}\sum_{i=1}^K \alpha_i\int a_i\wedge \mathrm{d}\Omega_2
\;=\; -\,\sum_{i=1}^K \alpha_i \int f_i \frac{\Omega_2}{2\pi}
\;=\; -\,2\pi\sum_{i=1}^K \frac{\alpha_i\,\ell_i}{N_i},
\end{equation}
so the weight of a configuration in the instanton sector of topological charge $Q\in\mathbb{Z}$ acquires the phase
\begin{equation}
\exp\!\left[i Q\,\Delta S\right]\;=\;\exp\!\left(-\,i\,2\pi Q \sum_{i=1}^K \frac{\alpha_i\,\ell_i}{N_i}\right).
\end{equation}
The path integral sums over all discrete flux sectors $\ell_i\in\mathbb{Z}_{N_i}$. 
Using the elementary character identity
\begin{equation}
\frac1{N}\sum_{\ell=0}^{N-1}e^{-2\pi i r\,\ell/N}=1
\end{equation}
if and only if $N\mid r$ and $0$ otherwise, the multiple sum factorizes and enforces, for each $i$,
\begin{equation}
N_i\mid \alpha_i\,Q.
\end{equation}
At this stage the standard clockwork choice
\begin{equation}
\label{eq:qchoice}
q_i\;=\;\prod_{j=1}^{i-1}N_j\qquad (i\ge 2),
\end{equation}
together with the mild arithmetic hypothesis that the $N_i$ are pairwise coprime,\footnote{If the $N_i$ share common prime factors, the argument yields the order $N_{\rm eff}=\mathrm{lcm}(N_1,\dots,N_K)$. The coprime case, which we adopt henceforth for definiteness, maximizes the enhancement and gives the product. One can also engineer the product without pairwise coprimality by slightly enlarging the chain with auxiliary prime factors absorbed into the $q_i$; the Smith normal form analysis then again returns the product, see below.} implies $\gcd(\alpha_i,N_i)=1$ for all $i$. 
Hence $N_i\mid Q$ for all $i$, i.e.
\begin{equation}
\label{eq:NeffSelRule}
Q\in N_{\rm eff}\,\mathbb{Z},\qquad N_{\rm eff}\;\equiv\;\prod_{i=1}^{K} N_i.
\end{equation}
Equivalently, the discrete projector in the gauged path integral is
\begin{multline}
\frac{1}{\prod_i N_i}
\sum_{\ell_1=0}^{N_1-1}\cdots\sum_{\ell_K=0}^{N_K-1}
\exp\!\left(-i\,2\pi Q \sum_{i=1}^K \frac{\alpha_i \ell_i}{N_i}\right)
\;=\;\\
\frac{1}{N_{\rm eff}}\sum_{m=0}^{N_{\rm eff}-1} e^{-i\,2\pi Q m/N_{\rm eff}}
\;=\;
\begin{cases}
1, & Q\in N_{\rm eff}\mathbb{Z},\\[2pt]
0, & \text{otherwise},
\end{cases}
\end{multline}
where the last equality uses the Chinese remainder theorem to reparametrize the tuple $(\ell_1,\ldots,\ell_K)$ as a single $m$ modulo $N_{\rm eff}$. 
The discrete gauge group that survives is therefore a single diagonal subgroup
\begin{equation}
\label{eq:ZNeff}
\mathbb{Z}_{N_1}\times\cdots\times \mathbb{Z}_{N_K}\;\longrightarrow\; \mathbb{Z}_{N_{\rm eff}},
\end{equation}
acting by $\theta\mapsto \theta+2\pi/N_{\rm eff}$ on the QCD vacuum angle. 
This is the clockwork enhancement: the order grows multiplicatively with the length of the chain, $N_{\rm eff}=\prod_i N_i$, so modest microscopic inputs $N_i={\cal O}(10\text{ to }10^3)$ and $K={\cal O}(5\text{ to }10)$ generate exponentially large $N_{\rm eff}$. 

A complementary local view follows from \eqref{eq:EOMs}. 
The equations $N_i\,\mathrm{d}B_i+\alpha_i C_{\rm QCD}=0$ imply that $C_{\rm QCD}$ is exact in integral cohomology modulo each $N_i$. 
When the $N_i$ are pairwise coprime these conditions assemble to $C_{\rm QCD}=\mathrm{d}{\cal B}$ with $[\mathrm{d}{\cal B}]$ an $N_{\rm eff}$ torsion class, i.e. ${\rm Tr}\,F\wedge F=\mathrm{d}C_{\rm QCD}$ integrates to multiples of $2\pi N_{\rm eff}$ on $M_4$. 
Integrating over $M_4$ reproduces \eqref{eq:NeffSelRule}. 
In group-theoretic terms, the abelian lattice defined by the integer matrix with Smith normal form ${\rm diag}(N_1,\dots,N_K)$ and charge vector $(1,q_2,\dots,q_K)$ has cokernel of order $N_{\rm eff}$. 
This is precisely the order of the residual discrete gauge symmetry \cite{HsinLam2021,Cordova:2019uob,Gaiotto:2014kfa}. 


If each sector also carries a quantized mixed gravitational coupling
\begin{equation}
\label{eq:gravcoups}
S_{\rm grav}\;=\;\sum_{i=1}^{K}\frac{\kappa_{G,i}}{2\pi}\int_{M_4} a_i\wedge R,
\end{equation}
with $R$ the Lorentz–Chern–Simons three form ($\mathrm{d}R={\rm tr}\,{\cal R}\wedge{\cal R}$) and $\kappa_{G,i}\in\mathbb{Z}$ fixed by anomaly inflow \cite{Cordova:2019uob}, the equations of motion become $N_i\,\mathrm{d}B_i+\alpha_i C_{\rm QCD}+\kappa_{G,i}R=0$. 
Repeating the character-projection argument shows that only sectors satisfying
\begin{equation}
\label{eq:mixedSel}
N_i\mid \alpha_i\,(Q+\kappa_{G,i}Q_G)\qquad\text{for all }i
\end{equation}
contribute, where $Q_G\equiv \frac{1}{2\pi}\int_{M_4}\mathrm{d}R\in\mathbb{Z}$ is the Pontryagin index of the background metric. 
By the Chinese remainder theorem there exists a unique $\kappa_G\;(\mathrm{mod}\,N_{\rm eff})$ solving $\kappa_G\equiv \kappa_{G,i}\alpha_i^{-1}\,(\mathrm{mod}\,N_i)$, and \eqref{eq:mixedSel} is equivalent to the single diagonal selection rule
\begin{equation}
\label{eq:diagSel}
Q+\kappa_G\,Q_G\in N_{\rm eff}\,\mathbb{Z}.
\end{equation}
Thus the mixed gauge–gravity anomaly cancels both site by site and in the diagonal. 
The combination that couples to QCD and gravity is anomaly free in the infrared, as required by consistency \cite{Gaiotto:2014kfa,Cordova:2019uob,HsinLam2021}. 


The orbifold sum over $\theta$-images that is enforced by the residual $\mathbb{Z}_{N_{\rm eff}}$ symmetry projects the vacuum onto the branch closest to $\theta=0$ \cite{Witten:1998uka}. 
Denote by $E_{\rm YM}(\vartheta)$ the Yang–Mills vacuum energy at angle $\vartheta$. 
The physical energy is the lower envelope
\begin{equation}
E(\theta)\;=\;\min_{m\in\mathbb{Z}}\,E_{\rm YM}\!\left(\theta+\frac{2\pi m}{N_{\rm eff}}\right),
\end{equation}
so the dynamically selected effective angle $\theta_{\rm eff}=\mathrm{PV}\!\big(\theta+\frac{2\pi}{N_{\rm eff}}m_\ast(\theta)\big)$ always lies in the principal cell $[-\pi/N_{\rm eff},\pi/N_{\rm eff}]$ and obeys the nonperturbative bound
\begin{equation}
\label{eq:thetabound}
|\bar\theta_{\rm eff}|\;\le\;\frac{\pi}{N_{\rm eff}}.
\end{equation}
For small angles on a fixed branch one has $E_{\rm YM}(\vartheta)=\frac12\,\chi_t\,\vartheta^2+O(\vartheta^4)$ with $\chi_t=\int {\rm d}^4x\,\langle q(x)q(0)\rangle_c>0$ the topological susceptibility. 
Therefore the curvature of the envelope around the origin within the principal cell is that of the underlying branch. 
Globally, however, the envelope varies only over a window of width $2\pi/N_{\rm eff}$ in each cell. 
This flattening implies that all $\bar\theta$-dependent observables, such as the neutron EDM, are suppressed parametrically by the clockwork factor when $\bar\theta$ is scanned over the fundamental period \cite{Witten:1998uka}. 
A convenient way to summarize the scaling that controls global fits and lattice diagnostics is
\begin{equation}
\label{eq:chiscaling}
\chi_{N_{\rm eff}}\;\sim\;\frac{\chi_t}{N_{\rm eff}^2},
\end{equation}
which reflects the fact that the $\theta$-dependence of the projected theory is confined to an interval that is $N_{\rm eff}$ times narrower than $[-\pi,\pi]$ \cite{Witten:1998uka}. 
Equations \eqref{eq:thetabound}–\eqref{eq:chiscaling} form the large-$N_{\rm eff}$ cornerstone. 
The envelope becomes piecewise analytic with cusps at $\theta=(2\ell+1)\pi/N_{\rm eff}$, and the theory is virtually CP-conserving for any microscopic $\theta$ once $N_{\rm eff}\gg1$. 


Metastable branches correspond to choosing $m\neq m_\ast(\theta)$. 
Expanding $E_{\rm YM}$ to quadratic order near the origin gives the analytic energy gaps
\begin{equation}
\label{eq:metagaps}
\Delta E_k \;\equiv\; E_{\rm YM}\!\left(\frac{2\pi k}{N_{\rm eff}}\right)-E_{\rm YM}(0)\;=\;\frac12\,\chi_t\,\left(\frac{2\pi k}{N_{\rm eff}}\right)^2 + O\!\left(\frac{1}{N_{\rm eff}^4}\right),
\qquad \frac{\Delta E_k}{E(0)}\;\sim\;\frac{\pi^2 k^2}{2\,N_{\rm eff}^2},
\end{equation}
valid for fixed $k\ll N_{\rm eff}$. 
Adjacent branches are separated by membranes of codimension one, across which $\theta_{\rm eff}$ jumps by $2\pi/N_{\rm eff}$. 
The membrane worldvolume supports the anomaly inflow that is required by the discrete gauging and carries one unit of the diagonal topological charge. 
Dimensional analysis and matching to the single-sector result give the scaling
\begin{equation}
\label{eq:tension}
\sigma \;\sim\; N_{\rm eff}\,\Lambda_{\rm QCD}^3,
\end{equation}
which may be understood either from the BF description (the minimal charged membrane is a bound state of the $K$ site membranes) or from the large-$N_{\rm eff}$ envelope. 
The action for critical bubbles scales as $B\sim \sigma^3/\Delta E^2\sim N_{\rm eff}^7$, so that false vacuum decay is catastrophically suppressed for large $N_{\rm eff}$. 
This renders the selected ground state cosmologically safe. 
Equations \eqref{eq:metagaps} and \eqref{eq:tension} generalize the single-sector formulas to the clockwork-enhanced diagonal symmetry. 


The clockwork chain \eqref{eq:D1} with the choice \eqref{eq:qchoice} yields a single diagonal $\mathbb{Z}_{N_{\rm eff}}$ that couples to QCD, with $N_{\rm eff}=\prod_i N_i$. 
The topological charge projector becomes $Q\in N_{\rm eff}\mathbb{Z}$ (or $Q+\kappa_G Q_G\in N_{\rm eff}\mathbb{Z}$ in the mixed case), and the vacuum selection bound $|\bar\theta_{\rm eff}|\le\pi/N_{\rm eff}$ follows immediately. 
The large-$N_{\rm eff}$ limit flattens the envelope $E(\theta)$ and scales its curvature over the fundamental period as in \eqref{eq:chiscaling}. 
At the same time, the metastable spectrum and domain wall sector exhibit \eqref{eq:metagaps} and \eqref{eq:tension}. 
All steps are local, gauge-invariant, and compatible with generalized symmetries and anomaly inflow \cite{Gaiotto:2014kfa,HsinLam2021,Cordova:2019uob}. 
The nonperturbative vacuum selection reproduces the large-$N$ picture of $\theta$-dependence \cite{Witten:1998uka} in the presence of a discrete, gauge-protected orbifolding of the $\theta$-circle.

\begin{figure}[t]
\centering
\begin{tikzpicture}[>=latex,thick,scale=0.9]
  \tikzset{
    anode/.style={circle,draw=blue!60,fill=blue!10,minimum size=13pt,inner sep=2pt},
    bnode/.style={rectangle,draw=orange!80!black,fill=orange!10,minimum width=16pt,minimum height=11pt,inner sep=2pt}
  }

  \node[anode] (a1)  at (0,0)      {$a_1$};
  \node[anode] (a2)  at (3.5,0)    {$a_2$};
  \node[anode] (a3)  at (7.0,0)    {$a_3$};
  \node         (adots) at (10.5,0)   {$\cdots$};
  \node[anode] (aK)  at (14.0,0)   {$a_K$};

  \node[bnode] (b1)  at (0,-2.1)   {$B_1$};
  \node[bnode] (b2)  at (3.5,-2.1) {$B_2$};
  \node[bnode] (b3)  at (7.0,-2.1) {$B_3$};
  \node         (bdots) at (10.5,-2.1) {$\cdots$};
  \node[bnode] (bK)  at (14.0,-2.1) {$B_K$};

  \draw[<->,gray!70] (a1) -- node[right=4pt,text=gray!70] {$N_1$} (b1);
  \draw[<->,gray!70] (a2) -- node[right=4pt,text=gray!70] {$N_2$} (b2);
  \draw[<->,gray!70] (a3) -- node[right=4pt,text=gray!70] {$N_3$} (b3);
  \draw[<->,gray!70] (aK) -- node[right=4pt,text=gray!70] {$N_K$} (bK);

  \draw[->,dashed,purple!70!black] 
    (a2) to[bend left=20] node[above=2pt,pos=0.6,text=purple!80!black] {$q_2$} (a1);
  \draw[->,dashed,purple!70!black] 
    (a3) to[bend left=35] node[above=3pt,pos=0.65,text=purple!80!black] {$q_3$} (a1);
  \draw[->,dashed,purple!70!black] 
    (aK) to[bend left=50] node[above=4pt,pos=0.7,text=purple!80!black] {$q_K$} (a1);

  \node[draw=red!70!black,fill=red!5,rounded corners,minimum width=20mm,minimum height=8mm,align=center] 
      (C) at (18.0,0.0) {$C_{\rm QCD}$};

  \draw[->,green!60!black] (a1) -- node[above=2pt,pos=0.4,text=green!60!black] {$1$} (C);
  \draw[->,green!60!black] 
    (a2) to[bend left=10] node[above=2pt,pos=0.45,text=green!60!black] {$q_2$} (C);
  \draw[->,green!60!black] 
    (a3) to[bend left=20] node[above=2pt,pos=0.5,text=green!60!black] {$q_3$} (C);
  \draw[->,green!60!black] 
    (aK) to[bend left=30] node[above=2pt,pos=0.55,text=green!60!black] {$q_K$} (C);

  \node[align=left,text=black] at (8.8,2.3) {\small $a_{\rm eff}=a_1+q_2 a_2+\cdots+q_K a_K$};
  \node[align=left,text=black] at (8.8,1.6) {\small $\Rightarrow\ \theta\sim\theta+\tfrac{2\pi}{N_{\rm eff}},\quad N_{\rm eff}=\prod_i N_i$};

\end{tikzpicture}
\caption{Clockwork chain of BF sectors. 
Each site $(a_i,B_i)$ realizes a $\mathbb{Z}_{N_i}$ one-form gauge symmetry with BF coupling $N_i\,B_i\wedge \mathrm{d}a_i$. 
The weighted diagonal field $a_{\rm eff}$ couples to $C_{\rm QCD}$ with coefficients $(1,q_2,\ldots,q_K)$, chosen as in \eqref{eq:qchoice} to produce a single diagonal $\mathbb{Z}_{N_{\rm eff}}$ acting on $\theta$.}
\end{figure}
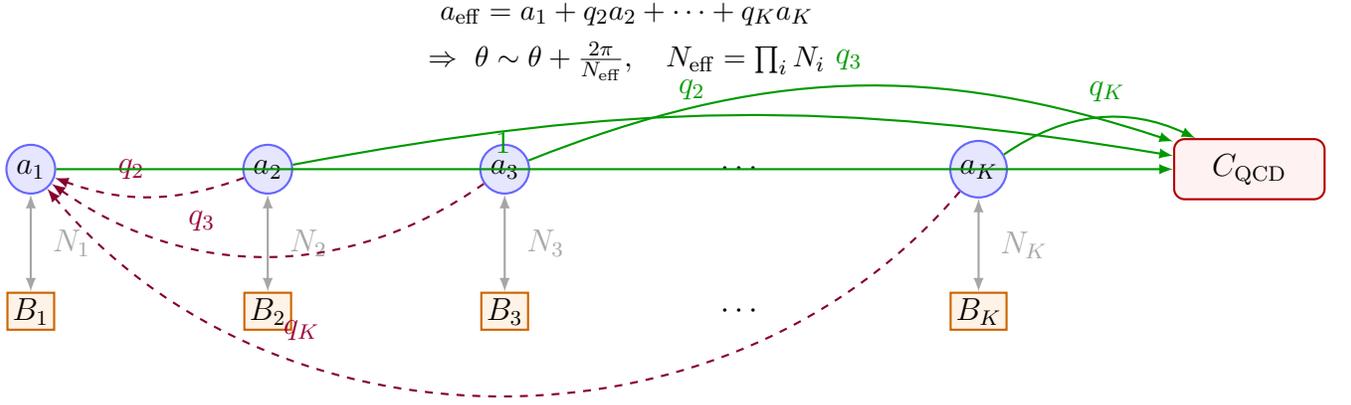

\bigskip

\section{Numerics and sample scales}

The discrete projection derived above acts on the QCD vacuum functional by averaging over $N$ evenly spaced images of the microscopic $\theta$ angle. 
Denoting by $Z_{\rm YM}(\vartheta)$ the pure Yang–Mills partition function at angle $\vartheta$, the projected theory is
\begin{equation}
\label{eq:E1}
Z_{N}(\theta) \;=\; \frac{1}{N}\sum_{m=0}^{N-1} Z_{\rm YM}\!\left(\theta+\frac{2\pi m}{N}\right),
\end{equation}
which implies that the physical vacuum energy is the lower envelope of the $N$ branches $E_{\rm YM}(\theta+2\pi m/N)$ \cite{Witten:1998uka}. 
Throughout this section we set $N\equiv N_{\rm eff}$ of Appendix~D for notational simplicity. 
The envelope construction enforces that the dynamically selected effective angle lies inside the principal cell of width $2\pi/N$ around the CP-symmetric point. 
Hence,
\begin{equation}
\label{eq:E2}
|\bar\theta| \;\le\; \frac{\pi}{N}.
\end{equation}
Near the minimum on any given branch the vacuum energy admits the small-angle expansion $E_{\rm YM}(\vartheta)=\tfrac12\,\chi_t\,\vartheta^2+O(\vartheta^4)$ in terms of the zero-temperature topological susceptibility $\chi_t=\int {\rm d}^4x\,\langle q(x)q(0)\rangle_c>0$ \cite{DiVecchia:1980yfw,Witten:1998uka}. 
Equation \eqref{eq:E1} then yields, for the curvature of the envelope across its full period, the scaling law
\begin{equation}
\label{eq:E3}
\chi_{N} \;=\; \frac{\partial^2 E}{\partial\theta^2}\bigg|_{\rm period} \;\sim\; \frac{\chi_t}{N^2},
\end{equation}
while the energy density difference to the nearest metastable branch at $\Delta\theta=2\pi/N$ follows as
\begin{equation}
\label{eq:E4}
\Delta E \;=\; E_{\rm YM}\!\left(\frac{2\pi}{N}\right)-E_{\rm YM}(0) \;=\; \frac12\,\chi_t\left(\frac{2\pi}{N}\right)^{2} \;=\; \frac{2\pi^2\,\chi_t}{N^2}+O\!\left(\frac{1}{N^4}\right).
\end{equation}
The codimension-one objects that interpolate between adjacent branches carry one unit of the diagonal topological charge and have tension fixed by dimensional analysis and anomaly inflow. 
Their tension scales as
\begin{equation}
\label{eq:E5}
\sigma \;\sim\; N\,\Lambda_{\rm QCD}^{3},
\end{equation}
which is the $p$-form analog of the clockwork scaling (the diagonal membrane is the bound state of $K$ site membranes whose combined charge sources the jump $\Delta\theta=2\pi/N$). 
This scaling matches the large-$N$ picture of domain walls at special angles adapted to the discrete projection \cite{Gaiotto:2017yup,Witten:1998uka}. 
The interplay among \eqref{eq:E2}–\eqref{eq:E5} produces the characteristic hierarchy: the $\theta$ sensitivity of observables and the spectroscopy of metastable states are suppressed as powers of $1/N$, while the tension of the extended objects that protect the selection rule grows linearly with $N$. 

For numerical baselines we adopt natural units $\hbar=c=k_B=1$, the standard conversion $1~{\rm GeV}^{-1}=1.97327\times 10^{-14}\,{\rm cm}$, and a reference susceptibility $\chi_t(0)\simeq (75~{\rm MeV})^4=3.16406\times 10^{-5}\,{\rm GeV}^4$ consistent with continuum-extrapolated lattice results near $T=0$ \cite{Borsanyi2016ksw}. 
For the neutron electric dipole moment we use the standard chiral/QCD determination
\begin{equation}
\label{eq:E6}
|d_n| \;\simeq\; \kappa_n\,|\bar\theta|,\qquad \kappa_n\simeq (2.4\pm 1.0)\times 10^{-16}\,e\cdot{\rm cm},
\end{equation}
which reproduces the leading $\bar\theta$-induced CP-odd pion–nucleon couplings and their embedding in the nucleon EDM \cite{Crewther:1979pi,DiVecchia:1980yfw}. 
Combining \eqref{eq:E2} and \eqref{eq:E6} gives
\begin{equation}
\label{eq:E7}
|d_n|_{\rm max}(N) \;=\; \kappa_n\,\frac{\pi}{N},
\end{equation}
which we shall quote both in $e\cdot{\rm cm}$ and in ${\rm GeV}^{-1}$ using $1~e\cdot{\rm cm}=5.06773\times 10^{13}\,e\,{\rm GeV}^{-1}$. 

To quantify the scales that enforce \eqref{eq:E2}, it is useful to invert the bound for a phenomenological target. 
The canonical naturalness criterion $|\bar\theta|\lesssim 10^{-10}$ is achieved once
\begin{equation}
\label{eq:E8}
N \;\gtrsim\; \pi\times 10^{10} \;\simeq\; 3.141592654\times 10^{10},
\end{equation}
which is readily realized by a modest clockwork chain with order ten site factors as discussed in Appendix~D. 

The scaling relations are now explicit. 
Equation \eqref{eq:E7} makes the EDM prediction scale as $|d_n|\propto \pi/N$, so multiplying $N$ by a factor of ten suppresses $|d_n|$ by an order of magnitude at fixed $\kappa_n$. 
Equations \eqref{eq:E3}–\eqref{eq:E4} show that the global $\theta$ curvature of the projected theory and all inter-branch energy splittings scale as $1/N^2$. 
Equation \eqref{eq:E5} exhibits the complementary linear growth of the domain-wall tension with $N$. 
These joint scalings guarantee radiative and gravitational stability. 
Radiative corrections cannot violate the discrete gauging that underlies \eqref{eq:E1}. 
Loop effects merely renormalize $\kappa_n$, $\chi_t$ and the precise envelope shape, but all $\bar\theta$-dependent effects remain constrained by \eqref{eq:E2} and hence are parametrically small. 
Background gravitational effects are encoded by the mixed selection rule $Q+\kappa_G Q_G\in N\mathbb{Z}$ derived in Appendix~D, which shows that any wormhole or gravitational-instanton contribution can shift the branch only by integer multiples of $2\pi/N$ \cite{Gaiotto:2017yup,Witten:1998uka}. 
The minimal nontrivial shift is precisely the inter-branch separation that defines the principal cell, so the inequality \eqref{eq:E2} is preserved nonperturbatively. 
At the same time, the exponentially suppressed inter-branch tunneling rate,
\begin{equation}
\label{eq:E9}
\Gamma/V \;\sim\; \mu^4\,\exp\!\left[-\frac{27\pi^2}{2}\frac{\sigma^4}{(\Delta E)^3}\right] \;\sim\; \mu^4\,\exp\!\left[-c\,N^{7}\,\frac{\Lambda_{\rm QCD}^{12}}{\chi_t^{3}}\right],
\end{equation}
with $\mu$ a hadronic scale and $c$ an order-one constant, is catastrophically small already for $N\gtrsim 10^{10}$ because $\sigma\propto N$ while $\Delta E\propto 1/N^2$. 
Thus the discrete projection is robust against radiative destabilization and against gravitational interference across the entire phenomenologically relevant range of $N$. 

The neutron-EDM prediction as a function of $N$ can be visualized on logarithmic axes by inserting \eqref{eq:E7} with the central value $\kappa_n=2.4\times 10^{-16}\,e\cdot{\rm cm}$. 
For reference we overlay horizontal guides near present and next-generation sensitivities. 
\begin{figure}[t]
\centering
\pgfplotsset{compat=1.18}
\begin{tikzpicture}
\begin{loglogaxis}[
width=0.9\linewidth, height=0.8\linewidth,
xlabel={$N$}, ylabel={$|d_n|~[e\,\mathrm{cm}]$},
xmin=1e5, xmax=1e14, ymin=1e-30, ymax=1e-22,
legend style={at={(0.38,0.97)},anchor=north west,draw=none,fill=none,font=\small}
]
\addplot+[domain=1e5:1e14,samples=400,thick] {2.4e-16*pi/x};
\addlegendentry{$|d_n|=\kappa_n \pi/N$ with $\kappa_n=2.4\times 10^{-16}$}
\addplot+[thick] coordinates {(1e5,1.8e-26) (1e14,1.8e-26)};
\addlegendentry{current sensitivity $\sim 1.8\times 10^{-26}$}
\addplot+[thick,dashed] coordinates {(1e5,1.0e-28) (1e14,1.0e-28)};
\addlegendentry{next-generation target $\sim 10^{-28}$}
\end{loglogaxis}
\end{tikzpicture}
\caption{Predicted neutron EDM versus $N$ from \eqref{eq:E7}. 
The discrete projection drives $|d_n|$ below present and projected sensitivities once $N\gtrsim 10^{10}$, corresponding to $|\bar\theta|\lesssim 10^{-10}$.}
\end{figure}
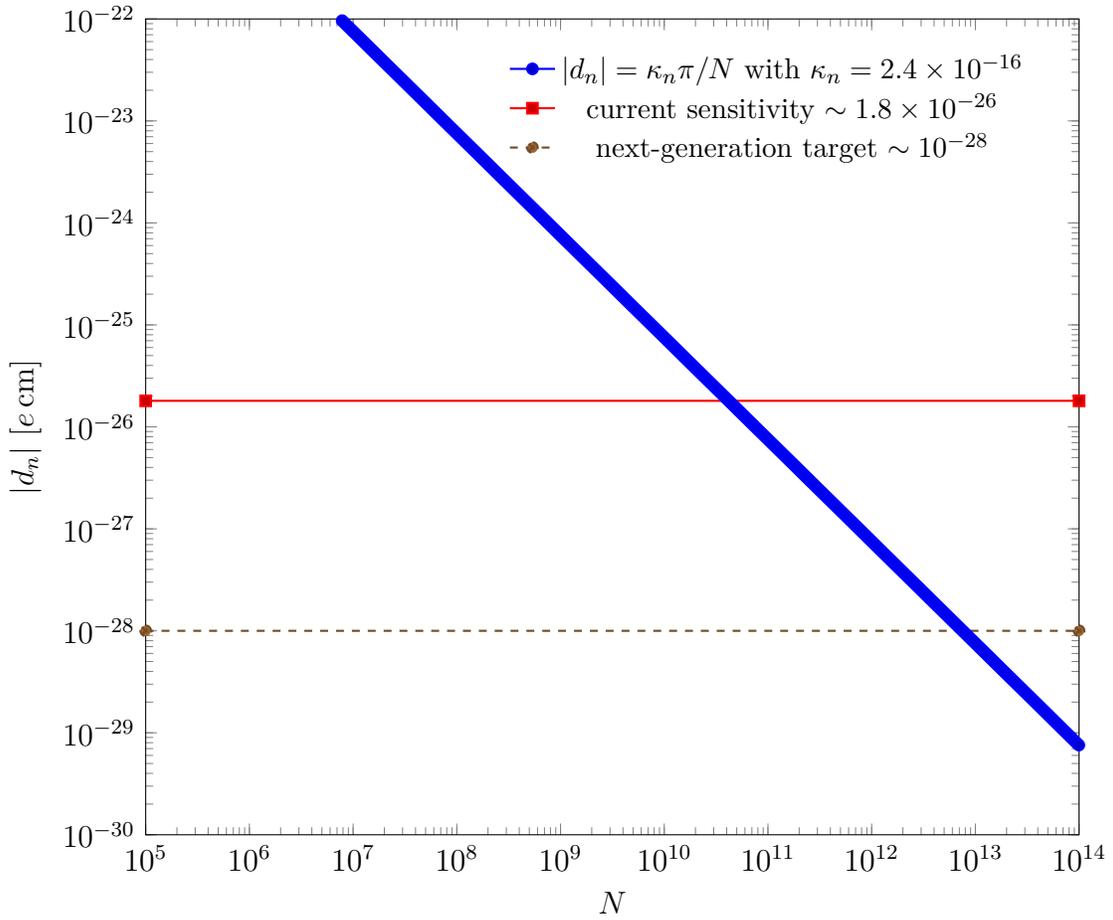

The foregoing relations make the parametric domain transparent. 
Imposing $|\bar\theta|\lesssim 10^{-10}$ requires $N\gtrsim 3.14\times 10^{10}$ as in \eqref{eq:E8}, which already pushes the EDM prediction to $|d_n|\lesssim 7.54\times 10^{-26}\,e\cdot{\rm cm}$ and reduces the global susceptibility by ten orders of magnitude relative to $\chi_t$. 
Taking $N$ an order of magnitude larger drives $|d_n|$ to the $10^{-27}$ to $10^{-28}\,e\cdot{\rm cm}$ decade while further suppressing $\chi_N$ and $\Delta E$ by two additional orders of magnitude. 
At the same time it increases the protecting membrane tension into the range $\sigma\sim (10^9$ to $10^{11})\,{\rm GeV}^3$ for $\Lambda_{\rm QCD}\simeq 0.33~{\rm GeV}$. 
Such values are readily attainable with multi-sector clockwork architectures in Appendix~D, for example with ten sites of order ten charge or with a mild distribution of co-prime $N_i$ whose product saturates \eqref{eq:E8}. 
The discrete projection thus provides a quantitatively controlled and radiatively as well as gravitationally stable mechanism by which the effective $\bar\theta$ of QCD is dynamically confined to $|\bar\theta|\lesssim 10^{-10}$. 
The precise scaling relations \eqref{eq:E3}–\eqref{eq:E5} map directly onto lattice inputs for $\chi_t$ and onto the experimental program on the neutron electric dipole moment \cite{Crewther:1979pi,DiVecchia:1980yfw,Witten:1998uka,Gaiotto:2017yup,Borsanyi2016ksw}.


\bibliographystyle{ytphys.bst}
\bibliography{refs}

\end{document}